\RequirePackage{etoolbox}

\newbool{ieee}
\boolfalse{ieee} % FALSE

\ifieee

\documentclass[twocolumn, twoside]{IEEEtran}
\usepackage{cite, array}
\usepackage{amsmath, amsfonts, amssymb, enumitem, aliascnt}
\renewcommand{\IEEEQED}{\IEEEQEDopen}
\newenvironment{proof}[1][Proof]{\begin{IEEEproof}[#1]}{\end{IEEEproof}}
\newenvironment{claimproof}[1][Proof of claim]{\begin{IEEEproof}[#1]\renewcommand*{\IEEEQED}{\(\boxdot\)}}{\end{IEEEproof}}

\else

\documentclass[11pt]{article}
\usepackage{amsmath, amsfonts, amssymb, amsthm, enumitem, fullpage}
\newenvironment{claimproof}[1][Proof of claim]{\begin{proof}[#1]}{\end{proof}}

\fi

\usepackage{mathtools, bm, hyperref, microtype}
\usepackage[noabbrev, capitalise]{cleveref}

\newtheorem{theorem}{Theorem}

\ifieee
\newaliascnt{lemma}{theorem}
\newtheorem{lemma}[lemma]{Lemma}
\aliascntresetthe{lemma}
\crefname{lemma}{Lemma}{Lemmas}

\newaliascnt{corollary}{theorem}
\newtheorem{corollary}[corollary]{Corollary}
\aliascntresetthe{corollary}
\crefname{lemma}{Corollary}{Corollaries}

\newaliascnt{proposition}{theorem}
\newtheorem{proposition}[proposition]{Proposition}
\aliascntresetthe{proposition}
\crefname{proposition}{Proposition}{Propositions}
\else
\newtheorem{lemma}[theorem]{Lemma}
\newtheorem{corollary}[theorem]{Corollary}
\newtheorem{proposition}[theorem]{Proposition}
\fi

\newtheorem{conjecture}{Conjecture}

\newtheorem{claim}{Claim}
\ifieee
\newtheorem{claimm}{Claim}

\else
\newtheorem*{claimm}{Claim}
\fi

\DeclarePairedDelimiter\abs{\lvert}{\rvert}%
\DeclarePairedDelimiter\ceil{\lceil}{\rceil}%
\newcommand{\aaa}{\bar{a}}
\newcommand{\bb}{\bar{b}}
\newcommand{\aaai}{\bar{a}_i}
\newcommand{\bbi}{\bar{b}_i}
\newcommand{\cab}{c_{a,b}}
\newcommand{\cabi}{c_{a_i,b_i}}
\newcommand{\cinterval}{[-\min(a_i\bbi, \aaai b_i), \min(a_ib_i, \aaai\bbi)]}
\newcommand{\sabn}{(a_1, \dots, a_n, b_1, \dots, b_n)}
\newcommand{\san}{a_1, \dots, a_n}
\newcommand{\sbn}{b_1, \allowbreak \dots, b_n}
\newcommand{\spn}{(p_1, \dots, p_n)}
\newcommand{\spabn}{(p_1, \dots, \allowbreak p_n, \allowbreak a_1, \dots, a_n, \allowbreak b_1, \dots, b_n)}
\newcommand{\ph}{\varphi}
\newcommand{\phab}{\ph_{a,b}}
\newcommand{\eps}{\varepsilon}
\newcommand{\la}{\lambda}
\newcommand{\oooo}{\{(0,0),(1,1)\}}
\newcommand{\sumn}{\sum_{i=1}^n}
\newcommand{\sumnpabc}{\sumn p_i(a_i\bbi + \aaai b_i + 2c_i)}
\newcommand{\sumnpabcx}{\sumn p_i(a_i\bbi + \aaai b_i + 2c_{a_i,b_i}(x))}

\newcommand{\sumnpaabb}{\sumn p_i\left(H_2(a_i, \aaai) + H_2(b_i, \bbi)\right)}
\newcommand{\sumnphabcxs}{\sumn p_iS_4(a_i,b_i,c_{a_i,b_i}(x^*))}
\newcommand{\fpp}[2]{\frac{\partial #1}{\partial #2}}
\newcommand{\Of}{\mathcal{O}_f}
\newcommand{\pr}{\mathrm{Pr}}
\newcommand{\set}[1]{\left\{#1 \right\}}
\newcommand{\dset}[2]{\left\{#1 \colon #2\right\}}
\newcommand{\wtdn}{\widetilde{D}_n}
\newcommand{\dd}{\mathrm{d}}

\usepackage{tikz, pgfplots}
\pgfplotsset{compat=1.17}
\pgfplotsset{general/.style={ thick, axis x line=middle, axis y line=none, ticks=none, width=5cm, height=6cm, ymin=-1.2, ymax=1 }}
\definecolor{red}{rgb}{.859375,.265625,.21484375}
\definecolor{blue}{rgb}{.26171875,.51953125,.95703125}
\definecolor{darkgray}{RGB}{64,64,64}
\ifieee
\tikzstyle{block}=[draw, rectangle, minimum height=1cm, text width=2cm, text centered, draw=darkgray, font=\large]
\tikzstyle{block_medium}=[draw, rectangle, minimum height=1.5cm, text width=2cm, text centered, draw=darkgray, font=\large]
\tikzstyle{circ}=[draw, circle, radius=2cm, text centered, draw=darkgray, font=\LARGE]
\else
\tikzstyle{block}=[draw, rectangle, minimum height=1cm, text width=2cm, text centered, draw=darkgray, font=\small]
\tikzstyle{block_medium}=[draw, rectangle, minimum height=1.5cm, text width=2cm, text centered, draw=darkgray, font=\small]
\tikzstyle{circ}=[draw, circle, radius=2cm, text centered, draw=darkgray, font=\large]
\fi
\tikzstyle{line} = [draw, -latex]

\title{On the binary adder channel with complete feedback,\\with an application to quantitative group testing}
\ifieee
\author{Samuel~H.~Florin,
        Matthew~H.~Ho,
        and Zilin~Jiang%
\thanks{Manuscript received Xxxxx xx, xxxx; revised Xxxxx xx, xxxx; accepted Xxxxx xx, xxxx. Z.~Jiang was supported in part by an AMS Simons Travel Grant, and by U.S. taxpayers through NSF grant DMS-1953946.}
\thanks{S.~H.~Florin was with Greenwich High School, Greenwich, CT 06830. He is now with the Department of Mathematics, Massachusetts Institute of Technology, Cambridge, MA 02139 (email: \href{mailto:sflorin@mit.edu}{sflorin@mit.edu}).}
\thanks{M.~H.~Ho was with Palo Alto High School, Palo Alto, CA 94301. He is now with the the Department of Mathematics, Massachusetts Institute of Technology, Cambridge, MA 02139 (email: \href{mailto:mattho@mit.edu}{mattho@mit.edu}).}
\thanks{Z.~Jiang was with Massachusetts Insitute of Technology. He is now with the School of Mathematical and Statistical Sciences, and the School of Computing and Augmented Intelligence, Arizona State University, Tempe, AZ 85281 (email: \href{mailto:zilinj@asu.edu}{zilinj@asu.edu}).}
\thanks{Copyright (c) 2017 IEEE. Personal use of this material is permitted.  However, permission to use this material for any other purposes must be obtained from the IEEE by sending a request to \href{mailto:pubs-permissions@ieee.org}{pubs-permissions@ieee.org}.}}
\else
\author{
  Samuel H.~Florin\thanks{Department of Mathematics, Massachusetts Institute of Technology, Cambridge, MA 02139.}${\ }^{,}$\footnote{Email: {\tt sflorin@mit.edu}. The work was done when S.~H.~Florin was a student at Greenwich High School.}
  \and Matthew H.~Ho\footnotemark[1]${\ }^{,}$\footnote{Email: {\tt mattho@mit.edu}. The work was done when M.~H.~Ho was a student at Palo Alto High School.}
  \and Zilin Jiang\thanks{School of Mathematical and Statistical Sciences, and School of Computing and Augmented Intelligence, Arizona State University, Tempe, AZ 85281. Email: {\tt zilinj@asu.edu}. The work was done when Z.~Jiang was an Applied Mathematics Instructor at Massachusetts Institute of Technology.}
}
\fi
\date{}

\begin{document}

\ifieee
\markboth{IEEE Transactions on Information Theory,~Vol.~xx, No.~x, Xxxx~xxxx}{Florin \MakeLowercase{\textit{et al.}}: On the binary adder channel with complete feedback}
\fi

\maketitle

\begin{abstract}
  We determine the exact value of the optimal symmetric rate point $(r, r)$ in the Dueck zero-error capacity region of the binary adder channel with complete feedback. We proved that the average zero-error capacity $r = h(1/2-\delta) \approx 0.78974$, where $h(\cdot)$ is the binary entropy function and $\delta = 1/(2\log_2(2+\sqrt3))$. Our motivation is a problem in quantitative group testing. Given a set of $n$ elements two of which are defective, the quantitative group testing problem asks for the identification of these two defectives through a series of tests. Each test gives the number of defectives contained in the tested subset, and the outcomes of previous tests are assumed known at the time of designing the current test. We establish that the minimum number of tests is asymptotic to $(\log_2 n) / r$ as $n \to \infty$.
\end{abstract}

\ifieee
\fi

\section{Introduction}

\ifieee \IEEEPARstart{T}{he} \else The \fi
(two-user) binary adder channel is a discrete memoryless multiple-access channel. Let the message sets specified for the senders be of size $M_1, M_2$, and let $w_1 \in [M_1]$, and $w_2 \in [M_2]$ be two messages chosen by the two senders beforehand. During the $k$th use of the channel, two functions $e_{1k}$ and $e_{2k}$ respectively encode $w_1$ and $w_2$ to two channel inputs $x_{1k}, x_{2k} \in \set{0,1}$. The binary adder channel then takes $x_{1k}, x_{2k}$ and outputs $y_k = x_{1k} + x_{2k} \in \set{0,1,2}$. The sequence of outputs $y(w_1, w_2) := (y_k)_{k=1}^n$ is decoded by the receiver. We say that a channel is the \emph{binary adder channel with complete feedback} when the encoders know all the previous outputs of the channel, namely every $e_{1k}$ and $e_{2k}$ depend not only on $w_1$ and $w_2$ respectively but also on $(y_i)_{i=1}^{k-1}$ (see \cref{fig:bac} for a schematic diagram).

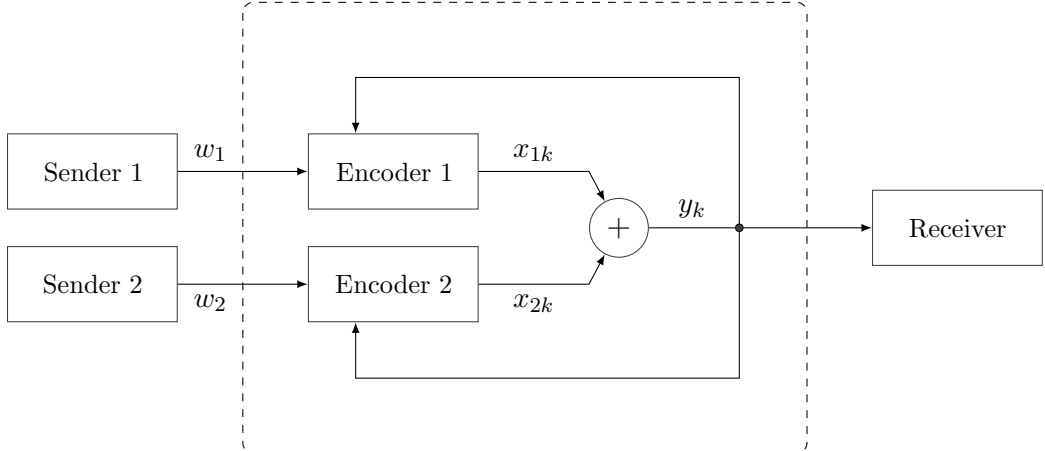
\begin{figure}[t]
  \centering
  \ifieee
  \begin{tikzpicture}[scale=0.64, every node/.style={scale=0.64}]
  \else
  \begin{tikzpicture}
  \fi
    \node[circ] (c) at (0,0) {+};
    \node[block] (d) at (4.5,0) {Receiver};
    \node[block] (e1) at (-3,0.75) {Encoder 1};
    \node[block] (e2) at (-3,-0.75) {Encoder 2};
    \node[block] (u1) at (-7,0.75) {Sender 1};
    \node[block] (u2) at (-7,-0.75) {Sender 2};
    \node[draw, circle, minimum size=1mm, inner sep=0pt, outer sep=0pt, fill=darkgray] (dot) at (1.6,0) {};
    \path[line] (u1) -- node[near start, above] {$w_1$} (e1);
    \path[line] (u2) -- node[near start, below] {$w_2$} (e2);
    \path[line] (e1.east) -- node[above] {$x_{1k}$} ([yshift=0.75cm]c.west) -- (c);
    \path[line] (e2.east) -- node[below] {$x_{2k}$} ([yshift=-0.75cm]c.west) -- (c);
    \path[line] (c) -- node[above] {$y_k$} (dot) -- (d);
    \path[line] (dot) -- (1.6,2) -- (-3.5,2) -- ([xshift=-0.5cm]e1.north);
    \path[line] (dot) -- (1.6,-2) -- (-3.5,-2) -- ([xshift=-0.5cm]e2.south);
    \path[draw, dashed, rounded corners] (-5,3) -- (-5,-3) -- (2.5,-3) -- (2.5,3) -- cycle;
  \end{tikzpicture}
  \caption{Two-user binary adder channel with complete feedback.}
  \label{fig:bac}
\end{figure}

An $(M_1, M_2, n)$ \emph{uniquely decodable code} for the binary adder channel with complete feedback consists of a collection encoding functions such that the output sequences satisfy that
\ifieee
\begin{multline*}
  y(w_1, w_2) = y(w_1', w_2') \Leftrightarrow (w_1, w_2) = (w_1', w_2'), \\
  \text{for all }(w_1, w_2), (w_1', w_2') \in [M_1] \times [M_2].
\end{multline*}
\else
\[
  y(w_1, w_2) = y(w_1', w_2') \Leftrightarrow (w_1, w_2) = (w_1', w_2'), \text{ for all }(w_1, w_2), (w_1', w_2') \in [M_1] \times [M_2].
\]
\fi
The \emph{zero-error capacity region} captures the rates at which the information can be transmitted over the channel without error:
\ifieee
\begin{multline*}
  \Of := \text{closure of }\{(R_1, R_2) \colon \exists n_0\,\forall n \ge n_0\,\\
  \exists (\ceil{2^{nR_1}}, \ceil{2^{nR_2}}, n) \text{ uniquely decodable code}\}.
\end{multline*}
\else
\[
  \Of := \text{closure of }\dset{(R_1, R_2)}{\exists n_0\,\forall n \ge n_0\,\exists (\ceil{2^{nR_1}}, \ceil{2^{nR_2}}, n) \text{ uniquely decodable code}}.
\]
\fi
The \emph{average zero-error capacity} is defined by
\[
  R(\Of) := \sup\dset{\tfrac{1}{2}(R_1 + R_2)}{(R_1, R_2) \in \Of}.
\]
Because $\Of$ is convex and symmetric with respect to the line $R_1 = R_2$, the average zero-error capacity can also be defined as
\[
  R(\Of) := \sup\dset{R}{(R,R)\in\Of}.
\]
The point $(R(\Of), R(\Of))$ is known as the equal-rate point or the symmetric rate point in the existing literature. We refer the readers to \cite[Chapter 4]{Ahl18} for a broader view on multiple-access channels.

Dueck characterized the zero-error capacity region for a class of discrete memoryless multiple-access channels. For simplicity, we state his characterization specialized for the binary adder channel with complete feedback.

\begin{theorem}[Dueck~\cite{Due85}] \label{thm:dueck-rv}
  The rate pair $(R_1, R_2)$ belongs to the zero-error capacity region $\Of$ of the binary adder channel with complete feedback if and only if there exist two Bernoulli random variables $X_1, X_2$ and an auxiliary discrete random variable $U$ such that
  \begin{subequations}
    \ifieee
    \begin{gather}
      R_1 \le H(X_1 \mid U), \quad R_2 \le H(X_2 \mid U), \\
      \begin{split} \label{eqn:dueck-rv}
        H(X_1', X_2' \mid U, X_1' + X_2') \le I(U; X_1' + X_2') \\
        \text{for every }(X_1', X_2') \in \mathcal{P}(X_1, X_2, U),
      \end{split}
    \end{gather}
    \else
    \begin{gather} 
      R_1 \le H(X_1 \mid U), \quad R_2 \le H(X_2 \mid U), \\
      H(X_1', X_2' \mid U, X_1' + X_2') \le I(U; X_1' + X_2') 
      \text{ for every }(X_1', X_2') \in \mathcal{P}(X_1, X_2, U), \label{eqn:dueck-rv}
    \end{gather}
    \fi
  \end{subequations}
  where $H(\cdot \mid \cdot)$ is the conditional entropy, $I(\cdot;\cdot)$ is the mutual information, and $\mathcal{P}(X_1, X_2, U)$ consists of pairs $(X_1', X_2')$ of Bernoulli random variables such that their conditional probability distributions satisfy $p_{X_1'\mid U} = p_{X_1\mid U}$, $p_{X_2'\mid U} = p_{X_2\mid U}$.%
  \footnote{We point out that Dueck \cite[Section~2]{Due85} asserted that, through the linear dependency approach of Ahlswede and K\"orner~\cite{AK75} (see also \cite{S78}), the auxiliary random variable $U$ in \cref{thm:dueck-rv} can be assumed to take no more than $6$ values. However, as far as we are aware, it is unclear how the linear dependency approach could be applied directly to Dueck's characterization. The major obstacle to the linear dependency approach is that \eqref{eqn:dueck-rv} a priori represents infinitely many linear constraints on the probability distribution of $U$ --- one for each $(X_1', X_2')$ in $\mathcal{P}(X_1, X_2, U)$.}
\end{theorem}

Pinning down the precise value of $R(\Of)$ has remained as an open problem since 1985. Zhang, Berger and Massey \cite{ZBM87} wrote ``However, numerical evaluation of his [Dueck's] capacity region description is fraught with challenging obstacles even in this special case [the binary adder channel with complete feedback]'', which is reiterated by Ahlswede \cite[Section 4.9.4]{Ahl18}. Before our work, the best upper bound $R(\Of) \le 0.79113$ can be deduced from the average Cover--Leung channel capacity of the binary adder channel due to Willems~\cite{Wil84}, whereas the lower bound $R(\Of) \ge 0.78974$, due to Belokopytov~\cite{Bel89}, is conjectured to be tight in \cite[Conjecture~A]{JPV19}.

We make significant progress on the computation of $R(\Of)$, and we settle Conjecture~A in \cite{JPV19} in the affirmative. Here, as well as throughout the paper,
\[
  H_n(x_1, \dots, x_n) := -(x_1\log x_1 + \dots + x_n\log x_n),
\]
with the convention $0 \log 0 = 0$, and all the logarithms are in base $2$.

\begin{theorem} \label{thm:exact}
  The average zero-error capacity $R(\Of)$ of the binary adder channel with complete feedback is $H_2(1/2-\delta,1/2+\delta)\approx 0.78974$, where $\delta = 1/(2 \log(2+\sqrt3))$.
\end{theorem}

It is worth pointing out that the exact value $R(\mathcal{O})$ of the average zero-error capacity of the binary adder channel \emph{without} feedback still remains unknown. The current lower bound is $R(\mathcal{O}) \ge (\log 240)/12 \approx 0.65891$ due to Mattas and \"Osterg{\aa}rd~\cite[Section III]{MO05}, while the current upper bound is $R(\mathcal{O}) \le 3/4 = 0.75$.

Our motivation to determine the exact value of $R(\Of)$ comes from quantitative group testing. In a typical quantitative group testing problem, we are given a set of $n$ elements, some of which are defective, and we wish to identify the defectives by testing subsets of these $n$ elements. The classical additive model assumes that each test could precisely tell the number of defectives contained in the tested subset. This model also goes under the name ``coin weighing problem with a spring scale''.

In this paper we focus on the worst-case analysis in the adaptive setting, in which the outcomes of previous tests are assumed known at the time of designing the current test. By first testing the entire collection of $n$ elements, we may assume that the number $d$ of defectives is known. The $d = 1$ case is a classical puzzle, and it is known that $\ceil{\log n}$ tests are needed to identify the defective.

However even the $d = 2$ case was far from being fully understood. Denote $t(n)$ the minimum number of tests required to identify the two defectives among $n$ elements with certainty. The ``Fibonaccian algorithm'' by Christen~\cite{Chr80} and Aigner~\cite{Aig86} gives $t(n) \le \log n / \log \ph + O(1) \approx 1.44042\log n$, where $\ph = (\sqrt{5}+1)/2$ is the golden ratio. The algorithm was improved to $t(n) \le (12/\log 330)\log n + O(1) \approx 1.43432\log n$ by Hao~\cite[Section~4]{Hao90}. Using the language of decision trees, Gargano et al. \cite{GMSV92} improved the upper bound to $t(n) \le (7/5)\log n + O(1) \approx 1.4\log n$. Christen reported in \cite[Section~6]{Chr94} that a more involved recursive method yields a better upper bound $t(n) \le (6 / \log 20)\log n + O(1) \approx 1.38827\log n$. On the flip side, the information-theoretic bound gives $t(n) \ge \log_3\binom{n}{2} \approx 1.26185\log n$, and no better lower bound is known.\footnote{In \cite[Section~6]{Chr94}, Christen asserted that the lower bound $(4/3)\log n$ of Lindstr\"om~\cite{Lin69} for the predetermined setting also holds for $t(n)$ in the adaptive setting. We believe this assertion is incorrect as it is inconsistent with \cref{cor:tn}.}

Hao \cite{Hao90} observed that there exists $r > 0$ such that $t(n) \sim (\log n)/r$ (equivalently, $t(n)/\log n$ approaches to $1/r$ as $n \to \infty$). To characterize the constant $r$, we take advantage of the following correspondence between the quantitative group testing problem and the binary adder channel.

\begin{theorem} \label{thm:correspondence}
  The minimum number $t(n)$ of tests needed in an adaptive strategy to identify two defectives among $n$ elements satisfies $t(n) \sim (\log n) / R(\Of)$, where the constant $R(\Of)$ is the average zero-error capacity of the binary adder channel with complete feedback.
\end{theorem}

To the best of our knowledge, although \cref{thm:correspondence} is probably widely known among the information theory community, it was only recently mentioned in \cite{JPV19}. As an immediate consequence of \cref{thm:exact,thm:correspondence}, we establish the asymptotic formula of $t(n)$.

\begin{corollary} \label{cor:tn}
  The minimum number $t(n)$ of tests needed in a adaptive strategy to identify two defectives among $n$ elements satisfies $t(n) \sim (\log n)/H_2(1/2-\delta,1/2+\delta) \approx 1.26624\log n$, where $\delta := 1/(2\log (2 + \sqrt3))$.
  \ifieee
  \leavevmode\unskip\penalty9999\hbox{}\nobreak\hfill\quad\hbox{\IEEEQED}
  \else
  \qed
  \fi
\end{corollary}

As is common in group testing theory, one distinguishes between the predetermined and adaptive testing strategies. For the predetermined setting, in which all the tests are decided in advance, determining the minimum number $\tilde{t}(n)$ of tests to identify two defectives among $n$ elements is equivalent to determining the maximum size of a $B_2$-sequence of $n$-dimensional binary vectors. The results on binary $B_2$-sequences in \cite{Lin69} and \cite{CLZ01} are best known, and these results imply that $(1.73837-o(1))\log n \le \tilde{t}(n) \le (2+o(1))\log n$. Finding the asymptotic behavior of $\tilde{t}(n)$ remains as an open problem.

\section{Proof ideas}

In this section we outline the three major steps that lead to the proof of \cref{thm:exact}. The first step is a more down-to-earth reformulation of Dueck's characterization, that is \cref{thm:dueck-rv}, for the symmetric rate point. To state the reformulation, we adopt the following notations for the entirety of the paper:
\begin{gather*}
  \aaai := 1 - a_i, \quad \bbi := 1 - b_i, \quad \aaa = 1 - a, \quad \bb = 1 - b,\\
  S_3(x) := H_3((1-x)/2, x, (1-x)/2), \\
  S_4(a,b,c) := H_4(ab-c, a\bb+c, \aaa b+c, \aaa\bb-c), \\
  \Delta^{n-1} := \dset{\spn \in [0,1]^n}{p_1 + \dots + p_n = 1}.
\end{gather*}

\begin{theorem} \label{thm:dueck-concrete}
  The rate pair $(R, R)$ belongs to the zero-error capacity region $\Of$ of the binary adder channel with complete feedback if and only if there exist $n \in \mathbb{N}$, $\spn \in \Delta^{n-1}$, $\san \in [0,1]$ and $\sbn \in [0,1]$ such that
  \begin{subequations}
    \ifieee
    \begin{gather}
      R \le \tfrac12 \sumnpaabb, \label{eqn:r-bound} \\
      \begin{split} \label{eqn:h3-h4}
        S_3{\left(\sumnpabc\right)} \ge \sumn p_i S_4(a_i, b_i, c_i)\\
        \text{for every }c_1 \in C_1, \dots, c_n \in C_n,
      \end{split}
    \end{gather}
    \else
    \begin{gather}
      R \le \tfrac12 \sumnpaabb, \label{eqn:r-bound} \\
      S_3{\left(\sumnpabc\right)} \ge \sumn p_i S_4(a_i, b_i, c_i), \text{ for every }c_1 \in C_1, \dots, c_n \in C_n, \label{eqn:h3-h4}
    \end{gather}
    \fi
  \end{subequations}
  where $C_i := \cinterval$.
\end{theorem}

\begin{figure}[t]
  \centering
  \begin{tikzpicture}
    \draw[scale=.5, red, thick] (1.4,4.5942738)--(1.42,4.4128884)--(1.44,4.2770064)--(1.46,4.1641156)--(1.48,4.0658275)--(1.5,3.9779066)--(1.52,3.8978565)--(1.54,3.8240572)--(1.56,3.7553856)--(1.58,3.6910248)--(1.6,3.6303576)--(1.62,3.5729044)--(1.64,3.5182833)--(1.66,3.4661839)--(1.68,3.4163502)--(1.7,3.3685679)--(1.72,3.3226549)--(1.74,3.2784557)--(1.76,3.2358355)--(1.78,3.1946769)--(1.8,3.1548766)--(1.82,3.1163434)--(1.84,3.0789962)--(1.86,3.0427621)--(1.88,3.0075761)--(1.9,2.9733789)--(1.92,2.9401173)--(1.94,2.9077424)--(1.96,2.8762097)--(1.98,2.8454786)--(2.0,2.8155114)--(2.02,2.7862735)--(2.04,2.757733)--(2.06,2.7298603)--(2.08,2.7026278)--(2.1,2.6760098)--(2.12,2.6499825)--(2.14,2.6245234)--(2.16,2.5996117)--(2.18,2.5752276)--(2.2,2.5513527)--(2.22,2.5279695)--(2.24,2.5050618)--(2.26,2.482614)--(2.28,2.4606116)--(2.3,2.4390406)--(2.32,2.4178879)--(2.34,2.3971412)--(2.36,2.3767885)--(2.38,2.3568187)--(2.4,2.3372211)--(2.42,2.3179854)--(2.44,2.299102)--(2.46,2.2805616)--(2.48,2.2623554)--(2.5,2.2444748)--(2.52,2.2269117)--(2.54,2.2096585)--(2.56,2.1927077)--(2.58,2.1760521)--(2.6,2.1596849)--(2.62,2.1435996)--(2.64,2.1277899)--(2.66,2.1122496)--(2.68,2.096973)--(2.7,2.0819544)--(2.72,2.0671885)--(2.74,2.05267)--(2.76,2.0383939)--(2.78,2.0243553)--(2.8,2.0105497)--(2.82,1.9969725)--(2.84,1.9836193)--(2.86,1.9704859)--(2.88,1.9575683)--(2.9,1.9448626)--(2.92,1.9323649)--(2.94,1.9200715)--(2.96,1.9079789)--(2.98,1.8960837)--(3.0,1.8843825)--(3.02,1.872872)--(3.04,1.8615492)--(3.06,1.850411)--(3.08,1.8394544)--(3.1,1.8286765)--(3.12,1.8180747)--(3.14,1.8076462)--(3.16,1.7973884)--(3.18,1.7872987)--(3.2,1.7773747)--(3.22,1.767614)--(3.24,1.7580143)--(3.26,1.7485732)--(3.28,1.7392886)--(3.3,1.7301584)--(3.32,1.7211805)--(3.34,1.7123527)--(3.36,1.7036733)--(3.38,1.6951402)--(3.4,1.6867517)--(3.42,1.6785058)--(3.44,1.6704008)--(3.46,1.6624351)--(3.48,1.6546068)--(3.5,1.6469145)--(3.52,1.6393565)--(3.54,1.6319312)--(3.56,1.6246373)--(3.58,1.6174731)--(3.6,1.6104373)--(3.62,1.6035285)--(3.64,1.5967453)--(3.66,1.5900865)--(3.68,1.5835506)--(3.7,1.5771366)--(3.72,1.5708432)--(3.74,1.5646691)--(3.76,1.5586132)--(3.78,1.5526745)--(3.8,1.5468517)--(3.82,1.5411439)--(3.84,1.53555)--(3.86,1.530069)--(3.88,1.5246999)--(3.9,1.5194417)--(3.92,1.5142936)--(3.94,1.5092545)--(3.96,1.5043236)--(3.98,1.4995001)--(4.0,1.4947831)--(4.02,1.4901718)--(4.04,1.4856655)--(4.06,1.4812632)--(4.08,1.4769644)--(4.1,1.4727682)--(4.12,1.4686739)--(4.14,1.4646809)--(4.16,1.4607886)--(4.18,1.4569961)--(4.2,1.453303)--(4.22,1.4497086)--(4.24,1.4462124)--(4.26,1.4428137)--(4.28,1.439512)--(4.3,1.4363067)--(4.32,1.4331974)--(4.34,1.4301835)--(4.36,1.4272646)--(4.38,1.4244402)--(4.4,1.4217099)--(4.42,1.4190731)--(4.44,1.4165295)--(4.46,1.4140788)--(4.48,1.4117204)--(4.5,1.409454)--(4.52,1.4072793)--(4.54,1.4051959)--(4.56,1.4032036)--(4.58,1.4013019)--(4.6,1.3994906)--(4.62,1.3977695)--(4.64,1.3961382)--(4.66,1.3945965)--(4.68,1.3931441)--(4.7,1.3917809)--(4.72,1.3905066)--(4.74,1.389321)--(4.76,1.388224)--(4.78,1.3872154)--(4.8,1.386295)--(4.82,1.3854626)--(4.84,1.3847182)--(4.86,1.3840617)--(4.88,1.3834929)--(4.9,1.3830118)--(4.92,1.3826182)--(4.94,1.3823122)--(4.96,1.3820936)--(4.98,1.3819625)--(5.0,1.3819187)--(5.02,1.3819625)--(5.04,1.3820936)--(5.06,1.3823122)--(5.08,1.3826182)--(5.1,1.3830118)--(5.12,1.3834929)--(5.14,1.3840617)--(5.16,1.3847182)--(5.18,1.3854626)--(5.2,1.386295)--(5.22,1.3872154)--(5.24,1.388224)--(5.26,1.389321)--(5.28,1.3905066)--(5.3,1.3917809)--(5.32,1.3931441)--(5.34,1.3945965)--(5.36,1.3961382)--(5.38,1.3977695)--(5.4,1.3994906)--(5.42,1.4013019)--(5.44,1.4032036)--(5.46,1.4051959)--(5.48,1.4072793)--(5.5,1.409454)--(5.52,1.4117204)--(5.54,1.4140788)--(5.56,1.4165295)--(5.58,1.4190731)--(5.6,1.4217099)--(5.62,1.4244402)--(5.64,1.4272646)--(5.66,1.4301835)--(5.68,1.4331974)--(5.7,1.4363067)--(5.72,1.439512)--(5.74,1.4428137)--(5.76,1.4462124)--(5.78,1.4497086)--(5.8,1.453303)--(5.82,1.4569961)--(5.84,1.4607886)--(5.86,1.4646809)--(5.88,1.4686739)--(5.9,1.4727682)--(5.92,1.4769644)--(5.94,1.4812632)--(5.96,1.4856655)--(5.98,1.4901718)--(6.0,1.4947831)--(6.02,1.4995001)--(6.04,1.5043236)--(6.06,1.5092545)--(6.08,1.5142936)--(6.1,1.5194417)--(6.12,1.5246999)--(6.14,1.530069)--(6.16,1.53555)--(6.18,1.5411439)--(6.2,1.5468517)--(6.22,1.5526745)--(6.24,1.5586132)--(6.26,1.5646691)--(6.28,1.5708432)--(6.3,1.5771366)--(6.32,1.5835506)--(6.34,1.5900865)--(6.36,1.5967453)--(6.38,1.6035285)--(6.4,1.6104373)--(6.42,1.6174731)--(6.44,1.6246373)--(6.46,1.6319312)--(6.48,1.6393565)--(6.5,1.6469145)--(6.52,1.6546068)--(6.54,1.6624351)--(6.56,1.6704008)--(6.58,1.6785058)--(6.6,1.6867517)--(6.62,1.6951402)--(6.64,1.7036733)--(6.66,1.7123527)--(6.68,1.7211805)--(6.7,1.7301584)--(6.72,1.7392886)--(6.74,1.7485732)--(6.76,1.7580143)--(6.78,1.767614)--(6.8,1.7773747)--(6.82,1.7872987)--(6.84,1.7973884)--(6.86,1.8076462)--(6.88,1.8180747)--(6.9,1.8286765)--(6.92,1.8394544)--(6.94,1.850411)--(6.96,1.8615492)--(6.98,1.872872)--(7.0,1.8843825)--(7.02,1.8960837)--(7.04,1.9079789)--(7.06,1.9200715)--(7.08,1.9323649)--(7.1,1.9448626)--(7.12,1.9575683)--(7.14,1.9704859)--(7.16,1.9836193)--(7.18,1.9969725)--(7.2,2.0105497)--(7.22,2.0243553)--(7.24,2.0383939)--(7.26,2.05267)--(7.28,2.0671885)--(7.3,2.0819544)--(7.32,2.096973)--(7.34,2.1122496)--(7.36,2.1277899)--(7.38,2.1435996)--(7.4,2.1596849)--(7.42,2.1760521)--(7.44,2.1927077)--(7.46,2.2096585)--(7.48,2.2269117)--(7.5,2.2444748)--(7.52,2.2623554)--(7.54,2.2805616)--(7.56,2.299102)--(7.58,2.3179854)--(7.6,2.3372211)--(7.62,2.3568187)--(7.64,2.3767885)--(7.66,2.3971412)--(7.68,2.4178879)--(7.7,2.4390406)--(7.72,2.4606116)--(7.74,2.482614)--(7.76,2.5050618)--(7.78,2.5279695)--(7.8,2.5513527)--(7.82,2.5752276)--(7.84,2.5996117)--(7.86,2.6245234)--(7.88,2.6499825)--(7.9,2.6760098)--(7.92,2.7026278)--(7.94,2.7298603)--(7.96,2.757733)--(7.98,2.7862735)--(8.0,2.8155114)--(8.02,2.8454786)--(8.04,2.8762097)--(8.06,2.9077424)--(8.08,2.9401173)--(8.1,2.9733789)--(8.12,3.0075761)--(8.14,3.0427621)--(8.16,3.0789962)--(8.18,3.1163434)--(8.2,3.1548766)--(8.22,3.1946769)--(8.24,3.2358355)--(8.26,3.2784557)--(8.28,3.3226549)--(8.3,3.3685679)--(8.32,3.4163502)--(8.34,3.4661839)--(8.36,3.5182833)--(8.38,3.5729044)--(8.4,3.6303576)--(8.42,3.6910248)--(8.44,3.7553856)--(8.46,3.8240572)--(8.48,3.8978565)--(8.5,3.9779066)--(8.52,4.0658275)--(8.54,4.1641156)--(8.56,4.2770064)--(8.58,4.4128884)--(8.6,4.5942738)--(8.6,5.4057262)--(8.58,5.5871116)--(8.56,5.7229936)--(8.54,5.8358844)--(8.52,5.9341725)--(8.5,6.0220934)--(8.48,6.1021435)--(8.46,6.1759428)--(8.44,6.2446144)--(8.42,6.3089752)--(8.4,6.3696424)--(8.38,6.4270956)--(8.36,6.4817167)--(8.34,6.5338161)--(8.32,6.5836498)--(8.3,6.6314321)--(8.28,6.6773451)--(8.26,6.7215443)--(8.24,6.7641645)--(8.22,6.8053231)--(8.2,6.8451234)--(8.18,6.8836566)--(8.16,6.9210038)--(8.14,6.9572379)--(8.12,6.9924239)--(8.1,7.0266211)--(8.08,7.0598827)--(8.06,7.0922576)--(8.04,7.1237903)--(8.02,7.1545214)--(8.0,7.1844886)--(7.98,7.2137265)--(7.96,7.242267)--(7.94,7.2701397)--(7.92,7.2973722)--(7.9,7.3239902)--(7.88,7.3500175)--(7.86,7.3754766)--(7.84,7.4003883)--(7.82,7.4247724)--(7.8,7.4486473)--(7.78,7.4720305)--(7.76,7.4949382)--(7.74,7.517386)--(7.72,7.5393884)--(7.7,7.5609594)--(7.68,7.5821121)--(7.66,7.6028588)--(7.64,7.6232115)--(7.62,7.6431813)--(7.6,7.6627789)--(7.58,7.6820146)--(7.56,7.700898)--(7.54,7.7194384)--(7.52,7.7376446)--(7.5,7.7555252)--(7.48,7.7730883)--(7.46,7.7903415)--(7.44,7.8072923)--(7.42,7.8239479)--(7.4,7.8403151)--(7.38,7.8564004)--(7.36,7.8722101)--(7.34,7.8877504)--(7.32,7.903027)--(7.3,7.9180456)--(7.28,7.9328115)--(7.26,7.94733)--(7.24,7.9616061)--(7.22,7.9756447)--(7.2,7.9894503)--(7.18,8.0030275)--(7.16,8.0163807)--(7.14,8.0295141)--(7.12,8.0424317)--(7.1,8.0551374)--(7.08,8.0676351)--(7.06,8.0799285)--(7.04,8.0920211)--(7.02,8.1039163)--(7.0,8.1156175)--(6.98,8.127128)--(6.96,8.1384508)--(6.94,8.149589)--(6.92,8.1605456)--(6.9,8.1713235)--(6.88,8.1819253)--(6.86,8.1923538)--(6.84,8.2026116)--(6.82,8.2127013)--(6.8,8.2226253)--(6.78,8.232386)--(6.76,8.2419857)--(6.74,8.2514268)--(6.72,8.2607114)--(6.7,8.2698416)--(6.68,8.2788195)--(6.66,8.2876473)--(6.64,8.2963267)--(6.62,8.3048598)--(6.6,8.3132483)--(6.58,8.3214942)--(6.56,8.3295992)--(6.54,8.3375649)--(6.52,8.3453932)--(6.5,8.3530855)--(6.48,8.3606435)--(6.46,8.3680688)--(6.44,8.3753627)--(6.42,8.3825269)--(6.4,8.3895627)--(6.38,8.3964715)--(6.36,8.4032547)--(6.34,8.4099135)--(6.32,8.4164494)--(6.3,8.4228634)--(6.28,8.4291568)--(6.26,8.4353309)--(6.24,8.4413868)--(6.22,8.4473255)--(6.2,8.4531483)--(6.18,8.4588561)--(6.16,8.46445)--(6.14,8.469931)--(6.12,8.4753001)--(6.1,8.4805583)--(6.08,8.4857064)--(6.06,8.4907455)--(6.04,8.4956764)--(6.02,8.5004999)--(6.0,8.5052169)--(5.98,8.5098282)--(5.96,8.5143345)--(5.94,8.5187368)--(5.92,8.5230356)--(5.9,8.5272318)--(5.88,8.5313261)--(5.86,8.5353191)--(5.84,8.5392114)--(5.82,8.5430039)--(5.8,8.546697)--(5.78,8.5502914)--(5.76,8.5537876)--(5.74,8.5571863)--(5.72,8.560488)--(5.7,8.5636933)--(5.68,8.5668026)--(5.66,8.5698165)--(5.64,8.5727354)--(5.62,8.5755598)--(5.6,8.5782901)--(5.58,8.5809269)--(5.56,8.5834705)--(5.54,8.5859212)--(5.52,8.5882796)--(5.5,8.590546)--(5.48,8.5927207)--(5.46,8.5948041)--(5.44,8.5967964)--(5.42,8.5986981)--(5.4,8.6005094)--(5.38,8.6022305)--(5.36,8.6038618)--(5.34,8.6054035)--(5.32,8.6068559)--(5.3,8.6082191)--(5.28,8.6094934)--(5.26,8.610679)--(5.24,8.611776)--(5.22,8.6127846)--(5.2,8.613705)--(5.18,8.6145374)--(5.16,8.6152818)--(5.14,8.6159383)--(5.12,8.6165071)--(5.1,8.6169882)--(5.08,8.6173818)--(5.06,8.6176878)--(5.04,8.6179064)--(5.02,8.6180375)--(5.0,8.6180813)--(4.98,8.6180375)--(4.96,8.6179064)--(4.94,8.6176878)--(4.92,8.6173818)--(4.9,8.6169882)--(4.88,8.6165071)--(4.86,8.6159383)--(4.84,8.6152818)--(4.82,8.6145374)--(4.8,8.613705)--(4.78,8.6127846)--(4.76,8.611776)--(4.74,8.610679)--(4.72,8.6094934)--(4.7,8.6082191)--(4.68,8.6068559)--(4.66,8.6054035)--(4.64,8.6038618)--(4.62,8.6022305)--(4.6,8.6005094)--(4.58,8.5986981)--(4.56,8.5967964)--(4.54,8.5948041)--(4.52,8.5927207)--(4.5,8.590546)--(4.48,8.5882796)--(4.46,8.5859212)--(4.44,8.5834705)--(4.42,8.5809269)--(4.4,8.5782901)--(4.38,8.5755598)--(4.36,8.5727354)--(4.34,8.5698165)--(4.32,8.5668026)--(4.3,8.5636933)--(4.28,8.560488)--(4.26,8.5571863)--(4.24,8.5537876)--(4.22,8.5502914)--(4.2,8.546697)--(4.18,8.5430039)--(4.16,8.5392114)--(4.14,8.5353191)--(4.12,8.5313261)--(4.1,8.5272318)--(4.08,8.5230356)--(4.06,8.5187368)--(4.04,8.5143345)--(4.02,8.5098282)--(4.0,8.5052169)--(3.98,8.5004999)--(3.96,8.4956764)--(3.94,8.4907455)--(3.92,8.4857064)--(3.9,8.4805583)--(3.88,8.4753001)--(3.86,8.469931)--(3.84,8.46445)--(3.82,8.4588561)--(3.8,8.4531483)--(3.78,8.4473255)--(3.76,8.4413868)--(3.74,8.4353309)--(3.72,8.4291568)--(3.7,8.4228634)--(3.68,8.4164494)--(3.66,8.4099135)--(3.64,8.4032547)--(3.62,8.3964715)--(3.6,8.3895627)--(3.58,8.3825269)--(3.56,8.3753627)--(3.54,8.3680688)--(3.52,8.3606435)--(3.5,8.3530855)--(3.48,8.3453932)--(3.46,8.3375649)--(3.44,8.3295992)--(3.42,8.3214942)--(3.4,8.3132483)--(3.38,8.3048598)--(3.36,8.2963267)--(3.34,8.2876473)--(3.32,8.2788195)--(3.3,8.2698416)--(3.28,8.2607114)--(3.26,8.2514268)--(3.24,8.2419857)--(3.22,8.232386)--(3.2,8.2226253)--(3.18,8.2127013)--(3.16,8.2026116)--(3.14,8.1923538)--(3.12,8.1819253)--(3.1,8.1713235)--(3.08,8.1605456)--(3.06,8.149589)--(3.04,8.1384508)--(3.02,8.127128)--(3.0,8.1156175)--(2.98,8.1039163)--(2.96,8.0920211)--(2.94,8.0799285)--(2.92,8.0676351)--(2.9,8.0551374)--(2.88,8.0424317)--(2.86,8.0295141)--(2.84,8.0163807)--(2.82,8.0030275)--(2.8,7.9894503)--(2.78,7.9756447)--(2.76,7.9616061)--(2.74,7.94733)--(2.72,7.9328115)--(2.7,7.9180456)--(2.68,7.903027)--(2.66,7.8877504)--(2.64,7.8722101)--(2.62,7.8564004)--(2.6,7.8403151)--(2.58,7.8239479)--(2.56,7.8072923)--(2.54,7.7903415)--(2.52,7.7730883)--(2.5,7.7555252)--(2.48,7.7376446)--(2.46,7.7194384)--(2.44,7.700898)--(2.42,7.6820146)--(2.4,7.6627789)--(2.38,7.6431813)--(2.36,7.6232115)--(2.34,7.6028588)--(2.32,7.5821121)--(2.3,7.5609594)--(2.28,7.5393884)--(2.26,7.517386)--(2.24,7.4949382)--(2.22,7.4720305)--(2.2,7.4486473)--(2.18,7.4247724)--(2.16,7.4003883)--(2.14,7.3754766)--(2.12,7.3500175)--(2.1,7.3239902)--(2.08,7.2973722)--(2.06,7.2701397)--(2.04,7.242267)--(2.02,7.2137265)--(2.0,7.1844886)--(1.98,7.1545214)--(1.96,7.1237903)--(1.94,7.0922576)--(1.92,7.0598827)--(1.9,7.0266211)--(1.88,6.9924239)--(1.86,6.9572379)--(1.84,6.9210038)--(1.82,6.8836566)--(1.8,6.8451234)--(1.78,6.8053231)--(1.76,6.7641645)--(1.74,6.7215443)--(1.72,6.6773451)--(1.7,6.6314321)--(1.68,6.5836498)--(1.66,6.5338161)--(1.64,6.4817167)--(1.62,6.4270956)--(1.6,6.3696424)--(1.58,6.3089752)--(1.56,6.2446144)--(1.54,6.1759428)--(1.52,6.1021435)--(1.5,6.0220934)--(1.48,5.9341725)--(1.46,5.8358844)--(1.44,5.7229936)--(1.42,5.5871116)--(1.4,5.4057262)--cycle;
    \foreach \angle in {0, 180}
      \fill[scale=.5, blue, opacity=0.5, rotate around={\angle:(5,5)}] (0,0)--(0,10)--(0.02,9.76)--(0.04,9.58)--(0.06,9.4)--(0.08,9.24)--(0.1,9.1)--(0.12,8.94)--(0.14,8.8)--(0.16,8.66)--(0.18,8.54)--(0.2,8.4)--(0.22,8.28)--(0.24,8.16)--(0.26,8.04)--(0.28,7.92)--(0.3,7.82)--(0.32,7.7)--(0.34,7.6)--(0.36,7.48)--(0.38,7.38)--(0.4,7.28)--(0.42,7.18)--(0.44,7.08)--(0.46,6.98)--(0.48,6.9)--(0.5,6.8)--(0.52,6.7)--(0.54,6.62)--(0.56,6.52)--(0.58,6.44)--(0.6,6.36)--(0.62,6.28)--(0.64,6.18)--(0.66,6.1)--(0.68,6.02)--(0.7,5.94)--(0.72,5.88)--(0.74,5.8)--(0.76,5.72)--(0.78,5.64)--(0.8,5.58)--(0.82,5.5)--(0.84,5.42)--(0.86,5.36)--(0.88,5.3)--(0.9,5.22)--(0.92,5.16)--(0.94,5.1)--(0.96,5.02)--(0.98,4.96)--(1.0,4.9)--(1.02,4.84)--(1.04,4.78)--(1.06,4.72)--(1.08,4.66)--(1.1,4.6)--(1.12,4.56)--(1.14,4.5)--(1.16,4.44)--(1.18,4.38)--(1.2,4.34)--(1.22,4.28)--(1.24,4.22)--(1.26,4.18)--(1.28,4.12)--(1.3,4.08)--(1.32,4.04)--(1.34,3.98)--(1.36,3.94)--(1.38,3.9)--(1.4,3.84)--(1.42,3.8)--(1.44,3.76)--(1.46,3.72)--(1.48,3.68)--(1.5,3.62)--(1.52,3.58)--(1.54,3.54)--(1.56,3.5)--(1.58,3.48)--(1.6,3.44)--(1.62,3.4)--(1.64,3.36)--(1.66,3.32)--(1.68,3.28)--(1.7,3.24)--(1.72,3.22)--(1.74,3.18)--(1.76,3.14)--(1.78,3.12)--(1.8,3.08)--(1.82,3.06)--(1.84,3.02)--(1.86,2.98)--(1.88,2.96)--(1.9,2.92)--(1.92,2.9)--(1.94,2.88)--(1.96,2.84)--(1.98,2.82)--(2.0,2.78)--(2.02,2.76)--(2.04,2.74)--(2.06,2.7)--(2.08,2.68)--(2.1,2.66)--(2.12,2.64)--(2.14,2.6)--(2.16,2.58)--(2.18,2.56)--(2.2,2.54)--(2.22,2.52)--(2.24,2.5)--(2.26,2.48)--(2.28,2.44)--(2.3,2.42)--(2.32,2.4)--(2.34,2.38)--(2.36,2.36)--(2.38,2.34)--(2.4,2.32)--(2.42,2.3)--(2.44,2.28)--(2.46,2.26)--(2.48,2.26)--(2.5,2.24)--(2.52,2.22)--(2.54,2.2)--(2.56,2.18)--(2.58,2.16)--(2.6,2.14)--(2.62,2.12)--(2.64,2.12)--(2.66,2.1)--(2.68,2.08)--(2.7,2.06)--(2.72,2.04)--(2.74,2.04)--(2.76,2.02)--(2.78,2.0)--(2.8,1.98)--(2.82,1.98)--(2.84,1.96)--(2.86,1.94)--(2.88,1.94)--(2.9,1.92)--(2.92,1.9)--(2.94,1.88)--(2.96,1.88)--(2.98,1.86)--(3.0,1.84)--(3.02,1.84)--(3.04,1.82)--(3.06,1.82)--(3.08,1.8)--(3.1,1.78)--(3.12,1.78)--(3.14,1.76)--(3.16,1.74)--(3.18,1.74)--(3.2,1.72)--(3.22,1.72)--(3.24,1.7)--(3.26,1.68)--(3.28,1.68)--(3.3,1.66)--(3.32,1.66)--(3.34,1.64)--(3.36,1.64)--(3.38,1.62)--(3.4,1.62)--(3.42,1.6)--(3.44,1.6)--(3.46,1.58)--(3.48,1.58)--(3.5,1.56)--(3.52,1.54)--(3.54,1.54)--(3.56,1.52)--(3.58,1.52)--(3.6,1.5)--(3.62,1.5)--(3.64,1.48)--(3.66,1.48)--(3.68,1.48)--(3.7,1.46)--(3.72,1.46)--(3.74,1.44)--(3.76,1.44)--(3.78,1.42)--(3.8,1.42)--(3.82,1.4)--(3.84,1.4)--(3.86,1.38)--(3.88,1.38)--(3.9,1.38)--(3.92,1.36)--(3.94,1.36)--(3.96,1.34)--(3.98,1.34)--(4.0,1.32)--(4.02,1.32)--(4.04,1.32)--(4.06,1.3)--(4.08,1.3)--(4.1,1.28)--(4.12,1.28)--(4.14,1.26)--(4.16,1.26)--(4.18,1.26)--(4.2,1.24)--(4.22,1.24)--(4.24,1.22)--(4.26,1.22)--(4.28,1.22)--(4.3,1.2)--(4.32,1.2)--(4.34,1.2)--(4.36,1.18)--(4.38,1.18)--(4.4,1.16)--(4.42,1.16)--(4.44,1.16)--(4.46,1.14)--(4.48,1.14)--(4.5,1.14)--(4.52,1.12)--(4.54,1.12)--(4.56,1.12)--(4.58,1.1)--(4.6,1.1)--(4.62,1.08)--(4.64,1.08)--(4.66,1.08)--(4.68,1.06)--(4.7,1.06)--(4.72,1.06)--(4.74,1.04)--(4.76,1.04)--(4.78,1.04)--(4.8,1.02)--(4.82,1.02)--(4.84,1.02)--(4.86,1.0)--(4.88,1.0)--(4.9,1.0)--(4.92,0.98)--(4.94,0.98)--(4.96,0.98)--(4.98,0.96)--(5.0,0.96)--(5.02,0.96)--(5.04,0.94)--(5.06,0.94)--(5.08,0.94)--(5.1,0.94)--(5.12,0.92)--(5.14,0.92)--(5.16,0.92)--(5.18,0.9)--(5.2,0.9)--(5.22,0.9)--(5.24,0.88)--(5.26,0.88)--(5.28,0.88)--(5.3,0.88)--(5.32,0.86)--(5.34,0.86)--(5.36,0.86)--(5.38,0.84)--(5.4,0.84)--(5.42,0.84)--(5.44,0.82)--(5.46,0.82)--(5.48,0.82)--(5.5,0.82)--(5.52,0.8)--(5.54,0.8)--(5.56,0.8)--(5.58,0.8)--(5.6,0.78)--(5.62,0.78)--(5.64,0.78)--(5.66,0.76)--(5.68,0.76)--(5.7,0.76)--(5.72,0.76)--(5.74,0.74)--(5.76,0.74)--(5.78,0.74)--(5.8,0.74)--(5.82,0.72)--(5.84,0.72)--(5.86,0.72)--(5.88,0.72)--(5.9,0.7)--(5.92,0.7)--(5.94,0.7)--(5.96,0.68)--(5.98,0.68)--(6.0,0.68)--(6.02,0.68)--(6.04,0.66)--(6.06,0.66)--(6.08,0.66)--(6.1,0.66)--(6.12,0.64)--(6.14,0.64)--(6.16,0.64)--(6.18,0.64)--(6.2,0.62)--(6.22,0.62)--(6.24,0.62)--(6.26,0.62)--(6.28,0.62)--(6.3,0.6)--(6.32,0.6)--(6.34,0.6)--(6.36,0.6)--(6.38,0.58)--(6.4,0.58)--(6.42,0.58)--(6.44,0.58)--(6.46,0.56)--(6.48,0.56)--(6.5,0.56)--(6.52,0.56)--(6.54,0.54)--(6.56,0.54)--(6.58,0.54)--(6.6,0.54)--(6.62,0.54)--(6.64,0.52)--(6.66,0.52)--(6.68,0.52)--(6.7,0.52)--(6.72,0.5)--(6.74,0.5)--(6.76,0.5)--(6.78,0.5)--(6.8,0.5)--(6.82,0.48)--(6.84,0.48)--(6.86,0.48)--(6.88,0.48)--(6.9,0.48)--(6.92,0.46)--(6.94,0.46)--(6.96,0.46)--(6.98,0.46)--(7.0,0.44)--(7.02,0.44)--(7.04,0.44)--(7.06,0.44)--(7.08,0.44)--(7.1,0.42)--(7.12,0.42)--(7.14,0.42)--(7.16,0.42)--(7.18,0.42)--(7.2,0.4)--(7.22,0.4)--(7.24,0.4)--(7.26,0.4)--(7.28,0.4)--(7.3,0.38)--(7.32,0.38)--(7.34,0.38)--(7.36,0.38)--(7.38,0.38)--(7.4,0.36)--(7.42,0.36)--(7.44,0.36)--(7.46,0.36)--(7.48,0.36)--(7.5,0.34)--(7.52,0.34)--(7.54,0.34)--(7.56,0.34)--(7.58,0.34)--(7.6,0.34)--(7.62,0.32)--(7.64,0.32)--(7.66,0.32)--(7.68,0.32)--(7.7,0.32)--(7.72,0.3)--(7.74,0.3)--(7.76,0.3)--(7.78,0.3)--(7.8,0.3)--(7.82,0.3)--(7.84,0.28)--(7.86,0.28)--(7.88,0.28)--(7.9,0.28)--(7.92,0.28)--(7.94,0.26)--(7.96,0.26)--(7.98,0.26)--(8.0,0.26)--(8.02,0.26)--(8.04,0.26)--(8.06,0.24)--(8.08,0.24)--(8.1,0.24)--(8.12,0.24)--(8.14,0.24)--(8.16,0.24)--(8.18,0.22)--(8.2,0.22)--(8.22,0.22)--(8.24,0.22)--(8.26,0.22)--(8.28,0.22)--(8.3,0.2)--(8.32,0.2)--(8.34,0.2)--(8.36,0.2)--(8.38,0.2)--(8.4,0.2)--(8.42,0.18)--(8.44,0.18)--(8.46,0.18)--(8.48,0.18)--(8.5,0.18)--(8.52,0.18)--(8.54,0.18)--(8.56,0.16)--(8.58,0.16)--(8.6,0.16)--(8.62,0.16)--(8.64,0.16)--(8.66,0.16)--(8.68,0.14)--(8.7,0.14)--(8.72,0.14)--(8.74,0.14)--(8.76,0.14)--(8.78,0.14)--(8.8,0.14)--(8.82,0.12)--(8.84,0.12)--(8.86,0.12)--(8.88,0.12)--(8.9,0.12)--(8.92,0.12)--(8.94,0.12)--(8.96,0.1)--(8.98,0.1)--(9.0,0.1)--(9.02,0.1)--(9.04,0.1)--(9.06,0.1)--(9.08,0.1)--(9.1,0.1)--(9.12,0.08)--(9.14,0.08)--(9.16,0.08)--(9.18,0.08)--(9.2,0.08)--(9.22,0.08)--(9.24,0.08)--(9.26,0.06)--(9.28,0.06)--(9.3,0.06)--(9.32,0.06)--(9.34,0.06)--(9.36,0.06)--(9.38,0.06)--(9.4,0.06)--(9.42,0.04)--(9.44,0.04)--(9.46,0.04)--(9.48,0.04)--(9.5,0.04)--(9.52,0.04)--(9.54,0.04)--(9.56,0.04)--(9.58,0.04)--(9.6,0.02)--(9.62,0.02)--(9.64,0.02)--(9.66,0.02)--(9.68,0.02)--(9.7,0.02)--(9.72,0.02)--(9.74,0.02)--(9.76,0.02)--(9.78,0.0)--(9.8,0.0)--(9.82,0.0)--(9.84,0.0)--(9.86,0.0)--(9.88,0.0)--(9.9,0.0)--(9.92,0.0)--(9.94,0.0)--(9.96,0.0)--(9.98,0.0)--(10.0,0.0)--cycle;
      \draw[darkgray, dashed, thick] (0,0)--(5,5);
      \draw[darkgray, line] (0,0)--(5.5,0);
      \draw[darkgray, line] (0,0)--(0,5.5);
      \draw[blue] (5,0)--(5,5)--(0,5);
      \node at (5.5,0) [right] {$a_1$};
      \node at (0,5.5) [above] {$b_1$};
  \end{tikzpicture}
  \caption{The blue shaded region is the feasible region and the red closed curve is a level curve of the objective function.}
  \label{fig:feasible}
\end{figure}
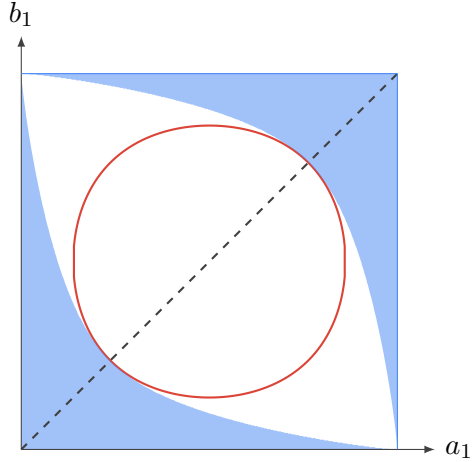

Even in the simplest case where $n = 1$, maximizing the right hand side of \eqref{eqn:r-bound} in the feasible region described by \eqref{eqn:h3-h4} is a nonlinear non-convex optimization problem. \cref{fig:feasible} suggests that the optimum occurs when $a_1 = b_1$. Indeed, Belokopytov obtained his lower bound on $R(\Of)$ by setting $n = 1$ and $a_1 = b_1$ in \cref{thm:dueck-concrete}.

\begin{corollary}[Theorem~1 of Belokopytov~\cite{Bel89}] \label{lem:belokopytov}
  The rate pair $(R, R)$, where $R = H_2(1/2-\delta,1/2+\delta)$ and $\delta = 1/(2\log (2 + \sqrt3))$, belongs to the zero-error capacity region $\Of$ of the binary adder channel with complete feedback.
\end{corollary}

Although \eqref{eqn:h3-h4} a priori represents infinitely many constraints, the second step essentially reduces it to a single constraint by eliminating the universal quantifier in \eqref{eqn:h3-h4}. To state this single constraint, we need the following technical lemma about convex combinations of certain functions.

\begin{lemma} \label{lem:unique-solution}
  For every $\spn \in \Delta^{n-1}$, $\san$, $\sbn$, and $x \in [0,1]$, for every $i \in [n]$ the equation
  \ifieee
  \begin{gather*}
    (a_ib_i - c)(\aaai\bbi - c)(2x)^2 = (a_i\bbi + c)(\aaai b_i + c)(1-x)^2, \\
    -\min(a_i\bbi, \aaai b_i) \le c \le \min(a_ib_i, \aaai\bbi)
  \end{gather*}
  \else
  \[
    (a_ib_i - c)(\aaai\bbi - c)(2x)^2 = (a_i\bbi + c)(\aaai b_i + c)(1-x)^2, \quad -\min(a_i\bbi, \aaai b_i) \le c \le \min(a_ib_i, \aaai\bbi)
  \]
  \fi
  has a unique solution $c = \cabi(x)$, and the equation
  \[
    x = \sumnpabcx, \quad x \in (0,1]
  \]
  has a unique solution $x = x^*$ if
  \ifieee
  $(p_1, \dots, p_n, a_1, \dots, a_n, \allowbreak b_1, \dots, b_n) \in D_n$,
  \else
  $\spabn \in D_n$,
  \fi
  where
  \ifieee
  \begin{multline*}
    D_n = \biggl\{\spabn \colon \\
    \sumn p_i\abs{a_i - b_i} > 0 \text{ or }\sumn p_i\sqrt{a_i\aaai} > 1/4\biggr\}.
  \end{multline*}
  \else
  \[
    D_n = \biggl\{\spabn \colon \sumn p_i\abs{a_i - b_i} > 0 \text{ or }\sumn p_i\sqrt{a_i\aaai} > 1/4\biggr\}.
  \]
  \fi
\end{lemma}

\begin{theorem} \label{thm:main}
  The average zero-error capacity $R(\Of)$ of the binary adder channel with complete feedback equals the optimum of the following optimization problem.
  \begin{align*}
    \text{Maximize: } & \tfrac12 \sumnpaabb, \\
    \text{subject to: } & S_3(x^*) \ge \sumnphabcxs, \\
    & n \in \mathbb{N}, \spabn \in D_n,
  \end{align*}
  where $\cabi(x)$, $x^*$, and $D_n$ are defined as in \cref{lem:unique-solution}.%
  \footnote{Using a linear dependency argument, $n$ can be assumed to be no more than $3$ in \cref{thm:main}.}
\end{theorem}

The third step is to make use of the Lagrangian function of the optimization problem in \cref{thm:main} with a fixed Lagrange multiplier to get an upper bound on $R(\Of)$ that matches Belokopytov's lower bound in \cref{lem:belokopytov}.

\begin{theorem} \label{thm:lag}
  For every $\la \in (0, 1/2)$, define $L\colon D_n \to \mathbb{R}$ by
  \ifieee
  \begin{multline*}
    L\spabn := \\
    \tfrac12 \sumn p_i\left(H_2(a_i, \aaai) + H_2(b_i, \bbi)\right) \\
    + \la\left(S_3(x^*) - \sumnphabcxs\right),
  \end{multline*}
  \else
  \begin{multline*}
    L\spabn := \tfrac12 \sumn p_i\left(H_2(a_i, \aaai) + H_2(b_i, \bbi)\right) \\ + \la\left(S_3(x^*) - \sumnphabcxs\right),
  \end{multline*}
  \fi
  where $\cabi(x)$, $x^*$ and $D_n$ are defined as in \cref{lem:unique-solution}. If the supremum of $L$ is greater than $(1+\la)/2$, then $L$ has a global maximum point, and its maximum value is at most
  \ifieee
  \begin{multline*}
    \max\Bigl(1 - \la \log(1+x_1), H_2{\left(\tfrac{r}{4} - \tfrac{\sqrt3}{4}x_1, \tfrac{1}{4r} + \tfrac{\sqrt3}{4}x_1\right)} \\
    + \la\left(-1 + \tfrac{\sqrt3}{2} (\log r) (1-x_1)\right)\Bigr),
  \end{multline*}
  \else
  \[
    \max\left(1 - \la \log(1+x_1), H_2{\left(\tfrac{r}{4} - \tfrac{\sqrt3}{4}x_1, \tfrac{1}{4r} + \tfrac{\sqrt3}{4}x_1\right)} + \la\left(-1 + \tfrac{\sqrt3}{2} (\log r) (1-x_1)\right)\right),
  \]
  \fi
  where $r = 2+\sqrt{3}$ and $x_1 = 1/(1 + 2r(r^{2\la}-1)/(r^2 - r^{2\la}))$.
\end{theorem}

We are ready to determine the exact value of $R(\Of)$.

\begin{proof}[Proof of \cref{thm:exact}]
  In view of \cref{lem:belokopytov}, it suffices to prove that $R(\Of) \le H_2(1/2-\delta, 1/2+\delta)$, where
  \[
    \delta = 1/(2\log r) \quad\text{and}\quad r = 2+\sqrt3.
  \]
  \cref{thm:main} says that there exists $\spabn \in D_n$ such that
  \begin{gather*}
    R(\Of) = \tfrac12 \sumnpaabb,\\
    S_3(x^*) \ge \sumnphabcxs.
  \end{gather*}
  
  Clearly $R(\Of) \le L\spabn$, where $L$ is defined as in \cref{thm:lag} for every $\la \in (0,1/2)$. With hindsight, we take
  \[
    \la := \delta\log\frac{1+2\delta}{1-2\delta} \approx 0.44424.
  \]
  \cref{thm:lag} implies that $R(\Of)$ is at most the largest of the following three quantities:
  \ifieee
  \begin{multline*}
    \tfrac{1+\la}{2}, \quad 1 + \la \log(1+x_1), \quad
    H_2{\left(\tfrac{r}{4} - \tfrac{\sqrt3}{4}x_1, \tfrac{1}{4r} + \tfrac{\sqrt3}{4}x_1\right)} \\ + \la\left(-1 + \tfrac{\sqrt3}{2} (\log r) (1-x_1)\right),
  \end{multline*}
  \else
  \[
    \tfrac{1+\la}{2}, \quad 1 + \la \log(1+x_1), \quad H_2{\left(\tfrac{r}{4} - \tfrac{\sqrt3}{4}x_1, \tfrac{1}{4r} + \tfrac{\sqrt3}{4}x_1\right)} + \la\left(-1 + \tfrac{\sqrt3}{2} (\log r) (1-x_1)\right),
  \]
  \fi
  where $x_1 = 1/(1 + 2r(r^{2\la}-1)/(r^2 - r^{2\la})) \approx 0.39225$.
  The above three quantities are approximately $0.72212, 0.76189$ and $0.78974$, and so the last quantity is the largest. Moreover, one can simplify $x_1$ to $1 - 4\delta/\sqrt3$, and the last quantity to $H_2(1/2-\delta, 1/2+\delta)$. Therefore $R(\Of) \le H_2(1/2-\delta, 1/2+\delta)$.
\end{proof}

The rest of the paper is organized as follows. In \cref{sec:connection} we document the folklore correspondence between the group testing  problem and the binary adder channel, and we prove \cref{thm:correspondence}. In \cref{sec:reformulation} we prove \cref{thm:dueck-concrete}, that is the reformulation of Dueck's characterization, and we reproduce the proof of \cref{lem:belokopytov}, that is Belokopytov's lower bound on $R(\Of)$. In \cref{sec:simplification} we prove \cref{thm:main} assuming \cref{lem:unique-solution}, which is proved later in \cref{sec:convex}. In \cref{sec:proof-of-main} we finish the proof of \cref{thm:lag}. We conclude in \cref{sec:open} with an open problem.

\section{Application to quantitative group testing} \label{sec:connection}

The connection between the quantitative group testing problem with two defectives and the binary adder channel is established through a variation of the group testing problem. Aigner \cite{Aig86} introduced this variation in the context of search problems on graphs. Consider two disjoint sets $X_1$ and $X_2$, each of which consists of $n$ elements. It is known that each set contains precisely one defective element. We again perform a series of tests on subsets of $X_1 \cup X_2$. Let $t(n, n)$ be the minimum number of tests required in this situation. Hao \cite{Hao90} observed that $t(n) \sim t(n, n)$.

We establish a one-to-one correspondence between the above variation of the quantitative group testing problem, and the binary adder channel with complete feedback.

\begin{theorem} \label{thm:bijection}
  For every $m, n \in \mathbb{N}$, there is a bijection between
  \begin{enumerate}[label=(\alph*)]
    \item adaptive group testing strategies using $m$ tests that identify the defectives in two disjoint $n$-element sets, each of which contains exactly one defective, and
    \item $(n,n,m)$ uniquely decodeable codes for the binary adder channel with complete feedback.
  \end{enumerate}
\end{theorem}

\begin{proof}
  Let $M_1$ and $M_2$ be two disjoint $n$-element sets, each of which contains one defective. Formally, an adaptive group testing strategy using $n$ tests on $M_1$ and $M_2$, consists of two sequences $(d_{1k})_{k=1}^m$ and $(d_{2k})_{k=1}^n$ of \emph{decision functions}, where each $d_{ik}$ is a function from $\set{0,1,2}^{k-1}$ to $2^{X_i}$.
  
  The interpretation of the decision functions $d_{1k}$ and $d_{2k}$ is as follows. Suppose $w_1 \in M_1$ and $w_2 \in M_2$ are the two defectives. Before the $k$th test, if the previous $k-1$ results are $\tilde{y}_1, \tilde{y}_2, \dots, \tilde{y}_{k-1} \in \set{0,1,2}$, then the $k$th tested subset is the union of $X_{1k} := d_{1k}(y_1, \dots, y_{k-1}) \subseteq X_1$ and $X_{2k} := d_{2k}(y_1, \dots, y_{k-1}) \subseteq X_2$, and the $k$th result $\tilde{y}_{k}(x_1, x_2) := \bm{1}_{X_{1k}}(w_1) + \bm{1}_{X_{2k}}(w_2)$, where $\bm{1}_{X_{ik}}$ is the indicator function of $X_{ik}$. Let $\tilde{y}(w_1, w_2)$ be the sequence $(\tilde{y}_1, \dots, \tilde{y}_m)$ of results. Clearly such an adaptive strategy can identify the defectives if and only if
  \ifieee
  \begin{multline*}
    \tilde{y}(w_1, w_2) = \tilde{y}(w_1', w_2') \Leftrightarrow (w_1, w_2) = (w_1', w_2'), \\
    \text{for all }(w_1, w_2), (w_1', w_2') \in M_1 \times M_2.
  \end{multline*}
  \else
  \[
    \tilde{y}(w_1, w_2) = \tilde{y}(w_1', w_2') \Leftrightarrow (w_1, w_2) = (w_1', w_2'), \text{ for all }(w_1, w_2), (w_1', w_2') \in M_1 \times M_2.
  \]
  \fi

  Recall that an $(n,n,m)$ uniquely decodable code consists of two sequences $(e_{1k})_{k=1}^m$ and $(e_{2k})_{k=1}^m$ of encoding functions, where each $e_{ik}$ is a function from $\set{0,1,2}^{k-1} \times M_i \to \set{0,1}$, where $M_i$ is the $n$-message set of the $i$th sender. During the $k$th use of the channel, the binary adder channel outputs $y_k := e_{1k}(y_1, \dots, y_{k-1}, w_1) + e_{2k}(y_1, \dots, y_{k-1}, w_2)$. Let $y(w_1, w_2)$ be the sequence $(y_1, \dots, y_m)$ of outputs.
  
  The bijection maps a pair of sequences $(d_{1k})_{k=1}^m$ and $(d_{2k})_{k=1}^m$ of decision functions to a pair of sequences $(e_{1k})_{k=1}^m$ and $(e_{2k})_{k=1}^m$ of encoding functions as follows:
  \ifieee
  \begin{multline*}
    e_{ik}(y_1, \dots, y_{k-1}, w_i) = \bm{1}_{d_{ik}(y_1, \dots, y_{k-1})}(w_i) \\
    \text{for every }y_1, \dots, y_{k-1} \in \set{0,1,2} \text{ and } w_i \in M_i.
  \end{multline*}
  \else
  \[
    e_{ik}(y_1, \dots, y_{k-1}, w_i) = \bm{1}_{d_{ik}(y_1, \dots, y_{k-1})}(w_i)\quad
    \text{for every }y_1, \dots, y_{k-1} \in \set{0,1,2} \text{ and } w_i \in M_i.
  \]
  \fi
  Apparently the inverse of the pair $(e_{1k})_{k=1}^m$ and $(e_{2k})_{k=1}^m$ is the pair $(d_{1k})_{k=1}^m$ and $(d_{2k})_{k=1}^m$ defined as follows:
  \ifieee
  \begin{multline*}
    d_{ik}(y_1, \dots, y_{k-1}) = \dset{w_i \in M_i}{e_{ik}(y_1, \dots, y_{k-1, w_i}) = 1} \\
    \text{for every }y_1, \dots, y_{k-1} \in \set{0,1,2}.
  \end{multline*}
  \else
  \[
    d_{ik}(y_1, \dots, y_{k-1}) = \dset{w_i \in M_i}{e_{ik}(y_1, \dots, y_{k-1, w_i}) = 1}\quad \text{for every }y_1, \dots, y_{k-1} \in \set{0,1,2}.
  \]
  \fi
  
  Finally it is routine to check that for every $(w_1, w_2) \in M_1 \times M_2$, the sequence $\tilde{y}(w_1, w_2)$ of results determined by $(d_{1k})_{k=1}^m$ and $(d_{2k})_{k=1}^m$, and the sequence $y(w_1, w_2)$ of outputs determined by $(e_{1k})_{k=1}^m$ and $(e_{2k})_{k=1}^m$, are exactly the same. Therefore the adaptive strategy $((d_{1k})_{k=1}^m, (d_{2k})_{k=1}^m)$ can identify the defectives if and only if the $(n,n,m)$ code $((e_{1k})_{k=1}^m, (e_{2k})_{k=1}^m)$ is uniquely decodable.
\end{proof}

\begin{proof}[Proof of \cref{thm:correspondence}]
  Suppose that $t(n) \sim t(n,n) \sim (\log n)/r$ for some positive constant $r$. By \cref{thm:bijection}, for every $n$ there exists an $(n, n, t(n, n))$ uniquely decodable code, hence
  \[
    r = \lim_{n\to\infty} \frac{\log n}{t(n,n)} \le R(\Of).
  \]
  Conversely, for every $R < R(\Of)$, there exists $n_0$ such that for every $n \ge n_0$ there exists a $(\ceil{2^{nR}}, \ceil{2^{nR}}, n)$ uniquely decodable code. By \cref{thm:bijection}, such a uniquely decodable code implies that $t(\ceil{2^{nR}},\ceil{2^{nR}}) \le n$ for all $n \ge n_0$, and so
  \[
    \frac{nR}{r} = \frac{\log 2^{nR}}{r} \sim t(\ceil{2^{nR}},\ceil{2^{nR}}) \le n.
  \]
  This implies $r \ge R(\Of)$. Therefore $r = R(\Of)$.
\end{proof}

\section{Reformulation and lower bound on \texorpdfstring{$R(\Of)$}{R(Of)}} \label{sec:reformulation}

We follow Belokopytov's computation in \cite{Bel89}.

\begin{proof}[Proof of \cref{thm:dueck-concrete}]
  In view of \cref{thm:dueck-rv}, we know that $(R, R) \in \Of$ if and only if there exist two Bernoulli random variables $X_1, X_2$ and an auxiliary discrete random variable $U$ such that
  \begin{subequations}
    \ifieee
    \begin{gather}
      R \le H(X_1 \mid U), \quad R \le H(X_2 \mid U), \label{eqn:dueck-1}\\
      \begin{split} \label{eqn:dueck-2}
        H(X_1', X_2' \mid U, X_1' + X_2') \le I(U; X_1' + X_2') \\
        \text{for every }(X_1', X_2') \in \mathcal{P}(X_1, X_2, U),
      \end{split}
    \end{gather}
    \else
    \begin{gather}
      R \le H(X_1 \mid U), \quad R \le H(X_2 \mid U), \label{eqn:dueck-1}\\
      H(X_1', X_2' \mid U, X_1' + X_2') \le I(U; X_1' + X_2') \text{ for every }(X_1', X_2') \in \mathcal{P}(X_1, X_2, U), \label{eqn:dueck-2}
    \end{gather}
    \fi
  \end{subequations}
  where $\mathcal{P}(X_1, X_2, U)$ consists of pairs $(X_1', X_2')$ of discrete random variables such that the conditional probability distributions satisfy
  \begin{equation} \label{eqn:p-x1-x2-u}
    p_{X_1'\mid U} = p_{X_1\mid U}, \quad p_{X_2'\mid U} = p_{X_2\mid U}.
  \end{equation}

  Suppose that the discrete random variable takes $n$ values. Without loss of generality, we may assume that the values of $U$ are $1, \dots, n$.
  By setting the conditional probability distribution as follows,
  \ifieee
  \begin{gather*}
    p_i := \pr(U = i),\\
    a_i := \pr(X_1 = 0 \mid U = i), \quad b_i := \pr(X_2 = 0 \mid U = i),
  \end{gather*}
  \else
  \[
    p_i := \pr(U = i), \quad a_i := \pr(X_1 = 0 \mid U = i), \quad b_i := \pr(X_2 = 0 \mid U = i),
  \]
  \fi
  we know that
  \[
    \pr(X_1 = 1 \mid U = i) = \aaai, \quad \pr(X_2 = 1 \mid U = i) = \bbi.
  \]
  Furthermore, one can deduce from \eqref{eqn:p-x1-x2-u} that for every $(X_1', X_2') \in \mathcal{P}(X_1, X_2, U)$ there exist $c_1 \in C_1, \dots, c_n \in C_n$, where $C_i := \cinterval$, such that
  \ifieee
  \begin{gather*}
    \pr(X_1' = 0, X_2' = 0 \mid U = i) = a_ib_i - c_i, \\
    \pr(X_1' = 0, X_2' = 1 \mid U = i) = a_i\bbi + c_i, \\
    \pr(X_1' = 1, X_2' = 0 \mid U = i) = \aaai b_i + c_i, \\
    \pr(X_1' = 1, X_2' = 1 \mid U = i) = \aaai\bbi - c_i.
  \end{gather*}
  \else
  \begin{gather*}
    \pr(X_1' = 0, X_2' = 0 \mid U = i) = a_ib_i - c_i, \quad \pr(X_1' = 0, X_2' = 1 \mid U = i) = a_i\bbi + c_i, \\
    \pr(X_1' = 1, X_2' = 0 \mid U = i) = \aaai b_i + c_i, \quad \pr(X_1' = 1, X_2' = 1 \mid U = i) = \aaai\bbi - c_i.
  \end{gather*}
  \fi
  By the chain rule of conditional entropy, we have
  \ifieee
  \begin{align*}
    & \phantom{{}={}} I(U; X_1' + X_2') - H(X_1', X_2' \mid U, X_1' + X_2') \\
    & = H(X_1' + X_2') - H(X_1' + X_2' \mid U) \\
    & \phantom{{}={}} - H(X_1', X_2' \mid U, X_1' + X_2') \\
    & = H(X_1' + X_2') - H(X_1', X_2' \mid U),
  \end{align*}
  \else
  \begin{align*}
    & \phantom{{}={}} I(U; X_1' + X_2') - H(X_1', X_2' \mid U, X_1' + X_2') \\
    & = H(X_1' + X_2') - H(X_1' + X_2' \mid U) - H(X_1', X_2' \mid U, X_1' + X_2') \\
    & = H(X_1' + X_2') - H(X_1', X_2' \mid U),
  \end{align*}
  \fi
  and so $H(X_1', X_2' \mid U, X_1' + X_2') \le I(U; X_1' + X_2')$ is equivalent to $H(X_1' + X_2') \ge H(X_1', X_2' \mid U)$.
  Therefore \eqref{eqn:dueck-1} and \eqref{eqn:dueck-2} can be rewritten respectively in terms of the parameters $p_i, a_i, b_i$ as
  \begin{subequations}
    \ifieee
    \begin{gather}
      R \le \sumn p_iH_2(a_i, \aaai), \quad R \le \sumn p_iH_2(b_i, \bbi), \label{eqn:dueck-1a} \\
      \begin{split} \label{eqn:dueck-2a}
        H_3{\Biggl( \sumn p_i(a_ib_i - c_i), \sumnpabc,} \\
        \sumn p_i(\aaai\bbi - c_i)\Biggr) \ge \sumn p_iS_4(a_i, b_i, c_i) \\
        \text{for every }(c_1, \dots, c_n) \text{ with }c_i \in C_i.
      \end{split}
    \end{gather}
    \else
    \begin{gather}
      R \le \sumn p_iH_2(a_i, \aaai), \quad R \le \sumn p_iH_2(b_i, \bbi), \label{eqn:dueck-1a} \\
      \begin{multlined} \label{eqn:dueck-2a}
        H_3{\left( \sumn p_i(a_ib_i - c_i), \sumnpabc, \sumn p_i(\aaai\bbi - c_i)\right)} \ge \sumn p_iS_4(a_i, b_i, c_i) \\
        \text{ for every }(c_1, \dots, c_n) \text{ with }c_i \in C_i.
      \end{multlined}
    \end{gather}
    \fi
  \end{subequations}

  It suffices to show that \eqref{eqn:dueck-1a} and \eqref{eqn:dueck-2a} are equivalent to
  \begin{subequations}
    \ifieee
    \begin{gather}
      R \le \tfrac12 \sumnpaabb, \label{eqn:6a} \\
      \begin{split} \label{eqn:6b}
        S_3{\left(\sumnpabc\right)} \ge \sumn p_i S_4(a_i, b_i, c_i) \\
        \text{for every }c_1 \in C_1, \dots, c_n \in C_n.
      \end{split}
    \end{gather}
    \else
    \begin{gather}
      R \le \tfrac12 \sumnpaabb, \label{eqn:6a} \\
      S_3\left(\sumnpabc\right) \ge \sumn p_i S_4(a_i, b_i, c_i), \text{ for every }c_1 \in C_1, \dots, c_n \in C_n. \label{eqn:6b}
    \end{gather}
    \fi
  \end{subequations}
  Clearly \eqref{eqn:dueck-1a} implies \eqref{eqn:6a}. Because $x \mapsto -x \log x$ is concave, Jensen's inequality shows that \eqref{eqn:dueck-2a} implies \eqref{eqn:6b}. Conversely, suppose that $R$, $\spn \in \Delta^{n-1}$ and $\sabn \in [0,1]^{2n}$ satisfy \eqref{eqn:6a} and \eqref{eqn:6b}. Define
  \ifieee
  \begin{gather*}
    (p_1', \dots, p_{2n}') := (p_1/2, \dots, p_n/2, p_1/2, \dots, p_n/2), \\
    (a_1', \dots, a_{2n}') := (a_1, \dots, a_n, \bb_1, \dots, \bb_n), \\
    (b_1', \dots, b_{2n}') := (b_1, \dots, b_n, \aaa_1, \dots, \aaa_n).
  \end{gather*}
  \else
  \begin{gather*}
    (p_1', \dots, p_{2n}') := (p_1/2, \dots, p_n/2, p_1/2, \dots, p_n/2), \\
    (a_1', \dots, a_{2n}') := (a_1, \dots, a_n, \bb_1, \dots, \bb_n), \quad
    (b_1', \dots, b_{2n}') := (b_1, \dots, b_n, \aaa_1, \dots, \aaa_n).
  \end{gather*}
  \fi
  One can check that $R$, $(p_1', \dots, p_{2n}')$ and $(a_1', \dots, a_{2n}')$, $(b_1', \dots, b_{2n}')$ would satisfy \eqref{eqn:dueck-1a} and \eqref{eqn:dueck-2a} (with $n$ replaced by $2n$).
\end{proof}

For completeness we include Belokopytov's proof of the lower bound on $R(\Of)$ using our notations.

\begin{proof}[Proof of \cref{lem:belokopytov}]
  By setting $n = 1$ and $a_1 = b_1 = 1/2 - \delta$ in \cref{thm:dueck-concrete}, where
  \[
    \delta = 1/(2\log(2+\sqrt3)),
  \]
  we reduce \eqref{eqn:h3-h4} to
  \ifieee
  \begin{multline*}
    M(c) := S_3(1/2-2\delta^2 + 2c) - S_4(1/2-\delta, 1/2-\delta, c) \ge 0 \\
    \text{for every }c \in C,
  \end{multline*}
  \else
  \[
    M(c) := S_3(1/2-2\delta^2 + 2c) - S_4(1/2-\delta, 1/2-\delta, c) \ge 0, \quad \text{for every }c \in C,
  \]
  \fi
  where $C = [-(1/4-\delta^2), (1/2-\delta)^2]$. One can compute that the derivative of $M$ satisfies
  \[
    2^{M'(c)} = \frac{(1/4+\delta^2-c)^2}{4((1/2-\delta)^2-c)((1/2+\delta)^2-c)}.
  \]
  One can then check that the right hand side is increasing on $C$, and it equals $1$ only at
  \[
    c = 1/4 - 2\delta/\sqrt3 + \delta^2 =: c^*.
  \]
  Therefore $M$ is minimized at $c = c^*$ on $C$. Finally, one can check directly that
  \ifieee
  \begin{align*}
    \phantom{{}={}} M(c^*) & = H_3{\left(\tfrac{2\delta}{\sqrt3}, \tfrac{\sqrt3-4\delta}{\sqrt3}, \tfrac{2\delta}{\sqrt3}\right)} \\
    & \phantom{{}={}} - H_4{\left(\tfrac{(2-\sqrt3)\delta}{\sqrt3}, \tfrac{1}{2}-\tfrac{2\delta}{\sqrt3}, \tfrac12-\tfrac{2\delta}{\sqrt3}, \tfrac{(2+\sqrt3)\delta}{\sqrt3}\right)} \\
    & = \biggl(-\tfrac{4}{\sqrt3}+\tfrac{2-\sqrt3}{\sqrt3}\log(2-\sqrt3) \\
    & \phantom{{}={}} + \tfrac{2+\sqrt3}{\sqrt3}\log(2+\sqrt3)+\tfrac{4}{\sqrt3}\biggr)\delta - 1,
  \end{align*}
  \else
  \begin{align*}
    M(c^*) & = H_3{\left(\frac{2\delta}{\sqrt3}, \frac{\sqrt3-4\delta}{\sqrt3}, \frac{2\delta}{\sqrt3}\right)} - H_4{\left(\frac{(2-\sqrt3)\delta}{\sqrt3}, \frac{1}{2}-\frac{2\delta}{\sqrt3}, \frac12-\frac{2\delta}{\sqrt3}, \frac{(2+\sqrt3)\delta}{\sqrt3}\right)} \\
    & = \left(-\frac{4}{\sqrt3}+\frac{2-\sqrt3}{\sqrt3}\log(2-\sqrt3) + \frac{2+\sqrt3}{\sqrt3}\log(2+\sqrt3)+\frac{4}{\sqrt3}\right)\delta - 1,
  \end{align*}
  \fi
  which is exactly $0$ according to the choice of $\delta$.
\end{proof}

\section{Quantifier elimination} \label{sec:simplification}

In this section we reduce \eqref{eqn:h3-h4} in \cref{thm:dueck-concrete} to a single constraint using standard tools from multi-variable calculus. We need the following properties of $\cab(x)$.

\begin{proposition} \label{lem:cab-prop}
  For every $a, b$ and $x \in [0,1]$, the equation
  \begin{subequations} \label{eqn:cab-prop-c-def}
    \begin{gather}
      (ab - c)(\aaa\bb - c)(2x)^2 = (a\bb + c)(\aaa b + c)(1-x)^2, \label{eqn:cab-prop-c-def-a} \\
      -\min(a\bb, \aaa b) \le c \le \min(ab, \aaa\bb)
    \end{gather}
  \end{subequations}
  has a unique solution $c = \cab(x)$, and moreover
  \begin{equation} \label{eqn:cab-prop-c-01}
    \cab(0) = -\min(a\bb, \aaa b) \text{ and }\cab(1) = \min(ab, \aaa \bb).
  \end{equation}
\end{proposition}

\begin{proof}
  In the extremal cases, it is easy to check that when $x = 0$, $c = -\min(a\bb, \aaa b)$ is the unique solution; when $x = 1$, $c = \min(ab, \aaa\bb)$ is the unique solution. In both extremal cases, \eqref{eqn:cab-prop-c-def} has a unique solution and \eqref{eqn:cab-prop-c-01} holds. For $0 < x < 1$, as $c$ increases from $-\min(a\bb, \aaa b)$ to $\min(ab, \aaa\bb)$, the left hand side of \eqref{eqn:cab-prop-c-def-a} decreases to $0$ whereas the right hand side increases from $0$, and so \eqref{eqn:cab-prop-c-def} has a unique solution.
\end{proof}

\begin{lemma} \label{lem:q-elim}
  For every $\spn \in \Delta^{n-1}$, $\san$ and $\sbn \in [0,1]$, set
  \[
    C_i := \cinterval.
  \]
  Then the inequality
  \[
    S_3\left(\sumnpabc\right) \ge \sumn p_i S_4(a_i, b_i, c_i)
  \]
  holds for every $c_1 \in C_1, \dots, c_n \in C_n$ if and only if the inequality
  \[
    S_3(x) \ge \sumn p_iS_4(a_i,b_i,c_{a_i,b_i}(x))
  \]
  holds for every solution $x$ to the equation
  \begin{equation} \label{eqn:unique-x}
    x = \sumnpabcx, \quad x \in (0,1],
  \end{equation}
  where $\cabi(x)$ is the unique solution to the following equation for $c$ in $C_i$:
  \[
    (a_ib_i - c)(\aaai\bbi - c)(2x)^2 = (a_i\bbi + c)(\aaai b_i + c)(1-x)^2.
  \]
\end{lemma}

\begin{proof}
  Fix $\spn \in \Delta^{n-1}$, $\san$, and $\sbn \in [0,1]$. Note that the ``if and only if'' statement does not change at all if we drop the summands with $p_i = 0$ in the sums. Without loss of generality, we may assume that $p_i > 0$ for every $i \in [n]$. Define $M\colon C \to \mathbb{R}$ by
  \ifieee
  \begin{multline*}
    M(c_1, \dots, c_n) := S_3{\left(\sumnpabc\right)} \\ - \sumn p_iS_4(a_i,b_i,c_i), \quad C := C_1 \times \dots \times C_n.
  \end{multline*}
  \else
  \[
    M(c_1, \dots, c_n) := S_3{\left(\sumnpabc\right)} - \sumn p_iS_4(a_i,b_i,c_i), \quad C := C_1 \times \dots \times C_n.
  \]
  \fi
  The ``only if'' direction is trivial. For the ``if'' direction, we characterize the minimum points of $M$.
  
  \begin{claimm}
    Every minimum point $(c_1, \dots, c_n) \in C$ of the function $M$ satisfies
    \ifieee
    \begin{multline*}
      c_i = \cabi(x) \text{ for every }i\in[n], \\
      \text{where } x = \sumnpabc.
    \end{multline*}
    \else
    \[
      c_i = \cabi(x) \text{ for every }i\in[n], \text{ where } x = \sumnpabc.
    \]
    \fi
  \end{claimm}
  
  \begin{claimproof}
    Suppose that $(c_1, \dots, c_n)$ is a minimum point of $M$, and set
    \[
      x := \sumnpabc.
    \]
    From the assumption $c_i \in C_i$, it is easy to check $0 \le a_i\bbi + \aaai b_i + 2c_i \le 1$, and so $x \in [0,1]$. In the extremal cases, if $x = 0$ then $a_i\bbi + c_i = \aaai b_i + c_i = 0$, and so $c_i = -\min(a_i\bbi, \aaai b_i)$; and if $x = 1$ then $a_i\bbi + c_i = \aaai b_i + c_i = 1$, and so $c_i = \min(a_ib_i, \aaai\bbi)$. In both extremal cases, in view of \eqref{eqn:cab-prop-c-01}, $c_i = \cabi(x)$ for every $i \in [n]$.
  
    Hereafter we may assume that $0 < x < 1$. Since $M$ is a differentiable function on the interior of $C$, for every $i \in [n]$, $\partial M(c_1, \dots, c_n)/\partial c_i = 0$ or $c_i$ is on the boundary of $C_i$. Take an arbitrary $i \in [n]$. We break down the rest of the proof into two cases.

    \ifieee
      \medskip
      \textit{Case 1:} $\partial M(c_1, \dots, c_n) / \partial c_i = 0$.
    \else
      \paragraph{Case 1: $\partial M(c_1, \dots, c_n) / \partial c_i = 0$.}
    \fi
    We can compute
    \ifieee
    \begin{multline} \label{eqn:partial-di}
      \fpp{M}{c_i}(c_1, \dots, c_n) = p_i\biggl(2\log\mathopen{}\left(\frac{1-x}{2x}\right)\mathclose{} \\
      - \log\frac{(a_ib_i-c_i)(\aaai\bbi-c_i)}{(a_i\bbi + c_i)(\aaai b_i + c_i)}\biggr).
    \end{multline}
    \else
    \begin{equation} \label{eqn:partial-di}
      \fpp{M}{c_i}(c_1, \dots, c_n) = p_i\left(2\log\mathopen{}\left(\frac{1-x}{2x}\right)\mathclose{} - \log\frac{(a_ib_i-c_i)(\aaai\bbi-c_i)}{(a_i\bbi + c_i)(\aaai b_i + c_i)}\right).
    \end{equation}
    \fi
    Thus we know that
    \[
      \frac{(a_ib_i-c_i)(\aaai\bbi-c_i)}{(a_i\bbi + c_i)(\aaai b_i + c_i)} = \left(\frac{1-x}{2x}\right)^2,
    \]
    which implies that $c_i = \cabi(x)$.

    \ifieee
      \medskip
      \textit{Case 2:} $c_i \in \partial C_i$.
    \else
      \paragraph{Case 2: $c_i \in \partial C_i$.}
    \fi
    When $C_i = \set{0}$ is a degenerate interval, we have $a_ib_i\aaai \bbi = 0$ and $c_i = 0$, the former of which implies $\cabi(x) = 0$, and so $c_i = \cabi(x)$. Hereafter we consider the case where $C_i$ is a proper interval. For the $c_i = -\min(a_i\bbi, \aaai b_i)$ case, \eqref{eqn:partial-di} implies that
    \ifieee
    \begin{align*}
      & \phantom{{}={}} \lim_{c_i' \to c_i^+} \fpp{M}{c_i} (c_1, \dots, c_{i-1}, c_i', c_{i+1}, \dots, c_n) \\
      & = \lim_{c_i' \to c_i^+} p_i\left(2\log\mathopen{}\left(\frac{1-x}{2x}\right)\mathclose{}-\log\frac{(a_ib_i-c_i)(\aaai\bbi-c_i)}{(a_i\bbi + c_i)(\aaai b_i + c_i)}\right) \\
      & = -\infty,
    \end{align*}
    \else
    \begin{multline*}
      \lim_{c_i' \to c_i^+} \fpp{M}{c_i} (c_1, \dots, c_{i-1}, c_i', c_{i+1}, \dots, c_n) \\
      = \lim_{c_i' \to c_i^+} p_i\left(2\log\mathopen{}\left(\frac{1-x}{2x}\right)\mathclose{}-\log\frac{(a_ib_i-c_i)(\aaai\bbi-c_i)}{(a_i\bbi + c_i)(\aaai b_i + c_i)}\right) = -\infty,
    \end{multline*}
    \fi
    which contradicts with the assumption that $(c_1, \dots, c_n)$ is a minimum point. For the $c_i = \min(a_ib_i, \aaai\bbi)$ case, a similar computation gives
    \[
      \lim_{c_i' \to c_i^-} \fpp{M}{c_i} (c_1, \dots, c_{i-1}, c_i', c_{i+1}, \dots, c_n)=\infty,
    \]
    which leads to the same contradiction.
  \end{claimproof}
  Coming back to the ``if'' direction, let $(c_1, \dots, c_n)$ be a global minimum point of $M$. From the claim, we know that $c_i = \cabi(x)$ for every $i \in [n]$, where
  \[
    x = \sumnpabc, \quad x \in [0,1].
  \]
  If $x = 0$, then $a_i\bbi + c_i = \aaai b_i + c_i = 0$, and so
  \ifieee
  \begin{multline*}
    M(c_1, \dots, c_n) = S_3(0) - \sumn p_iS_4(a_i, b_i, c_i) \\
    = 1 - \sumn p_i H_2(a_ib_i - c_i,\aaai\bbi - c_i) \ge 1 - \sumn p_i = 0.
  \end{multline*}
  \else
  \[
    M(c_1, \dots, c_n) = S_3(0) - \sumn p_iS_4(a_i, b_i, c_i) = 1 - \sumn p_i H_2(a_ib_i - c_i,\aaai\bbi - c_i) \ge 1 - \sumn p_i = 0.
  \]
  \fi
  Otherwise $x$ is a solution to \eqref{eqn:unique-x}, and our assumption ensures that $M(c_1, \dots, c_n) \ge 0$ as well. Therefore the minimum of $M$ over $C$ is at least $0$.
\end{proof}

Observe that the existence and the uniqueness of the solution to \eqref{eqn:unique-x} is guaranteed by \cref{lem:unique-solution} for almost all $\spn \in \Delta^{n-1}$, $\san$ and $\sbn \in [0,1]$. Now we are ready to prove \cref{thm:main}.

\begin{proof} [Proof of \cref{thm:main}]
  In view of \cref{thm:dueck-concrete,lem:q-elim}, we know that the average zero-error capacity $R(\Of)$ is the optimum of the following optimization problem.
  \ifieee
  \begin{align*}
    \text{Maximize: } & \tfrac12 \sumnpaabb, \\
    \text{subject to: } & S_3(x) \ge \sumn p_iS_4(a_i, b_i, c_i(x)), \forall x\in (0,1] \\
    & \quad \text{s.t.\ } x = \sumnpabcx, \\
    & \spn \in \Delta^{n-1}, \\
    & \sabn \in [0,1]^{2n}.
  \end{align*}
  \else
  \begin{align*}
    \text{Maximize: } & \tfrac12 \sumnpaabb, \\
    \text{subject to: } & S_3(x) \ge \sumn p_iS_4(a_i, b_i, c_i(x)), \forall x\in (0,1] \text{ s.t.\ } x = \sumnpabcx, \\
    & \spn \in \Delta^{n-1}, \quad \sabn \in [0,1]^{2n}.
  \end{align*}
  \fi

  Recall from \cref{lem:unique-solution} that
  \ifieee
  \begin{multline*}
    D_n = \biggl\{\spabn \colon \\
    \sumn p_i\abs{a_i - b_i} > 0 \text{ or }\sumn p_i\sqrt{a_i\aaai} > 1/4\biggr\}.
  \end{multline*}
  \else
  \[
  D_n = \dset{\spabn}{\sumn p_i\abs{a_i - b_i} > 0 \text{ or }\sumn p_i\sqrt{a_i\aaai} > 1/4}.
  \]
  \fi
  If $\spabn \not\in D_n$, then $p_i = 0$ or $a_i = b_i$ for every $i \in [n]$ and $\sumn p_i\sqrt{a_i\aaa_i} \le 1/4$, hence the fact that $H_2(a, \aaa) \le 2\sqrt{a\aaa}$ for $a \in [0,1]$ implies
  \ifieee
  \begin{multline*}
    \tfrac12 \sumnpaabb \\
    = \sumn p_iH_2(a_i, \aaai) \le 2\sumn p_i\sqrt{a_i\aaai} \le 1/2.
  \end{multline*}
  \else
  \[
    \tfrac12 \sumnpaabb = \sumn p_iH_2(a_i, \aaai) \le 2\sumn p_i\sqrt{a_i\aaai} \le 1/2.
  \]
  \fi
  Since $R(\Of) > 1/2$ from \cref{lem:belokopytov}, we may assume that $\spabn \in D_n$ in the optimization problem. Thus \cref{lem:unique-solution} shows that the equation
  \[
    x = \sumnpabcx, \quad x \in (0,1]
  \]
  has a unique solution $x = x^*$. Therefore we can replace $x$ by the unique solution $x^*$ in the constraint:
  \ifieee
  \begin{align*}
    \text{Maximize: } & \tfrac12 \sumnpaabb, \\
    \text{subject to: } & S_3(x^*) \ge \sumnphabcxs, \\
    & \spabn \in D_n.
  \end{align*}
  \else
  \begin{align*}
    \text{Maximize: } & \tfrac12 \sumnpaabb, \\
    \text{subject to: } & S_3(x^*) \ge \sumnphabcxs, \\
    & \spabn \in D_n. \qedhere
  \end{align*}
  \fi
\end{proof}

\section{Convex combination of special functions} \label{sec:convex}

In this section we investigate the following special functions (see \cref{fig:phi-a-b} for their graphs). For every $a, b \in [0, 1]$, define the function $\phab\colon [0,1]\to\mathbb{R}$ by
\begin{equation} \label{eqn:phi-def}
  \phab(x) = a\bb + \aaa b + 2\cab(x) - x,
\end{equation}
where $\cab(x)$ is the unique solution to the following equation for $c$:
\begin{gather*}
  (ab-c)(\aaa \bb - c)(2x)^2 = (a\bb + c)(\aaa b+c)(1-x)^2, \\
  -\min(a\bb, \aaa b) \le \cab(x) \le \min(ab, \aaa\bb).
\end{gather*}

\begin{figure}[t]
  \centering
  \begin{tikzpicture}[baseline]
    \begin{axis}[general]
      \foreach \a in {0.0, 0.04, ..., 0.5} {
        \pgfmathsetmacro\b{\a}
        \pgfmathsetmacro\k{int((\a + \b) * 100)}
        \edef\AddPlot{\noexpand\addplot[smooth,color=red!\k!blue, domain=0:1,samples=50] {(sqrt(16 * x^4 - 8 * (\a + \b - 2 * \a * \b) * x^2 * (x + 1) * (3 * x - 1) + (\a - \b)^2 * (x + 1)^2 * (3 * x - 1)^2) + x * (3 * x + 1) * (x - 1))/((1 + x) * (1 - 3 * x))};}
        \AddPlot
      }
      \node at (0.5,-1) [below] {$t=0$};
    \end{axis}
  \end{tikzpicture}%
  \quad
  \begin{tikzpicture}[baseline]
    \begin{axis}[general]
      \foreach \a in {0.0, 0.04, ..., 0.4} {
        \pgfmathsetmacro\b{\a + 0.2}
        \pgfmathsetmacro\k{int((\a + \b - 0.2) / 0.8 * 100)}
        \edef\AddPlot{\noexpand\addplot[smooth,color=red!\k!blue, domain=0:1,samples=50] {(sqrt(16 * x^4 - 8 * (\a + \b - 2 * \a * \b) * x^2 * (x + 1) * (3 * x - 1) + (\a - \b)^2 * (x + 1)^2 * (3 * x - 1)^2) + x * (3 * x + 1) * (x - 1))/((1 + x) * (1 - 3 * x))};}
        \AddPlot
      }
      \node at (0.5,-1) [below] {$t = 0.2$};
    \end{axis}
  \end{tikzpicture}%
  \ifieee
  \\ \bigskip
  \else
  \quad
  \fi
  \begin{tikzpicture}[baseline]
    \begin{axis}[general]
      \foreach \a in {0.0, 0.04, ..., 0.3} {
        \pgfmathsetmacro\b{\a + 0.4}
        \pgfmathsetmacro\k{int((\a + \b - 0.4) / 0.6 * 100)}
        \edef\AddPlot{\noexpand\addplot[smooth,color=red!\k!blue, domain=0:1,samples=50] {(sqrt(16 * x^4 - 8 * (\a + \b - 2 * \a * \b) * x^2 * (x + 1) * (3 * x - 1) + (\a - \b)^2 * (x + 1)^2 * (3 * x - 1)^2) + x * (3 * x + 1) * (x - 1))/((1 + x) * (1 - 3 * x))};}
        \AddPlot
      }
      \node at (0.5,-1) [below] {$t = 0.4$};
    \end{axis}
  \end{tikzpicture}%
  \quad
  \begin{tikzpicture}[baseline]
    \begin{axis}[general]
      \foreach \a in {0.0, 0.04, ..., 0.2} {
        \pgfmathsetmacro\b{\a + 0.6}
        \pgfmathsetmacro\k{int((\a + \b - 0.6) / 0.4 * 100)}
        \edef\AddPlot{\noexpand\addplot[smooth,color=red!\k!blue, domain=0:1,samples=50] {(sqrt(16 * x^4 - 8 * (\a + \b - 2 * \a * \b) * x^2 * (x + 1) * (3 * x - 1) + (\a - \b)^2 * (x + 1)^2 * (3 * x - 1)^2) + x * (3 * x + 1) * (x - 1))/((1 + x) * (1 - 3 * x))};}
        \AddPlot
      }
      \node at (0.5,-1) [below] {$t = 0.6$};
    \end{axis}
  \end{tikzpicture}
  \caption{Graph families of $\phab$ with $t := \abs{a-b}$ fixed.}
  \label{fig:phi-a-b}
\end{figure}

Observe that \cref{lem:unique-solution} concerns the solutions to the equation
\ifieee
\begin{multline*}
  \sumnpabcx = x \\
  \text{or equivalently }\sumn p_i\ph_{a_i,b_i}(x) = 0.
\end{multline*}
\else
\[
  \sumnpabcx = x \quad\text{or equivalently}\quad\sumn p_i\ph_{a_i,b_i}(x) = 0.
\]
\fi
We need the following properties of $\phab$, the proof of which is done in \cref{sec:f-prime-over-f}.

\begin{proposition} \label{lem:phi-prop}
  For every $(a,b) \in [0,1]^2$, the function $\phab\colon [0,1]\to\mathbb{R}$ defined by \eqref{eqn:phi-def} has the following properties:
  \begin{enumerate}[label=(\alph*),nosep]
    \item $\phab(0) \ge 0$, equality holds if and only if $a = b$, and $\phab(1) \le 0$, equality holds if and only if $a + b = 1$; \label{lem:us-0}
    \item $\phab$ is a continuously differentiable function on $[0,1]$; \label{lem:us-c2}
    \item when $a = b$,
    \ifieee
    \[
      \phab(x) =
      \begin{dcases}
        -x & \text{ if } a \in \set{0,1}, \\
        \begin{multlined}
          (4\sqrt{a\aaa}  - 1)x \\ +(4\sqrt{a\aaa}- 4)x^2 + o(x^2)
        \end{multlined} & \text{ if } a \in (0,1);
      \end{dcases} 
    \]
    \else
    \[
      \phab(x) =
      \begin{cases}
        -x & \text{ if } a \in \set{0,1}, \\
        (4\sqrt{a\aaa}  - 1)x+(4\sqrt{a\aaa}- 4)x^2 + o(x^2) & \text{ if } a \in (0,1);
      \end{cases} 
    \]
    \fi \label{lem:us-a}
    \item if $(a,b) \not\in \oooo$, then the equation
    \[
      \phab(x) = x\phab'(x), \quad x \in (0,1]
    \]
  has no solution for $x$. \label{lem:us-b}
  \end{enumerate}
\end{proposition}

The following lemma deals with the uniqueness of the solution to $\ph(x) = 0$. For technical reasons, which will become clear in \cref{sec:proof-of-main}, we also study the derivative $\ph'(x)$ at the unique solution.

\begin{lemma} \label{lem:uniqueness}
  For every $\spn \in \Delta^{n-1}$, $\san$ and $\sbn \in [0,1]$, the equation
  \begin{equation} \label{eqn:uniqueness-ph-equals-zero}
    \ph(x) := \sumn p_i\ph_{a_i,b_i}(x) = 0, \quad x \in (0,1]
  \end{equation}
  has at most one solution for $x$, where $\ph_{a_i,b_i}$ is defined by \eqref{eqn:phi-def}. Moreover, if \eqref{eqn:uniqueness-ph-equals-zero} has a unique solution $x = x^*$, then $\ph'(x^*) < 0$.\footnote{\cref{lem:phi-prop}\ref{lem:us-c2} ensures that $\ph'$ is well-defined on $[0,1]$.}
\end{lemma}

\begin{proof}
  We prove by induction on $n$.

  \ifieee
    \medskip
    \textit{The base case:} $n = 1$.
  \else
    \paragraph{The base case: $n = 1$.}
  \fi
  It suffices to show that for every $a$ and $b\in [0,1]$, the following system of equations has at most one solution for $x$ in $(0,1]$:
  \ifieee
  \begin{gather*}
    a\bb + \aaa b + 2c - x = 0, \\
    (ab-c)(\aaa\bb - c)(2x)^2 = (a\bb + c)(\aaa b + c)(1-x)^2.
  \end{gather*}
  \else
  \[
    a\bb + \aaa b + 2c - x = 0, \quad
    (ab-c)(\aaa\bb - c)(2x)^2 = (a\bb + c)(\aaa b + c)(1-x)^2.
  \]
  \fi
  Substituting $c$ by $(x - a\bb - \aaa b)/2$ in the second equation, we obtain an equivalent equation:
  \[
    (a + b - x)(\aaa + \bb - x)(2x)^2 + (a-b+x)(a-b-x)(1-x)^2 = 0.
  \]
  Let $m(x)$ be the left hand side of the above equation.
  One can compute
  \begin{equation} \label{eqn:mxmx}
    2m(x) - xm'(x) = 2(1-x)((a-b)^2+3x^2).
  \end{equation}
  
  Assume for the sake of contradiction that $m(x) = 0$ has two solutions $x_1 < x_2$ in $(0,1]$. If the minimum value of the function $m$ over the interval $[0, x_2]$ is zero, then $x_1$ is a minimum point of $m$ over $(0,x_2)$, and so $m'(x_1) = 0$, which implies $2m(x_1) - x_1m'(x_1) = 0$, contradicting \eqref{eqn:mxmx} for $x_1 < 1$. Otherwise $m(x_0) < 0$, where $x_0$ is a minimum point of $m$ over $[0,x_2]$. Because $m(0) = (a-b)^2 \ge 0$ and $m(x_2) = 0$, $x_0 \in (0, x_2)$ and so $m'(x_0) = 0$, which implies $2m(x_0) - x_0m'(x_0) < 0$, contradicting \eqref{eqn:mxmx} directly.

  For the ``moreover'' part, suppose that $x = x^*$ is the unique solution to $\phab(x) = 0$ in $(0,1]$. Since $\phab(x) = -x$ when $(a,b) \in \oooo$ (via \cref{lem:phi-prop}\ref{lem:us-a}), which has no solution in $(0,1]$, we may assume that $(a,b) \not\in \oooo$. As $\phab(0) \ge 0$ (via \cref{lem:phi-prop}\ref{lem:us-0}) and $\phab(x) \neq 0$ for $x \in (0,x^*)$, $\phab(x) > 0$ for all $x \in (0, x^*)$, and so $\phab'(x^*) \le 0$. \cref{lem:phi-prop}\ref{lem:us-0} implies that $\phab'(x^*) \neq 0$ and so $\phab'(x^*) < 0$.
  
  \ifieee
    \medskip
    \textit{The base case:} $n = 2$.
  \else
    \paragraph{The base case: $n = 2$.}
  \fi
  We use the shorthand $\ph_1$ and $\ph_2$ for $\ph_{a_1,b_1}$ and $\ph_{a_2,b_2}$ respectively. Since $\ph_i(0) \ge 0 \ge \ph_i(1)$ (via \cref{lem:phi-prop}\ref{lem:us-0}), the intermediate value theorem implies that the equation $\ph_i(x) = 0$ has at least one solution for $x$ in $[0,1]$. Let $x = x_i$ be the largest solution to $\ph_i(x) = 0$ in $[0,1]$.

  \begin{claim} \label{claim:case2}
    For every $i \in [2]$, $\ph_i(x) > 0$ for $x \in (0, x_i)$ and $\ph_i(x) < 0$ for $x \in (x_i, 1]$.
  \end{claim}

  \begin{claimproof}[Proof of \cref{claim:case2}]
    If $x_i = 0$, then using the assumption that $x_i$ is the largest solution in $[0,1]$ and the fact that $\ph_i(1) \le 0$, we know that $\ph_i(x) < 0$ for $x\in (0,1]$. We are left to deal with the case where $x_i \in (0,1]$. In the base case where $n = 1$, we have already shown that $x = x_i$ is the unique solution to $\ph_i(x) = 0$ in $(0,1]$. The claim then follows immediately from \cref{lem:phi-prop}\ref{lem:us-0}.
  \end{claimproof}

  Without loss of generality, we may assume that $p_1, p_2 > 0$ and $x_1 \le x_2$. Assume for the sake of contradiction that the equation
  \[
    p_1\ph_1(x) + p_2\ph_2(x) = 0, \quad x \in (0,1]
  \]
  has two solutions $x_1^* < x_2^*$. \cref{claim:case2} shows that $x_1 < x_1^* < x_2^* < x_2$, and $\ph_1(x) < 0 < \ph_2(x)$ for $x \in (x_1, x_2)$. Consider the quotient $q(x)\colon (x_1, x_2) \to \mathbb{R}$ defined by
  \[
    q(x) := \ph_1(x)/\ph_2(x).
  \]
  Since $q(x_1^*) = -p_2 / p_1 = q(x_2^*)$, we obtain a contradiction with the following claim.
  
  \begin{claim} \label{claim:case2-d}
    If $x_1 < x_2$ and $\ph_1(x) < 0 < \ph_2(x)$ for $x\in (x_1,x_2)$, then $q'(x) < 0$ for $x \in (x_1, x_2)$.
  \end{claim}

  \begin{claimproof}[Proof of \cref{claim:case2-d}]
    It suffices to show that for $x \in (x_1, x_2)$,
    \ifieee
    \[
      \ph_1'(x)\ph_2(x) < \ph_1(x)\ph_2'(x)
      \text{ or equivalently }
      \frac{\ph_1'(x)}{\ph_1(x)} > \frac{\ph_2'(x)}{\ph_2(x)}.
    \]
    \else
    \[
      \ph_1'(x)\ph_2(x) < \ph_1(x)\ph_2'(x)
      \quad \text{or equivalently} \quad
      \frac{\ph_1'(x)}{\ph_1(x)} > \frac{\ph_2'(x)}{\ph_2(x)}.
    \]
    \fi
    Since $x_2 \in (0,1]$, the base case, where $n = 1$, says that $\ph_2'(x_2) < 0$, and so
    \begin{equation} \label{eqn:ph2-less-than-1-over-x}
      \lim_{x\to x_2^-} \frac{\ph_2'(x)}{\ph_2(x)} = -\infty < \frac{1}{x_2}.
    \end{equation}
    Clearly $(a_2, b_2) \not\in \oooo$ for otherwise $\ph_2(x) = -x$ (via \cref{lem:phi-prop}\ref{lem:us-0}) which has no solution in $(0,1]$. As $\ph_2$ is continuously differentiable (via \cref{lem:phi-prop}\ref{lem:us-c2}), the function $\ph_2'/\ph_2$ is continuous on $(x_1, x_2)$. Thus \cref{lem:phi-prop}\ref{lem:us-b} and \eqref{eqn:ph2-less-than-1-over-x} imply that
    \[
      \ph_2'(x)/\ph_2(x) < 1/x \quad \text{for }x \in (x_1, x_2).
    \]
    
    It suffices to prove that $\ph_1'(x)/\ph_1(x) \ge 1/x$ for every $x \in (x_1, x_2)$. We break into two cases. If $x_1 \in (0,1]$, then
    \[
      \lim_{x\to x_1^+}\ph_1'(x)/\ph_1(x) = \infty > 1/x_1,
    \]
    and a similar reasoning yields that
    \[
      \ph_1'(x)/\ph_1(x) > 1/x \quad \text{for } x \in (x_1, x_2).
    \]
    Otherwise $x_1 = 0$. Since $\ph_1(0) = 0$, \cref{lem:phi-prop}\ref{lem:us-0} implies that $a_1 = b_1$. From \cref{claim:case2}, we know that $\ph_1(x) < 0$ for $x \in (0,1]$, and so $\ph_1'(0) \le 0$, which implies through \cref{lem:phi-prop}\ref{lem:us-a} that $\sqrt{a_1\aaa_1} \le 1/4$. If $\sqrt{a_1\aaa_1} = 0$, that is, $(a_1, b_1) \in \oooo$, then $\ph_1(x) = -x$, hence $\ph_1'(x)/\ph_1(x) = 1/x$ and we are done. If $0 < \sqrt{a_1\aaa_1} \le 1/4$, we conclude from \cref{lem:phi-prop}\ref{lem:us-a} that
    \ifieee
    \begin{multline*}
      \frac{x\ph_1'(x)}{\ph_1(x)} = \frac{(4\sqrt{a_1\aaa_1}-1) + (8\sqrt{a_1\aaa_1}-8)x + o(x)}{(4\sqrt{a_1\aaa_1}-1) + (4\sqrt{a_1\aaa_1}-4)x + o(x)} \\
      = \begin{cases}
        1 + \left(\frac{4\sqrt{a_1\aaa_1}-4}{4\sqrt{a_1\aaa_1}-1}\right)x + o(x) & \text{if }0<\sqrt{a_1\aaa_1} < 1/4; \\
        2 + o(1) & \text{if }\sqrt{a_1\aaa_1} = 1/4,
      \end{cases}
    \end{multline*}
    \else
    \begin{align*}
      \frac{x\ph_1'(x)}{\ph_1(x)}
      & = \frac{(4\sqrt{a_1\aaa_1}-1) + (8\sqrt{a_1\aaa_1}-8)x + o(x)}{(4\sqrt{a_1\aaa_1}-1) + (4\sqrt{a_1\aaa_1}-4)x + o(x)} \\
      & = \begin{cases}
        1 + \left(\frac{4\sqrt{a_1\aaa_1}-4}{4\sqrt{a_1\aaa_1}-1}\right)x + o(x) & \text{if }0<\sqrt{a_1\aaa_1} < 1/4; \\
        2 + o(1) & \text{if }\sqrt{a_1\aaa_1} = 1/4,
      \end{cases}
    \end{align*}
    \fi
    and so there exists $\eps \in (0,1)$ such that
    \[
      x\ph_1'(x)/\ph_1(x) > 1\quad \text{for }x \in (0, \eps).
    \]
    Thus \cref{lem:phi-prop}\ref{lem:us-b} implies that
    \ifieee
    \begin{equation*}
      \ph_1'(x) / \ph_1(x) > 1/x\quad \text{for } x\in(0,1].
    \end{equation*}
    \else
    \begin{equation*}
      \ph_1'(x) / \ph_1(x) > 1/x\quad \text{for } x\in(0,1].\qedhere
    \end{equation*}
    \fi
  \end{claimproof}

  For the ``moreover'' part, suppose that $x = x^*$ is the unique solution to
  \[
    \ph(x) := p_1\ph_1(x) + p_2\ph_2(x) = 0, \quad x \in (0,1].
  \]
  If $\ph_1(x^*) = \ph_2(x^*) = 0$, then the base case where $n = 1$ says that both $\ph_1'(x^*)$ and $\ph_2'(x^*)$ are negative, and so is $\ph'(x^*)$. Without loss of generality we may assume that $\ph_1(x^*) < 0 < \ph_2(x^*)$, which implies that $x_1 < x^* < x_2$ via \cref{claim:case2}. Because $\ph_i(0) \ge 0$ for $i \in [2]$ (via \cref{lem:phi-prop}\ref{lem:us-0}), $\ph(0) \ge 0$ and so $\ph(x) > 0$ for $x\in (0,x^*)$, which implies $\ph'(x^*) \le 0$. Assume for the sake of contradiction that $\ph'(x^*) = 0$. This means $p_1\ph_1(x^*) + p_2\ph_2(x^*) = 0$ and $p_1\ph_1'(x^*) + p_2\ph_2'(x^*) = 0$, which implies that $\ph_1'(x^*)\ph_2(x^*) = \ph_1(x^*)\ph_2'(x^*)$. Thus $q'(x^*) = 0$ which contradicts \cref{claim:case2-d}.

  \ifieee
    \medskip
    \textit{The induction step.}
  \else
    \paragraph{The induction step.}
  \fi
  Suppose that $n \ge 3$. For convenience, we denote
  \[
    \ph(x; p) := \sumn p_i\ph_i(x),
  \]
  where $p := \spn$ and $\ph_i$ is the shorthand for $\ph_{a_i,b_i}$. For the sake of contradiction, we assume that $\ph(x; p) = 0$ has two solutions $x_1^* < x_2^*$ for $x$ in $(0,1]$. As $n \ge 3$, we use linear dependency to construct a nonzero vector $v := (v_1, \dots, v_n)$ such that
  \[
    \sumn v_i = 0, \quad \sumn v_i\ph_i(x_1^*) = 0,
  \]
  and we can assume in addition that
  \[
    \sumn v_i\ph_i(x_2^*) \ge 0.
  \]
  
  We can choose $t \ge 0$ such that $p + vt \in \Delta^{n-1}$ and $p_i + v_it = 0$ for some $i \in [n]$. Without loss of generality, assume that $p_n + v_nt = 0$. Set $p' := (p_1 + v_1t, \dots, p_{n-1} + v_{n-1}t)$. According to our choice of $v$, one can check that
  \[
    \ph(x_1^*; p') = 0, \quad \ph(x_2^*; p') \ge 0.
  \]
  Because $\ph_i(1) \le 0$ for $i \in [n]$ (via \cref{lem:phi-prop}\ref{lem:us-0}), we always have $\ph(1; p') \le 0$. The intermediate value theorem says that there exists a solution to $\ph(x; p') = 0$ in $[x_2, 1]$, hence the equation $\ph(x; p') = 0$ has at least two solutions in $(0,1]$, which contradicts with the inductive hypothesis because $\ph(\cdot; p')$ is a convex combination of $\ph_1, \dots, \ph_{n-1}$.

  For the ``moreover'' part, suppose that $x = x^*$ is the unique solution to $\ph(x; p) = 0$ in $(0,1]$, and assume for the sake of contradiction that $\ph'(x^*; p) \ge 0$. As $n \ge 3$, we use linear dependency to construct a nonzero vector $v := (v_1, \dots, v_n)$ such that
  \[
    \sumn v_i = 0, \quad \sumn v_i\ph_i(x^*) = 0,
  \]
  and we can assume in addition that
  \[
    \sumn v_i\ph_i'(x^*) \ge 0.
  \]
  Similarly, we can choose $t \ge 0$ such that $p + vt \in \Delta^{n-1}$ and $p_n + v_nt = 0$ without loss of generality. By setting $p' := (p_1 + v_1t, \dots, p_{n-1} + v_{n-1}t)$, one can check that $\ph(x^*; p') = 0$ and $\ph'(x^*; p') \ge 0$, which contradicts with the inductive hypothesis.
\end{proof}

Lastly we finish the proof of \cref{lem:unique-solution}.

\begin{proof}[Proof of \cref{lem:unique-solution}]
  The existence and uniqueness of the solution to the equation for $c$ has been done in \cref{lem:cab-prop}. Fix $\spn \in \Delta^{n-1}$, $\san$ and $\sbn \in [0,1]$ such that
  \[
    \sumn p_i\abs{a_i - b_i} > 0 \quad \text{or} \quad \sumn p_i\sqrt{a_i\aaa_i} > 1/4.
  \]
  It remains to show that the following equation has a unique solution:
  \[
    \ph(x) := \sumn p_i\ph_i(x) = 0, \quad x \in (0,1].
  \]
  
  In view of \cref{lem:uniqueness}, we only need to establish the existence of a solution. If $\sumn p_i\abs{a_i - b_i} > 0$, then $\ph(0) > 0$ and $\ph(1) \le 0$ (via \cref{lem:phi-prop}\ref{lem:us-0}), and so $\ph(x) = 0$ has a solution for $x$ in $(0,1]$ by the intermediate value theorem. Otherwise
  \ifieee
  \[
    p_i = 0 \text{ or }a_i = b_i \text{ for every }i \in [n], \text{ and } \sumn p_i\sqrt{a_i\aaa_i} > 1/4.
  \]
  \else
  \[
    p_i = 0 \text{ or }a_i = b_i \text{ for every }i \in [n] \quad \text{and} \quad \sumn p_i\sqrt{a_i\aaa_i} > 1/4.
  \]
  \fi
  For this special case, \cref{lem:phi-prop}\ref{lem:us-a} implies that
  \[
    \ph'(0) = \sumn p_i\ph_i'(0) = \sumn p_i\left(4\sqrt{a_i\aaai}-1\right) > 0.
  \]
  Because $\ph(0) = 0$, $\ph(1) \le 0$ (via \cref{lem:phi-prop}\ref{lem:us-0}) and $\ph'(0) > 0$, $\ph(x) = 0$ has a solution for $x$ in $(0,1]$.
\end{proof}

\section{Upper bound on \texorpdfstring{$R(\Of)$}{R(Of)}} \label{sec:proof-of-main}

In this section we investigate the function $L\colon D_n \to \mathbb{R}$, defined for every $\la \in (0,1/2)$ as follows:
\ifieee
\begin{multline*}
  L\spabn := \\
  \tfrac12 \sumn p_i\left(H_2(a_i, \aaai) + H_2(b_i, \bbi)\right) \\
  + \la\left(S_3(x^*) - \sumnphabcxs\right),
\end{multline*}
\else
\begin{multline*}
  L\spabn := \tfrac12 \sumn p_i\left(H_2(a_i, \aaai) + H_2(b_i, \bbi)\right) \\
  + \la\left(S_3(x^*) - \sumnphabcxs\right),
\end{multline*}
\fi
where $\cabi(x)$, $x^*$ and $D_n$ are defined as in \cref{lem:unique-solution}. In order to talk about the partial derivatives of $L$, we need the following properties of $\cabi$ and $x^*$.

\begin{proposition} \label{lem:diff}
  For every $\spn \in \Delta^{n-1}$, $\san$, $\sbn$ and $x \in [0,1]$, let $\cabi(x)$, $x^*$ and $D_n$ be defined as in \cref{lem:unique-solution}. Then
  \begin{enumerate}[label=(\alph*),nosep]
    \item $\cabi(x)$ is a continuously differentiable function of $x$ on $[0,1]$; \label{item:diff-cab-diff}
    \item $-\min(a_i\bbi,\aaai b_i) < \cabi(x) < \min(a_ib_i, \aaai\bbi)$ for $(a_i, b_i, x) \in (0,1)^3$, and $\cabi(x)$ is a continuously differentiable function of $(a_i, b_i, x)$ on $(0,1)^3$; \label{item:diff-cab-diff-3}
    \item $x^*$ is a continuously differentiable function of $\spabn$ on $D_n$; \label{item:diff-x-star-diff}
    \item $x^*$ converges to $0$ as $\spabn$ approaches a point outside $D_n$; \label{item:diff-x-star-0}
    \item if $p_i > 0$ for every $i \in [n]$, then $x^* = 1$ is equivalent to $a_i + b_i = 1$ for every $i \in [n]$. \label{item:diff-x-star-1}
  \end{enumerate}
\end{proposition}

\begin{proof}[Proof of \cref{lem:diff}\ref{item:diff-cab-diff}]
  \cref{lem:phi-prop}\ref{lem:us-c2} already shows that $\ph_{a_i,b_i}(x)$ is a continuously differentiable function of $x$ on $[0,1]$, where $\ph_{a_i,b_i}(x) = a_i\bbi + \aaai b_i + 2\cabi(x) - x$, and so is $\cabi(x)$.
\end{proof}

\begin{proof}[Proof of \cref{lem:diff}\ref{item:diff-cab-diff-3}]
  Recall that $c = \cabi(x)$ is the unique solution to
  \ifieee
  \begin{gather*}
    \begin{multlined}
      m(c, a_i, b_i, x) := (a_i\bbi + c)(\aaai b_i + c)(1-x)^2 \\
    - (a_ib_i - c)(\aaai\bbi - c)(2x)^2 = 0,
    \end{multlined} \\
    -\min(a_i\bbi, \aaai b_i) \le c \le \min(a_ib_i, \aaai\bbi).
  \end{gather*}
  \else
  \begin{gather*}
    m(c, a_i, b_i, x) := (a_i\bbi + c)(\aaai b_i + c)(1-x)^2 - (a_ib_i - c)(\aaai\bbi - c)(2x)^2 = 0, \\
    -\min(a_i\bbi, \aaai b_i) \le c \le \min(a_ib_i, \aaai\bbi).
  \end{gather*}
  \fi
  It is easy to check that $m(-\min(a_i\bbi, \aaai b_i), a_i, b_i, x) < 0$, and $m(\min(a_ib_i, \aaai\bbi), a_i, b_i, x) > 0$ whenever $(a_i, b_i, x) \in (0,1)^3$, hence
  \[
    -\min(a_i\bbi, \aaai b_i) < \cabi(x) < \min(a_ib_i, \aaai\bbi).
  \]
  One can compute
  \[
    \fpp{m}{c} = (a_i\bbi + \aaai b_i + 2c)(1-x)^2 + (a_ib_i + \aaai\bbi - 2c)(2x)^2,
  \]
  which is positive for $c \in (-\min(a_i\bbi, \aaai b_i), \min(a_ib_i, \aaai\bbi))$, and in particular at $c = \cabi(x)$ with $(a_i, b_i, x) \in (0,1)^3$. Thus the implicit function theorem implies that $\cabi$ is a continuously differentiable function of $(a_i, b_i, x)$ on $(0,1)^3$.
\end{proof}

\begin{proof}[Proof of \cref{lem:diff}\ref{item:diff-x-star-diff}]
  Since \cref{lem:uniqueness} says that $\ph'(x^*) \allowbreak \neq 0$, the implicit function theorem says that the unique solution $x = x^*$ to $\ph(x) = 0$ is a continuously differentiable function of $p_1, \dots, p_n$, $a_1, \dots, a_n$, and $b_1, \dots, b_n$.
\end{proof}

\begin{proof}[Proof of \cref{lem:diff}\ref{item:diff-x-star-0}]
  Pick $\eps > 0$ and a point $(p_1^*, \dots, \allowbreak p_n^*, a_1^*, \dots, a_n^*, b_1^*, \dots, b_n^*)$ outside $D_n$, that is,
  \ifieee
  \begin{multline*}
    p_i = 0 \text{ or }a_i^* = b_i^* \text{ for every }i \in [n], \\
    \text{and } \sumn p_i \sqrt{a_i^*\aaai^*} \le 1/4.
  \end{multline*}
  \else
  \[
    p_i = 0 \text{ or }a_i^* = b_i^* \text{ for every }i \in [n] \quad \text{and} \quad \sumn p_i \sqrt{a_i^*\aaai^*} \le 1/4.
  \]
  \fi
  \begin{claimm}
    There exists $x^- \in (0, \eps)$ such that
    \[
      \sumn p_i\ph_{a_i^*, b_i^*}(x^-) < 0.
    \]
  \end{claimm}
  \begin{claimproof}
    We break the argument into two cases.

    \ifieee
      \medskip
      \textit{Case 1:} $\sumn p_i\sqrt{a_i^*\aaai^*} < 1/4$.
    \else
      \paragraph{Case 1: $\sumn p_i\sqrt{a_i^*\aaai^*} < 1/4$.}
    \fi
      We conclude from \cref{lem:phi-prop}\ref{lem:us-a} that
    \[
      \sumn p_i\ph_{a_i^*, b_i^*}'(0) = \sumn p_i\left(4\sqrt{a_i^*\aaai^*}-1\right) < 0,
    \]
    from which the claim follows immediately.
    
    \ifieee
      \medskip
      \textit{Case 2:} $\sumn p_i\sqrt{a_i^*\aaai^*} = 1/4$.
    \else
      \paragraph{Case 2: $\sumn p_i\sqrt{a_i^*\aaai^*} = 1/4$.}
    \fi
    Without loss of generality, we may assume that $a_1^*, \dots, a_m^* \in (0,1)$ and $a_{m+1}^*, \dots, a_n^* \in \set{0,1}$ for some $m \in [n]$. We conclude from \cref{lem:phi-prop}\ref{lem:us-a} that
    \ifieee
    \begin{multline*}
      \sumn p_i\ph_{a_i^*, b_i^*}(x) = \sumn p_i\left(4\sqrt{a_i^*\aaai^*}-1\right)x \\
      + \sum_{i=1}^m p_i\left(4\sqrt{a_i^*\aaai^*}-4\right)x^2 + o(x^2).
    \end{multline*}
    \else
    \[
      \sumn p_i\ph_{a_i^*, b_i^*}(x) = \sumn p_i\left(4\sqrt{a_i^*\aaai^*}-1\right)x + \sum_{i=1}^m p_i\left(4\sqrt{a_i^*\aaai^*}-4\right)x^2 + o(x^2).
    \]
    \fi
    As $a\aaa \le 1/4$ for all $a \in (0,1)$ and $a\aaa = 0$ for $a \in \set{0,1}$, we estimate
    \[
      \sum_{i=1}^m p_i \ge \sum_{i=1}^m p_i\left(2\sqrt{a_i^*\aaai^*}\right) = \sumn p_i\left(2\sqrt{a_i^*\aaai^*}\right) = 1/2,
    \]
    and so
    \[
      \sumn p_i\ph_{a_i^*, b_i^*}(x) \le -x^2 + o(x^2),
    \]
    from which the claim follows immediately.
  \end{claimproof}
  Now take $x^- \in (0, \eps)$ according to the claim. When a point $\spabn \in D_n$ approaches $(p_1^*, \dots, p_n^*, a_1^*, \dots, a_n^*, b_1^*, \dots, b_n^*)$, because $\ph_{a_i, b_i}(x)$ depends continuously on $a_i$ and $b_i$, $\sumn p_i\ph_{a_i, b_i}(x^-)$ converges to $\sumn p_i\ph_{a_i^*, b_i^*}(x^-)$. Therefore there exists $\delta > 0$ such that
  \[
    \sumn p_i\ph_{a_i, b_i}(x^-) < 0
  \]
  for every $\spabn \in D_n$ in the $\delta$-neighborhood of $(p_1^*, \dots, p_n^*, a_1^*, \dots, a_n^*, \allowbreak b_1^*, \dots, b_n^*)$. We already know that the equation
  \[
    \ph(x) := \sumn p_i\ph_{a_i, b_i}(x) = 0, \quad x \in (0,1]
  \]
  has a unique solution $x = x^*$. Assume for the sake of contradiction that $x^* \in (x^-, 1]$. \cref{lem:uniqueness} asserts that $\ph'(x^*) < 0$. Because $\ph(x^-) < 0$, $\ph(x^*) = 0$ and $\ph'(x^*) < 0$, there exists a solution to $\ph(x) = 0$ in $(x^-, x^*)$, which contradicts the uniqueness of $x^*$. Therefore $x^* \in (0, x^-) \subseteq (0, \eps)$.
\end{proof}

\begin{proof}[Proof of \cref{lem:diff}\ref{item:diff-x-star-1}]
  Suppose that $p_i > 0$ for every $i \in [n]$. \cref{lem:phi-prop}\ref{lem:us-0} says that
  $\ph_{a_i, b_i}(1) \ge 0$ and equality holds if and only if $a_i + b_i = 1$. Therefore $\sumn p_i\ph_{a_i,b_i}(1) \ge 0$ and equality holds, that is $x^* = 1$, if and only if $a_i + b_i = 1$ for every $i \in [n]$.
\end{proof}

We also need the following elementary algebraic fact.

\begin{proposition} \label{lem:e1234}
  For every $\la \in (0,1/2)$ and $(e_1, e_2, e_3, e_4) \in \Delta^{3}$ with $e_1, e_2, e_3, e_4 > 0$, set $y := \sqrt{e_1e_4}/\sqrt{e_2e_3}$ and define $\rho\colon (0,1)\cup(1,\infty) \to \mathbb{R}$ by $\rho(r) := \sqrt{r}(r^\la-1)/(r-r^\la)$. If
  \[
    \frac{(e_1e_2)^\la}{e_1+e_2} = \frac{(e_3e_4)^\la}{e_3+e_4}, \quad
    \frac{(e_1e_3)^\la}{e_1+e_3} = \frac{(e_2e_4)^\la}{e_2+e_4},
  \]
  then one of the following holds.
  \begin{enumerate}[label=(\alph*),nosep]
    \item $e_1 = e_4, e_2 = e_3$; \label{item:e1234b}
    \item $e_2 = e_3$ and $y = \rho(r)$, where $r = e_1 / e_4 \neq 1$; \label{item:e1234d}
    \item $e_1 = e_4$ and $y = 1/\rho(r)$, where $r = e_2 / e_3 \neq 1$. \label{item:e1234c}
  \end{enumerate}
  Moreover the function $\rho(r)$ is increasing on $(0,1)$ and decreasing on $(1,\infty)$, $\rho(r)$ takes values in $(0,\la/(1-\la))$, and if $\rho(r) = \rho(r')$ then $r = r'$ or $rr' = 1$.
\end{proposition}

\begin{proof}
  By setting $r_1 := e_1 / e_4$ and $r_2 := e_2 / e_3$, we rewrite the equations as
  \ifieee
  \begin{equation} \label{eqn:e1234-lin}
    \begin{gathered}
      \left(r_1^\la r_2^\la - r_2\right)e_3 + \left(r_1^\la r_2^\la - r_1\right)e_4 = 0, \\
      \left(r_1^\la r_2 - r_2^\la\right)e_3 + \left(r_1^\la - r_1r_2^\la\right)e_4 = 0,
    \end{gathered}  
  \end{equation}
  \else
  \begin{equation} \label{eqn:e1234-lin}
    \left(r_1^\la r_2^\la - r_2\right)e_3 + \left(r_1^\la r_2^\la - r_1\right)e_4 = 0,
    \quad
    \left(r_1^\la r_2 - r_2^\la\right)e_3 + \left(r_1^\la - r_1r_2^\la\right)e_4 = 0,
  \end{equation}
  \fi
  which implies that
  \[
    \left(r_1^\la r_2^\la - r_2\right)\left(r_1^\la - r_1r_2^\la\right) - \left(r_1^\la r_2^\la - r_1\right)\left(r_1^\la r_2 - r_2^\la\right) = 0,
  \]
  which is equivalent to
  \[
    (r_1^{2\la} - r_1)r_2^\la(1 - r_2) + r_1^\la(1-r_1)(r_2^{2\la}-r_2) = 0.
  \]
  Since $r_i^{2\la} - r_i$ and $1-r_i$ have the same sign for $i \in [2]$, for the above equation to hold, it must be the case that $r_1 = 1$ or $r_2 = 1$. Case \ref{item:e1234b} corresponds to the case where $r_1 = 1$ and $r_2 = 1$. Case \ref{item:e1234d} corresponds to the case where $r_1 \neq 1$ and $r_2 = 1$. In this case, \eqref{eqn:e1234-lin} is equivalent to
  \[
    (r_1^\la - 1)e_3 + (r_1^\la - r_1)e_4 = 0,
  \]
  which implies that
  \[
    y = \frac{\sqrt{e_1e_4}}{\sqrt{e_2e_3}} = \frac{\sqrt{r_1}e_4}{e_3} = \frac{\sqrt{r_1}(r_1^\la-1)}{r_1-r_1^\la} = \rho(r_1).
  \]
  Case \ref{item:e1234c} corresponds to the case where $r_1 = 1$ and $r_2 \neq 1$, which can be dealt similarly.

  For the ``moreover'' part, one can compute
  \[
    \rho'(r) = \frac{r^\la((1-2\la)(1-r)-(r^\la-r^{1-\la}))}{2\sqrt{r}(r-r^\la)^2}.
  \]
  When $r \in (0,1)$, the mean value theorem states that there exists $s \in (r,1)$ such that
  \ifieee
  \begin{align*}
    \frac{r^\la-r^{1-\la}}{1-r} & = \frac{(1^{1-\la}-1^\la)-(r^{1-\la}-r^\la)}{1-r} \\
    & = (1-\la)s^{-\la}-\la s^{\la-1} \\
    & = (1-2\la)s^{-\la} - \la s^{\la-1}(1-s^{1-2\la}),
  \end{align*}
  \else
  \[
    \frac{r^\la-r^{1-\la}}{1-r} = \frac{(1^{1-\la}-1^\la)-(r^{1-\la}-r^\la)}{1-r} = (1-\la)s^{-\la}-\la s^{\la-1} = (1-2\la)s^{-\la} - \la s^{\la-1}(1-s^{1-2\la}),
  \]
  \fi
  which is clearly less than $1-2\la$ for $s \in (0,1)$. Thus $\rho'(r) > 0$ for $r \in (0,1)$, and so $\rho$ is increasing on $(0,1)$. One can check that $\rho(r) = \rho(1/r)$, and so $\rho$ is decreasing on $(1,\infty)$. Given $r' > 0$, the monotonicity of the function $\rho$ shows that the equation $\rho(r) = \rho(r')$ has at most two solutions, and so it has exactly two solutions $r = r'$ and $r = 1/r'$. Finally applying l'H\^{o}pital's rule results in $\lim_{r \to 1}\rho(r) = \la/(1-\la)$, and so $\rho(r)$ takes values in $(0, \la/(1-\la))$.
\end{proof}

\begin{proof}[Proof of \cref{thm:lag}]
  Suppose for a moment that $\spabn$ approaches a point $(p_1^*, \dots, p_n^*, a_1^*, \allowbreak \dots, a_n^*, b_1^*, \dots, b_n^*)$ outside $D_n$, that is
  \[
    a_i^* = b_i^* \text{ for }i \in [n] \text{ and }\sumn p_i \sqrt{a_i^*\aaai^*} \le 1/4.
  \]
  Then $a_i, b_i \to a_i^*$, and $x^* \to 0$ (via \cref{lem:diff}\ref{item:diff-x-star-0}), and so via \cref{lem:diff}\ref{item:diff-cab-diff}
  \[
    \cabi(x^*) \to \cabi(0) \stackrel{\eqref{eqn:cab-prop-c-01}}{=} -a_i^* \aaai^*,
  \]
  which yields
  \[
    S_4(a_i, b_i, \cabi(x^*)) \to H_2(a_i^*, \aaai^*).
  \]
  As a result, $L := L\spabn$ converges to
  \[
    \sumn p_iH_2(a_i^*, \aaai^*) + \la\left(1-\sumn p_iH_2(a_i^*, \aaai^*)\right) =: L^*.
  \]
  Using the fact $H_2(a, \aaa) \le 2\sqrt{a\aaa}$ for all $a \in [0,1]$, we know that
  \ifieee
  \begin{align*}
    L \to L^* & = \la + (1-\la)\sumn p_i H_2(a_i^*, \aaai^*) \\
    & \le \la + (1-\la)\sumn p_i\left(2\sqrt{a_i^*\aaai^*}\right) \\
    & \le \la + (1-\la)/2 = (1+\la)/2.
  \end{align*}
  \else
  \[
    L \to L^* = \la + (1-\la)\sumn p_i H_2(a_i^*, \aaai^*) \le \la + (1-\la)\sumn p_i\left(2\sqrt{a_i^*\aaai^*}\right) \le \la + (1-\la)/2 = (1+\la)/2.
  \]
  \fi
  Since the supremum of $L$ is greater than $(1+\la)/2$, $L$ must have a global maximum point in $D_n$.
    
  To characterize a global (or local) maximum point, we would like to compute $\partial L / \partial a_i$ and $\partial L / \partial b_i$. There are several caveats --- for example, the function $x \mapsto x\log x$ is not differentiable at $x = 0$. It turns out the differentiability issue could possibly arise only at the points in $D_n \setminus \wtdn$, where
  \ifieee
  \begin{multline*}
    \wtdn = D_n \cap \{\spabn \colon \\
    \forall i\ p_i > 0, a_i, b_i \in (0,1)\text{ and }\exists i\ a_i + b_i \neq 1\}.
  \end{multline*}
  \else
  \[
    \wtdn = D_n \cap \dset{\spabn}{\forall i\ p_i > 0, a_i, b_i \in (0,1)\text{ and }\exists i\ a_i + b_i \neq 1}.
  \]
  \fi
  Indeed because $x^* \in (0,1)$ for every $\spabn \in \wtdn$ (via \cref{lem:diff}\ref{item:diff-x-star-1}), $x^*$ is continuously differentiable on $\wtdn$. In addition, because $S_3(x)$ is continuously differentiable on $(0,1)$ and $S_4(a, b, c)$ is continuously differentiable for $(a,b,c)$ with $0 < a,b < 1$ and $-\min(a\bb, \aaa b)< c < \min(ab, \aaa\bb)$, \cref{lem:diff}\ref{item:diff-cab-diff-3} implies that $L$ is continuously differentiable on $\wtdn$. The following claim computes $\partial L / \partial a_i$ and $\partial L / \partial b_i$ on $\wtdn$. For brevity, denote
  \ifieee
  \begin{gather*}
    e_{1i} := a_ib_i - c_i^*, \quad e_{2i} := a_i\bbi + c_i^*, \\
    e_{3i} := \aaai b_i + c_i^*, \quad e_{4i} := \aaai\bbi - c_i^*, \\
    c_i^* := \cabi(x^*).
  \end{gather*}
  \else
  \[
    e_{1i} := a_ib_i - c_i^*, \quad e_{2i} := a_i\bbi + c_i^*, \quad e_{3i} := \aaai b_i + c_i^*, \quad e_{4i} := \aaai\bbi - c_i^*, \quad c_i^* := \cabi(x^*).
  \]
  \fi
  Note that the definition of $\cabi(x)$ implies that
  \begin{equation} \label{eqn:lag-ci-def-3}
    e_{1i}e_{4i}(2x^*)^2 = e_{2i}e_{3i}(1-x^*)^2.
  \end{equation}

  \setcounter{claim}{0}
  \begin{claim} \label{claim:partial-l}
    For every $\spabn \in \wtdn$ and $i \in [n]$,
    \begin{subequations}
      \begin{align}
        \fpp{L}{a_i} & = \frac{p_i}{2}\left(\log\frac{(e_{1i}e_{2i})^\la}{e_{1i}+e_{2i}} - \log\frac{(e_{3i}e_{4i})^\la}{e_{3i} + e_{4i}}\right), \label{eqn:partial-l-a} \\
        \fpp{L}{b_i} & = \frac{p_i}{2}\left(\log\frac{(e_{1i}e_{3i})^\la}{e_{1i}+e_{3i}} - \log\frac{(e_{2i}e_{4i})^\la}{e_{2i} + e_{4i}}\right). \label{eqn:partial-l-b}
      \end{align}        
    \end{subequations}
  \end{claim}
  \begin{claimproof}[Proof of \cref{claim:partial-l}]
    Fix $i \in [n]$. Because $a_i, b_i \in (0,1)$ and $x^* \in (0,1)$ (via \cref{lem:diff}\ref{item:diff-x-star-1}), \cref{lem:diff}\ref{item:diff-cab-diff-3} implies that $-\min(a_i\bbi, \aaai b_i) < c_i^* < \min(a_ib_i, \aaai\bbi)$, and so $e_{1i}, e_{2i}, e_{3i}, e_{4i} > 0$. We compute for $j \neq i$
    \ifieee
    \begin{align*}
      \fpp{S_4(a_j, b_j, c_j^*)}{a_i} & \stackrel{\phantom{\eqref{eqn:lag-ci-def-3}}}{=} \fpp{c_j^*}{a_i}\log\frac{e_{1j}e_{4j}}{e_{2j}e_{3j}} \\
      & \stackrel{\eqref{eqn:lag-ci-def-3}}{=} 2\left(\fpp{c_j^*}{a_i}\right)\log\frac{1-x^*}{2x^*} \\
      & \stackrel{\phantom{\eqref{eqn:lag-ci-def-3}}}{=} \fpp{(e_{2j}+e_{3j})}{a_i}\log\frac{1-x^*}{2x^*},
    \end{align*}
    \else
    \[
      \fpp{S_4(a_j, b_j, c_j^*)}{a_i} = \fpp{c_j^*}{a_i}\log\frac{e_{1j}e_{4j}}{e_{2j}e_{3j}}
      \stackrel{\eqref{eqn:lag-ci-def-3}}{=} 2\left(\fpp{c_j^*}{a_i}\right)\log\frac{1-x^*}{2x^*} = \fpp{(e_{2j}+e_{3j})}{a_i}\log\frac{1-x^*}{2x^*},
    \]
    \fi
    whereas for $j = i$
    \ifieee
    \begin{align*}
      \fpp{S_4(a_i, b_i, c_i^*)}{a_i} & \stackrel{\phantom{\eqref{eqn:lag-ci-def-3}}}{=} \left(\fpp{c_i^*}{a_i} - b_i\right) \log\frac{e_{1i}e_{4i}}{e_{2i}e_{3i}} + \log\frac{e_{4i}}{e_{2i}} \\
      & \stackrel{\eqref{eqn:lag-ci-def-3}}{=} 2\left(\fpp{c_i^*}{a_i}-b_i\right)\log\frac{1-x^*}{2x^*} + \log\frac{e_{4i}}{e_{2i}} \\
      & \stackrel{\phantom{\eqref{eqn:lag-ci-def-3}}}{=} \left(\fpp{(e_{2j}+e_{3j})}{a_i}-1\right)\log\frac{1-x^*}{2x^*} \\
      & \hspace{12.21em} + \log\frac{e_{4i}}{e_{2i}}.
    \end{align*}
    \else
    \begin{multline*}
      \fpp{S_4(a_i, b_i, c_i^*)}{a_i} = \left(\fpp{c_i^*}{a_i} - b_i\right) \log\frac{e_{1i}e_{4i}}{e_{2i}e_{3i}} + \log\frac{e_{4i}}{e_{2i}} \stackrel{\eqref{eqn:lag-ci-def-3}}{=} 2\left(\fpp{c_i^*}{a_i}-b_i\right)\log\frac{1-x^*}{2x^*} + \log\frac{e_{4i}}{e_{2i}} \\
      = \left(\fpp{(e_{2j}+e_{3j})}{a_i}-1\right)\log\frac{1-x^*}{2x^*} + \log\frac{e_{4i}}{e_{2i}}.
    \end{multline*}
    \fi
    Summing over $j \in [n]$ with weights $p_j$ yields
    \ifieee
    \begin{align*}
      & \phantom{{}={}} \fpp{\sum_{j=1}^n p_iS_4(a_j,b_j,c_j^*)}{a_i} \\
      & = \left(\sum_{j = 1}^n p_j\fpp{(e_{2j}+e_{3j})}{a_i} - p_i\right)\log\frac{1-x^*}{2x^*} + p_i\log\frac{e_{4i}}{e_{2i}}.
    \end{align*}
    \else
    \[
      \fpp{\sum_{j=1}^n p_iS_4(a_j,b_j,c_j^*)}{a_i} = \left(\sum_{j = 1}^n p_j\fpp{(e_{2j}+e_{3j})}{a_i} - p_i\right)\log\frac{1-x^*}{2x^*} + p_i\log\frac{e_{4i}}{e_{2i}}.
    \]
    \fi
    Recall that $x = x^*$ is a solution to the following equation:
    \[
      x = \sumnpabcx,
    \]
    which implies that
    \begin{equation} \label{eqn:lag-unique}
      x^* = \sumn p_i(a_i\bbi + \aaai b_i + 2c_i^*) = \sumn p_i(e_{2i}+e_{3i}).
    \end{equation}
    Taking the partial derivative of \eqref{eqn:lag-unique} with respect to $a_i$, we get that
    \[
      \fpp{x^*}{a_i} = \sum_{j=1}^n p_j\fpp{(e_{2i}+e_{3i})}{a_i},
    \]
    which gives the simplification
    \ifieee
    \begin{align*}
      & \phantom{{}={}} \fpp{\sum_{j=1}^n p_iS_4(a_j,b_j,c_j^*)}{a_i} \\
      & = \left(\fpp{x^*}{a_i}-p_i\right)\log\frac{1-x^*}{2x^*} + p_i\log\frac{e_{4i}}{e_{2i}}.
    \end{align*}
    \else
    \[
      \fpp{\sum_{j=1}^n p_iS_4(a_j,b_j,c_j^*)}{a_i} = \left(\fpp{x^*}{a_i}-p_i\right)\log\frac{1-x^*}{2x^*} + p_i\log\frac{e_{4i}}{e_{2i}}.
    \]
    \fi
    Therefore
    \ifieee
    \begin{align*}
      \frac{\partial L}{\partial a_i} & = \frac{p_i}{2}\left(\log\frac{\aaai}{a_i}\right) + \la\biggl(\frac{\partial x^*}{\partial a_i}\log\frac{1-x^*}{2x^*} \\
      & \phantom{{}={}} - \left(\fpp{x^*}{a_i}-p_i\right)\log\frac{1-x^*}{2x^*} - p_i\log\frac{e_{4i}}{e_{2i}}\biggr) \\
      & = p_i\left(\frac12\log\frac{\aaai}{a_i} + \la\left(\log\frac{1-x^*}{2x^*}-\log\frac{e_{4i}}{e_{2i}}\right)\right).
    \end{align*}
    \else
    \begin{align*}
      \frac{\partial L}{\partial a_i} & = \frac{p_i}{2}\left(\log\frac{\aaai}{a_i}\right) + \la\left(\frac{\partial x^*}{\partial a_i}\log\frac{1-x^*}{2x^*} - \left(\fpp{x^*}{a_i}-p_i\right)\log\frac{1-x^*}{2x^*} - p_i\log\frac{e_{4i}}{e_{2i}}\right) \\
      & = p_i\left(\frac12\log\frac{\aaai}{a_i} + \la\left(\log\frac{1-x^*}{2x^*}-\log\frac{e_{4i}}{e_{2i}}\right)\right).
    \end{align*}
    \fi
    Using \eqref{eqn:lag-ci-def-3} we obtain that
    \ifieee
    \begin{align*}
      \log\frac{1-x^*}{2x^*} - \log\frac{e_{4i}}{e_{2i}} & = \frac12\log\frac{e_{1i}e_{4i}}{e_{2i}e_{3i}} - \log\frac{e_{4i}}{e_{2i}} \\
      & = -\frac12\log\frac{e_{3i}e_{4i}}{e_{1i}e_{2i}}.
    \end{align*}
    \else
    \[
      \log\frac{1-x^*}{2x^*} - \log\frac{e_{4i}}{e_{2i}} = \frac12\log\frac{e_{1i}e_{4i}}{e_{2i}e_{3i}} - \log\frac{e_{4i}}{e_{2i}} = -\frac12\log\frac{e_{3i}e_{4i}}{e_{1i}e_{2i}}.
    \]
    \fi
    Together with the facts that $a_i = e_{1i}+e_{2i}$ and $\aaai = e_{3i}+e_{4i}$, we obtain \eqref{eqn:partial-l-a}, and we can obtain \eqref{eqn:partial-l-b} similarly.
  \end{claimproof}
  
  For the rest of the proof, suppose that
  \[
    z := \spabn
  \]
  is a global maximum point of $L$ in $D_n$. Since the value of $L$ does not change after we drop the summands with $p_i = 0$ in the sums appeared in $L$, without loss of generality, we may assume that $p_i > 0$ for every $i \in [n]$.

  We deal with the $x^* = 1$ case separately in the following claim.
  
  \begin{claim} \label{claim:x-star-is-1}
    If $p_i > 0$ for every $i \in [n]$ and $x^* = 1$, then $L(z) \le 1-\la$.
  \end{claim}

  \begin{claimproof}[Proof of \cref{claim:x-star-is-1}]
    Using $a_i + b_i = 1$ for every $i \in [n]$ from \cref{lem:diff}\ref{item:diff-cab-diff}, and $c_i^* = \cabi(1) = \min(a_i\bbi, \aaai b_i)$ from \eqref{eqn:cab-prop-c-01}, one can easily compute
    \[
      e_{1i} = 0, e_{2i} = a_i, e_{3i} = \aaai, e_{4i} = 0.
    \]
    Thus
    \[
      S_4(a_i, b_i, c_i^*) = H_4(e_{1i}, e_{2i}, e_{3i}, e_{4i}) = H_2(a_i, \aaai),
    \]
    which yields
    \ifieee
    \begin{align*}
      L(z) & = \sumn p_iH_2(a_i, \aaai) + \la\left(S_3(1) - \sumn p_iH_2(a_i, \aaai)\right) \\
      & = (1-\la)\sumn p_i H_2(a_i,\aaai) \le 1-\la.
    \end{align*}
    \else
    \begin{equation*}
      L(z) = \sumn p_iH_2(a_i, \aaai) + \la\left(S_3(1) - \sumn p_iH_2(a_i, \aaai)\right) = (1-\la)\sumn p_i H_2(a_i,\aaai) \le 1-\la. \qedhere
    \end{equation*}
    \fi
  \end{claimproof}
  
  Recall that $r = 2+\sqrt3$ and $x_1 = 1/(1+2r(r^{2\la})/(r^2-r^{2\la}))$. Because $x_1 \in (0,1)$, the upper bound $1 - \la$ in \cref{claim:x-star-is-1} is subsumed by $1-\la\log(1+x_1)$. Hereafter we may assume that $x^* \in (0,1)$ at the global maximum point $z$. Since we have assumed that $p_i > 0$ for every $i \in [n]$, \cref{lem:diff}\ref{item:diff-x-star-1} implies that $a_i + b_i \neq 1$ for some $i \in [n]$. We rule out the possibility that $z \in D_n \setminus \wtdn$ in the following claim.
  
  \begin{claim} \label{claim:z-is-in}
    If $p_i > 0$ for every $i \in [n]$ and $a_i + b_i \neq 1$ for some $i\in[n]$, then $z \in \wtdn$.
  \end{claim}

  \begin{claimproof}[Proof of \cref{claim:z-is-in}]
    Assume for the sake of contradiction that $z \not\in \wtdn$, that is $a_i \in \set{0,1}$ or $b_i \in \set{0,1}$ for some $i \in [n]$. We shall get a contradiction by nudging $(\san, \sbn)$ towards a point in $\wtdn$. For every $i \in [n]$, define
    \ifieee
    \begin{align*}
      u_i & = \begin{cases}
        1 - 2a_i & \text{if }a_i \in \set{0,1}, \\
        0 & \text{if }a_i \in (0,1),
      \end{cases} \\
      v_i & = \begin{cases}
        1 - 2b_i & \text{if }b_i \in \set{0,1}, \\
        0 & \text{if }b_i \in (0,1).
      \end{cases}
    \end{align*}
    \else
    \[
      u_i = \begin{cases}
        1 - 2a_i & \text{if }a_i \in \set{0,1}, \\
        0 & \text{if }a_i \in (0,1),
      \end{cases}
      \quad
      v_i = \begin{cases}
        1 - 2b_i & \text{if }b_i \in \set{0,1}, \\
        0 & \text{if }b_i \in (0,1).
      \end{cases}
    \]
    \fi
    It is easy to check that there exists $\eps > 0$ such that
    \ifieee
    \begin{multline*}
      z(t) := (p_1, \dots, p_n, a_1(t), \dots, a_n(t), b_1(t), \dots, b_n(t)), \\
      \text{where }a_i(t) = a_i + u_it \text{ and }b_i(t) = b_i + v_it,
    \end{multline*}
    \else
    \[
      z(t) := (p_1, \dots, p_n, a_1(t), \dots, a_n(t), b_1(t), \dots, b_n(t)), \text{where }a_i(t) = a_i + u_it \text{ and }b_i(t) = b_i + v_it,
    \]
    \fi
    is in $\wtdn$ for every $t \in (0, \eps)$. We break the computation of the limit of
    \[
      d_i(t) := u_i\fpp{L}{a_i}(z(t)) + v_i\fpp{L}{b_i}(z(t))
    \]
    as $t \to 0^+$ into three cases. For brevity,
    \ifieee
    \begin{gather*}
      e_{1i}(t) := a_i(t)b_i(t) - c_i^*(t), \quad e_{2i}(t) := a_i(t)\bbi(t) + c_i^*(t), \\
      e_{3i}(t) := \aaai(t)b_i(t) + c_i^*(t), \quad e_{4i}(t) := \aaai(t)\bbi(t) - c_i^*(t), \\
      c_i^*(t) := \cabi(x^*(z(t))).
    \end{gather*}
    \else
    \begin{gather*}
      e_{1i}(t) := a_i(t)b_i(t) - c_i^*(t), \quad e_{2i}(t) := a_i(t)\bbi(t) + c_i^*(t), \\
      e_{3i}(t) := \aaai(t)b_i(t) + c_i^*(t), \quad e_{4i}(t) := \aaai(t)\bbi(t) - c_i^*(t), \quad c_i^*(t) := \cabi(x^*(z(t))).
    \end{gather*}
    \fi

    \ifieee
      \medskip
      \textit{Case 1:} $a_i \in \set{0,1}$ and $b_i \in \set{0,1}$.
    \else
      \paragraph{Case 1: $a_i \in \set{0,1}$ and $b_i \in \set{0,1}$.}
    \fi
    We only demonstrate how to deal with the case where $(a_i, b_i) = (0,0)$ as other cases are similar. In this case, $u_i = v_i = 1$, and via \cref{claim:partial-l}, we have
    \ifieee
    \begin{gather*}
      \begin{multlined}
        u_i\fpp{L}{a_i}(z(t)) = \frac{p_i}{2}\biggl(\log\frac{(e_{1i}(t)e_{2i}(t))^\la}{e_{1i}(t)+e_{2i}(t)} \\
        - \log\frac{(e_{3i}(t)e_{4i}(t))^\la}{e_{3i}(t) + e_{4i}(t)}\biggr),
      \end{multlined} \\
      \begin{multlined}
        v_i\fpp{L}{b_i}(z(t)) = \frac{p_i}{2}\biggl(\log\frac{(e_{1i}(t)e_{3i}(t))^\la}{e_{1i}(t)+e_{3i}(t)} \\
        - \log\frac{(e_{2i}(t)e_{4i}(t))^\la}{e_{2i}(t) + e_{4i}(t)}\biggr).
      \end{multlined}
    \end{gather*}
    \else
    \begin{align*}
      u_i\fpp{L}{a_i}(z(t)) & = \frac{p_i}{2}\left(\log\frac{(e_{1i}(t)e_{2i}(t))^\la}{e_{1i}(t)+e_{2i}(t)} - \log\frac{(e_{3i}(t)e_{4i}(t))^\la}{e_{3i}(t) + e_{4i}(t)}\right), \\
      v_i\fpp{L}{b_i}(z(t)) & = \frac{p_i}{2}\left(\log\frac{(e_{1i}(t)e_{3i}(t))^\la}{e_{1i}(t)+e_{3i}(t)} - \log\frac{(e_{2i}(t)e_{4i}(t))^\la}{e_{2i}(t) + e_{4i}(t)}\right).
    \end{align*}
    \fi
    Since $a_i(t) = b_i(t) = t$, the definition of $\cabi(x)$ implies that $-t(1-t)\le c_i^*(t) \le t^2$. Thus $c_i^*(t) = O(t)$, and so $e_{4i}(t) = (1-t)^2 - c_i^*(t) \sim 1$. Because $e_{2i}(t) = t(1-t) + c_i^* = e_{3i}(t)$, \eqref{eqn:lag-ci-def-3} implies
    \ifieee
    \begin{multline*}
      \frac{\sqrt{e_{1i}(t)}}{e_{2i}(t)} \sim \frac{\sqrt{e_{1i}(t)e_{4i}(t)}}{\sqrt{e_{2i}(t)e_{3i}(t)}} = \frac{1-x^*(z(t))}{2x^*(z(t))} \\
      \sim \frac{1-x^*}{2x^*} =: y, \text{ as }t \to 0^+.
    \end{multline*}
    \else
    \[
      \frac{\sqrt{e_{1i}(t)}}{e_{2i}(t)} \sim \frac{\sqrt{e_{1i}(t)e_{4i}(t)}}{\sqrt{e_{2i}(t)e_{3i}(t)}} = \frac{1-x^*(z(t))}{2x^*(z(t))} \sim \frac{1-x^*}{2x^*} =: y, \text{ as }t \to 0^+.
    \]
    \fi
    Together with the fact that $e_{1i}(t) + e_{2i}(t) = a_i(t) = t$, we can solve
    \[
      e_{1i}(t) \sim y^2t^2, \quad e_{2i}(t) = e_{3i}(t) \sim t, \quad e_{4i}(t) \sim 1,
    \]
    which implies that
    \[
      u_i\fpp{L}{a_i}(z(t)) = v_i\fpp{L}{b_i}(z(t)) \sim \frac{p_i}{2}\log t^{2\la-1},
    \]
    and so $\lim_{t\to 0^+} d_i(t) = \infty$.

    \ifieee
      \medskip
      \textit{Case 2:} $(a_i, b_i) \in \set{0,1} \times (0,1) \cup (0,1) \times \set{0,1}$.
    \else
      \paragraph{Case 2: $(a_i, b_i) \in \set{0,1} \times (0,1) \cup (0,1) \times \set{0,1}$.}
    \fi
    We only demonstrate how to deal with the case where $a_i = 0$ and $b_i \in (0,1)$ as other cases are similar. In this case, $u_i = 1$ and $v_i = 0$, and via \cref{claim:partial-l}, we have
    \ifieee
    \begin{align*}
      d_i(t) & = u_i\fpp{L}{a_i}(z(t)) \\
      & = \frac{p_i}{2}\left(\log\frac{(e_{1i}(t)e_{2i}(t))^\la}{e_{1i}(t)+e_{2i}(t)} - \log\frac{(e_{3i}(t)e_{4i}(t))^\la}{e_{3i}(t) + e_{4i}(t)}\right).
    \end{align*}
    \else
    \[
      d_i(t) = u_i\fpp{L}{a_i}(z(t)) = \frac{p_i}{2}\left(\log\frac{(e_{1i}(t)e_{2i}(t))^\la}{e_{1i}(t)+e_{2i}(t)} - \log\frac{(e_{3i}(t)e_{4i}(t))^\la}{e_{3i}(t) + e_{4i}(t)}\right).
    \]
    \fi
    Since $a_i(t) = t$ and $b_i(t) = b_i$, the definition of $\cabi(x)$ implies that $-t\bbi \le c_i^* \le tb_i$. Thus $c_i^* = O(t)$, and so $e_{3i} = (1-t)b_i + c_i^* = b_i + o(1)$ and $e_{4i} = (1-t)\bbi - c_i^* = \bbi + o(1)$. Therefore
    \ifieee
    \begin{multline*}
      \frac{\sqrt{e_{1i}\bbi}}{\sqrt{e_{2i}b_i}} \sim \frac{\sqrt{e_{1i}(t)e_{4i}(t)}}{\sqrt{e_{2i}(t)e_{3i}(t)}} = \frac{1-x^*(z(t))}{2x^*(z(t))} \\
      \sim \frac{1-x^*}{2x^*} =: y, \text{ as }t \to 0^+.
    \end{multline*}
    \else
    \[
      \frac{\sqrt{e_{1i}\bbi}}{\sqrt{e_{2i}b_i}} \sim \frac{\sqrt{e_{1i}(t)e_{4i}(t)}}{\sqrt{e_{2i}(t)e_{3i}(t)}} = \frac{1-x^*(z(t))}{2x^*(z(t))} \sim \frac{1-x^*}{2x^*} =: y, \text{ as }t \to 0^+.
    \]
    \fi
    Together with the fact that $e_{1i}(t) + e_{2i}(t) = a_i(t) = t$, we can solve
    \[
      e_{1i}(t) \sim \left(\frac{y^2b_i^2}{y^2b_i^2 + \bbi^2}\right)t, \quad e_{2i}(t) \sim \left(\frac{\bbi^2}{y^2b_i^2 + \bbi^2}\right)t,
    \]
    which implies that
    \[
      u_i\fpp{L}{a_i}(z(t)) \sim \frac{p_i}{2}\log t^{2\la-1},
    \]
    and so $\lim_{t \to 0^+}d_i(t) = \infty$.

    \ifieee
      \medskip
      \textit{Case 3:} $(a_i, b_i) \in (0,1)^2$.
    \else
      \paragraph{Case 3: $(a_i, b_i) \in (0,1)^2$.}
    \fi
    In this case $u_i = v_i = 0$, and so $d_i(t) = 0$.

    \ifieee
    \medskip
    \else
    \bigskip
    \fi

    Because $a_i \in \set{0,1}$ and $b_i \in \set{0,1}$ for some $i \in [n]$, we know from the computation done in the above three cases that
    \[
      \frac{\dd}{\dd t}L(z(t)) = \sumn d_i(t) \to \infty, \text{ as }t \to 0^+,
    \]
    in particular there exists $\tilde{\eps} \in (0,\eps)$ such that $\dd L(z(t))/ \dd t > 0$ for every $t \in (0, \tilde\eps)$. However because $L(z(0)) \ge L(z(\tilde\eps))$, the mean value theorem says that there exists $t \in (0, \tilde\eps)$ such that $\dd L(z(t))/ \dd t \le 0$, which yields a contradiction.
  \end{claimproof}
  
  According to \cref{claim:z-is-in}, the global maximum point $z$ is in $\wtdn$, that is $a_i, b_i \in (0,1)$ for $i \in [n]$, and $a_1 + b_1 \neq 1$ without loss of generality, which implies that $e_{11} \neq e_{41}$ because $e_{11} - e_{41} = a_1 + b_1 - 1$. In addition, for every $i \in [n]$, $\partial L/\partial a_i = 0$ and $\partial L/\partial b_i = 0$, which imply via \cref{claim:partial-l} that
  \[
    \frac{(e_{1i}e_{2i})^\la}{e_{1i}+e_{2i}} = \frac{(e_{3i}e_{4i})^\la}{e_{3i}+e_{4i}}, \quad \frac{(e_{1i}e_{3i})^\la}{e_{1i}+e_{3i}} = \frac{(e_{2i}e_{4i})^\la}{e_{2i}+e_{4i}}.
  \]
  Set
  \[
    y := \frac{1-x^*}{2x^*} \quad \text{and} \quad \rho(r) := \frac{\sqrt{r}(r^\la-1)}{r-r^\la}.
  \]
  Because $x^* \in (0,1)$, we know that $y > 0$. We also know from \eqref{eqn:lag-ci-def-3} that
  \begin{equation} \label{eqn:lag-thm-y-def}
    y = \sqrt{e_{1i}e_{4i}}/\sqrt{e_{2i}e_{3i}} \text{ for every } i \in [n].
  \end{equation}
  Define $r_i = e_{1i} / e_{4i}$ for all $i$. By \cref{lem:e1234} we have $y = \rho(r_1) \in (0, \la/(1-\la)) \subseteq (0,1)$ because $e_{11} \neq e_{41}$. This allows us to rule out case~\ref{item:e1234c} in \cref{lem:e1234}. For each $i \in [n]$, the remaining two cases \ref{item:e1234b} and \ref{item:e1234d} could apply to $(e_{1i}, e_{2i}, e_{3i}, e_{4i})$, that is, one of the following two holds.
  \begin{enumerate}[label=(\alph*),nosep]
    \item $e_{1i} = e_{4i}$, $e_{2i} = e_{3i}$. Using \eqref{eqn:lag-thm-y-def}, $e_{1i} = e_{4i}$, $e_{2i} = e_{3i}$ and $e_{1i} + e_{2i} + e_{3i} + e_{4i} = 1$, we can solve
    \[
      e_{1i} = e_{4i} = \frac{y}{2y+2}, \quad e_{2i} = e_{3i} = \frac{1}{2y+2}.
    \]
    Thus $a_i = e_{1i} + e_{2i} = 1/2$ and $b_i = e_{1i} + e_{3i} = 1/2$, and so
    \ifieee
    \begin{gather*}
      H_2(a_i, \aaai) = H_2(b_i, \bbi) = 1, \\
      c_i^* = a_ib_i - e_{1i} = \frac{-y+1}{4y+4} = \frac{3x^*-1}{4x^*+4}.
    \end{gather*}
    \else
    \[
      H_2(a_i, \aaai) = H_2(b_i, \bbi) = 1, \quad c_i^* = a_ib_i - e_{1i} = \frac{-y+1}{4y+4} = \frac{3x^*-1}{4x^*+4}.
    \]
    \fi
    \item $e_{2i} = e_{3i}$, $y = \rho(r_i)$ and $r_i \neq 1$. Since $y = \rho(r_1)$, \cref{lem:e1234} implies that $r_i \in \set{r_1, 1/r_1}$. Using \eqref{eqn:lag-thm-y-def}, $e_{1i}/e_{4i} = r_i$, $e_{2i} = e_{3i}$ and $e_{1i} + e_{2i} + e_{3i} + e_{4i} = 1$, we can solve
    \[
      (e_{1i},e_{2i},e_{3i},e_{4i}) = \begin{cases}
        (e_{11},e_{21},e_{31},e_{41}) & \text{if }r_i = r_1, \\
        (e_{41},e_{31},e_{21},e_{11}) & \text{if }r_i = 1/r_1.
      \end{cases}
    \]
    In either case we always have
    \ifieee
    \begin{gather*}
      H_2(a_i, \aaai) = H_2(b_i,\bbi) = H_2(a_1, \aaa_1), \\
      \begin{aligned}
        S_4(a_i, b_i, c_i^*) & = H_4(e_{1i},e_{2i},e_{3i},e_{4i}) \\
        & = H_4(e_{11},e_{21},e_{31},e_{41}) = S_4(a_1, a_1, c_1^*),
      \end{aligned} \\
      e_{2i} + e_{3i} = e_{21} + e_{31} = 2a_1\aaa_1 + 2c_1^*.
    \end{gather*}
    \else
    \begin{gather*}
      H_2(a_i, \aaai) = H_2(b_i,\bbi) = H_2(a_1, \aaa_1), \\
      S_4(a_i, b_i, c_i^*) = H_4(e_{1i},e_{2i},e_{3i},e_{4i}) = H_4(e_{11},e_{21},e_{31},e_{41}) = S_4(a_1, a_1, c_1^*), \\
      e_{2i} + e_{3i} = e_{21} + e_{31} = 2a_1\aaa_1 + 2c_1^*.
    \end{gather*}
    \fi
  \end{enumerate}
  
  After combining similar terms, $L(z)$ and \eqref{eqn:lag-unique} become
  \begin{subequations}
    \ifieee
    \begin{gather}
      \begin{aligned}
        L(z) & = (1-p) + pH_2(a_1, \aaa_1) \\
        & \phantom{{}={}} + \la\left(S_3(x^*) - (1-p)S_{40} - pS_{41}\right),
      \end{aligned} \\
      x^* = (1-p)(2a_0\aaa_0 + 2c_0^*) + p(2a_1\aaa_1 + 2c_1^*), \label{eqn:x-half}
    \end{gather}
    \else
      \begin{gather} 
      L(z) = (1-p) + pH_2(a_1, \aaa_1) + \la\left(S_3(x^*) - (1-p)S_{40} - pS_{41}\right), \\
      x^* = (1-p)(2a_0\aaa_0 + 2c_0^*) + p(2a_1\aaa_1 + 2c_1^*), \label{eqn:x-half}
    \end{gather}
    \fi
  \end{subequations}
  where
  \ifieee
  \begin{gather*}
    p = \sum_{i\colon r_i \neq 1}p_i, \\
    S_{40} = S_4(a_0, a_0, c_0^*), \quad S_{41} = S_4(a_1, a_1, c_1^*), \\
    a_0 = \frac12, \quad c_0^* = \frac{3x^*-1}{4x^*+4}.
  \end{gather*}
  \else
  \[
    p = \sum_{i\colon r_i \neq 1}p_i, \quad S_{40} = S_4(a_0, a_0, c_0^*), \quad S_{41} = S_4(a_1, a_1, c_1^*), \quad a_0 = \frac12, \quad c_0^* = \frac{3x^*-1}{4x^*+4}.
  \]
  \fi
  Because $r_1 \neq 1$ we know that $p \ge p_1 > 0$. We break the rest of the argument into two cases.
  
  \ifieee
    \medskip
    \textit{Case 1:} $p \in (0,1)$.
  \else
    \paragraph{Case 1: $p \in (0,1)$.}
  \fi
  We obtain for $i \in \set{0,1}$,
  \[
    \fpp{H_{4i}}{p} = \fpp{c_i^*}{p}\log\frac{(a_i^2-c_i)^2}{(a_i\aaai - c_i)^2} = \fpp{c_i^*}{p}\left(2\log\frac{1-x^*}{2x^*}\right).
  \]
  Therefore
  \ifieee
  \begin{align*}
    \fpp{L}{p} & = -1 + H_2(a_1, \aaa_1) \\
    & \hspace{1.1em} + \la\biggl(\fpp{x^*}{p} \log\frac{1-x^*}{2x^*} + S_{40} \\
    & \hspace{3em} - (1-p)\fpp{c_0^*}{p}\left(2\log\frac{1-x^*}{2x^*}\right) \\
    & \hspace{3em} - S_{41} - p\fpp{c_1^*}{p}\left(2\log\frac{1-x}{2x}\right)\biggr).
  \end{align*}
  \else
  \begin{multline*}
    \fpp{L}{p} = -1 + H_2(a_1, \aaa_1) + \la\left(\fpp{x^*}{p} \log\frac{1-x^*}{2x^*} + S_{40} - (1-p)\fpp{c_0^*}{p}\left(2\log\frac{1-x^*}{2x^*}\right) \right. \\
    \left. - S_{41} - p\fpp{c_1^*}{p}\left(2\log\frac{1-x}{2x}\right)\right).
  \end{multline*}
  \fi
  Taking the partial derivative of \eqref{eqn:x-half} with respect to $p$, we get that
  \ifieee
  \begin{multline*}
    \fpp{x^*}{p} = - (2a_0\aaa_0 + 2c_0^*) + (1-p)\left(2\fpp{c_0^*}{p}\right) \\
    + (2a_1\aaa_1 + 2c_1^*) + p\left(2\fpp{c_1^*}{p}\right),
  \end{multline*}
  \else
  \[
    \fpp{x^*}{p} = - (2a_0\aaa_0 + 2c_0^*) + (1-p)\left(2\fpp{c_0^*}{p}\right) + (2a_1\aaa_1 + 2c_1^*) + p\left(2\fpp{c_1^*}{p}\right),
  \]
  \fi
  which leads to the simplification
  \ifieee
  \begin{multline*}
    \fpp{L}{p} = -1 + H_2(a_1,\aaa_1) + \la\biggl((- 2a_0\aaa_0 - 2c_0^* \\
    + 2a_1\aaa_1 + 2c_1^*)\log\frac{1-x^*}{2x^*} + S_{40} - S_{41}\biggr).
  \end{multline*}
  \else
  \[
    \fpp{L}{p} = -1 + H_2(a_1,\aaa_1) + \la\left(\left(- 2a_0\aaa_0 - 2c_0^* + 2a_1\aaa_1 + 2c_1^*\right)\log\frac{1-x^*}{2x^*} + S_{40} - S_{41}\right).
  \]
  \fi
  As $\partial L / \partial p = 0$ at the global maximum point $z$ with $p \in (0,1)$, we get
  \ifieee
  \begin{align*}
    L(z) & \stackrel{\phantom{\eqref{eqn:x-half}}}{=} L(z) - p\fpp{L}{p}(z) \\
    & \stackrel{\phantom{\eqref{eqn:x-half}}}{=} 1 + \la\biggl(S_3(x^*) - S_{40} - p(-2a_0\aaa_0 \\
    & \phantom{\stackrel{\phantom{\eqref{eqn:x-half}}}{=}} \hspace{3em} - 2c_0^* + 2a_1\aaa_1 + 2c_1^*)\log\frac{1-x^*}{2x^*}\biggr) \\
    & \stackrel{\eqref{eqn:x-half}}{=} 1+\la\biggl(S_3(x^*) - S_{40} \\
    & \phantom{\stackrel{\phantom{\eqref{eqn:x-half}}}{=}} \hspace{3em} - (x^* - 2a_0\aaa_0 - 2c_0^*)\log\frac{1-x^*}{2x^*}\biggr).
  \end{align*}
  \else
  \begin{align*}
    L(z) & \stackrel{\phantom{\eqref{eqn:x-half}}}{=} L(z) - p\fpp{L}{p}(z) \\
    & \stackrel{\phantom{\eqref{eqn:x-half}}}{=} 1 + \la\left(S_3(x^*) - S_{40} - p\left(-2a_0\aaa_0 - 2c_0^* + 2a_1\aaa_1 + 2c_1^*\right)\log\frac{1-x^*}{2x^*}\right) \\
    & \stackrel{\eqref{eqn:x-half}}{=} 1+\la\left(S_3(x^*) - S_{40} - \left(x^* - 2a_0\aaa_0 - 2c_0^*\right)\log\frac{1-x^*}{2x^*}\right).
  \end{align*}
  \fi
  Because
  \ifieee
  \begin{align*}
      S_3(x^*) & = H_3{\left(\frac{1-x^*}{2}, x^*, \frac{1-x^*}{2}\right)} \\
      & = H_2(x^*, 1-x^*) + 1 - x^*, \\
      S_{40} & = H_4{\left(a_0^2-c_0^*, a_0\aaa_0+c_0^*, \aaa_0 a_0+c_0^*, \aaa_0^2-c_0^*\right)} \\
      & = H_2{\left(\frac{1-x^*}{1+x^*},\frac{2x^*}{1+x^*}\right)}+1,
  \end{align*}
  \else
  \begin{gather*}
    S_3(x^*) = H_3{\left(\frac{1-x^*}{2}, x^*, \frac{1-x^*}{2}\right)} = H_2(x^*, 1-x^*) + 1 - x^*, \\
    S_{40} = H_4{\left(a_0^2-c_0^*, a_0\aaa_0+c_0^*, \aaa_0 a_0+c_0^*, \aaa_0^2-c_0^*\right)} = H_2{\left(\frac{1-x^*}{1+x^*},\frac{2x^*}{1+x^*}\right)}+1,
  \end{gather*}
  \fi
  we simplify
  \ifieee
  \begin{align*}
    L(z) & = 1 + \la\biggl(H_2(x^*,1-x^*)+1-x^* \\
    & \hspace{3.1em} - H_2{\biggl(\frac{1-x^*}{1+x^*},\frac{2x^*}{1+x^*}\biggr)} - 1 \\
    & \hspace{3.1em} - \left(x^*-\frac12-\frac{3x^*-1}{2x^*+2}\right)\log\frac{1-x^*}{2x^*}\biggr) \\
    & = 1 - \la \log(1+x^*).
  \end{align*}
  \else
  \begin{multline*}
    L(z) = 1 + \la\left(H_2(x^*,1-x^*)+1-x^* - H_2{\left(\frac{1-x^*}{1+x^*},\frac{2x^*}{1+x^*}\right)}-1 \right. \\ \left. - \left(x^*-\frac12-\frac{3x^*-1}{2x^*+2}\right)\log\frac{1-x^*}{2x^*}\right) = 1 - \la \log(1+x^*).
  \end{multline*}
  \fi
  
  It suffices to show that $x^* > 1/(1+2\rho(r^2))$, where $r = 2 + \sqrt3$. Since $x^* \in (0,1)$, we know that $2a_0\aaa_0+2c_0^* = 2x^*/(1+x^*) > x^*$. We know from \eqref{eqn:x-half} that $x^*$ is a convex combination of $2a_0\aaa_0 + 2c_0^*$ and $2a_1\aaa_1 + 2c_1^*$, and so
  \[
    e_{21} + e_{31} = 2a_1\aaa_1 + 2c_1^* < x^* = \frac{1}{1+2y}.
  \]
  Using \eqref{eqn:lag-thm-y-def}, $e_{11}/e_{41} = r_1$, $e_{21} = e_{31}$ and $e_{11} + e_{21} + e_{31} + e_{41} = 1$, we can solve
  \[
    e_{21} + e_{31} = \frac{2\sqrt{r_1}}{(r_1+1)y+2\sqrt{r_1}},
  \]
  which together with the previous inequality gives
  \[
    \frac{2\sqrt{r_1}}{(r_1+1)y+2\sqrt{r_1}} < \frac{1}{1+2y},
  \]
  which is equivalent to
  \[
    4\sqrt{r_1} < r_1 + 1,
  \]
  and so $r_1 < 1/r^2$ or $r_1 > r^2$. From \cref{lem:e1234} we get that $\rho(r_1) < \rho(r^2)$, and so $x^* = 1/(1+2y) = 1/(1+2\rho(r_1)) > 1/(1+2\rho(r^2))$.

  \ifieee
    \medskip
    \textit{Case 2:} $p = 1$.
  \else
    \paragraph{Case 2: $p = 1$.}
  \fi
  Because $e_{11}/e_{41} = r_1$, $e_{21} = e_{31}$, $x^* = 2a_1\aaa_1 + 2c_1 = e_{21} + e_{31}$, and $e_{11} + e_{21} + e_{31} + e_{41} = 1$, we can solve
  \[
    e_{11} = \frac{r_1(1-x^*)}{r_1 + 1}, \quad e_{21} = e_{31} = \frac{x^*}{2}, \quad e_{41} = \frac{1-x^*}{r_1 + 1}.
  \]
  We know from \eqref{eqn:lag-thm-y-def} that
  \[
    \frac{\sqrt{r_1}(1-x^*)}{(r_1+1)(x^*/2)} = y = \frac{1-x^*}{2x^*},
  \]
  which is equivalent to
  \[
    r_1 + 1 = 4\sqrt{r_1}.
  \]
  Thus $r_1 \in \set{r^2, 1/r^2}$, where $r = 2+\sqrt3$, and
  \[
    e_{11} = \frac{r(1-x^*)}{4}, \quad e_{41} = \frac{1-x^*}{4r}.
  \]
  Therefore
  \ifieee
  \begin{align*}
    L(z) & = H_2(e_{11}+e_{21}, e_{31}+e_{41}) \\
    & \phantom{{}={}} + \la(S_3(x^*) - H_4(e_{11}, e_{21}, e_{31}, e_{41})) \\
    & = H_2{\left(\tfrac{r(1-x^*)}{4} + \tfrac{x^*}{2}, \tfrac{x^*}{2} + \tfrac{1-x^*}{4r}\right)} \\
    & \phantom{{}={}} + \la\left(H_3{\left(\tfrac{1-x^*}{2},x^*,\tfrac{1-x^*}{2}\right)} \right. \\
    & \hspace{2.8em} \left. - H_4{\left(\tfrac{r(1-x^*)}{4}, \tfrac{x^*}{2}, \tfrac{x^*}{2}, \tfrac{1-x^*}{4r}\right)}\right) \\
    & = H_2{\left(\tfrac{r}{4} - \left(\tfrac{r}{4}-\tfrac12\right)x^*, \tfrac{1}{4r} + \left(\tfrac12-\tfrac{1}{4r}\right)x^*\right)} \\
    & \phantom{{}={}} + \la\left(-1 + \left(\tfrac{r}{4} - \tfrac{1}{4r}\right) (\log r) (1-x^*)\right).
  \end{align*}
  \else
  \begin{align*}
    L(z) & = H_2(e_{11}+e_{21}, e_{31}+e_{41}) + \la(S_3(x^*) - H_4(e_{11}, e_{21}, e_{31}, e_{41})) \\
    & = H_2{\left(\tfrac{r(1-x^*)}{4} + \tfrac{x^*}{2}, \tfrac{x^*}{2} + \tfrac{1-x^*}{4r}\right)} + \la\left(H_3{\left(\tfrac{1-x^*}{2},x^*,\tfrac{1-x^*}{2}\right)} - H_4{\left(\tfrac{r(1-x^*)}{4}, \tfrac{x^*}{2}, \tfrac{x^*}{2}, \tfrac{1-x^*}{4r}\right)}\right) \\
    & = H_2{\left(\tfrac{r}{4} - \left(\tfrac{r}{4}-\tfrac12\right)x^*, \tfrac{1}{4r} + \left(\tfrac12-\tfrac{1}{4r}\right)x^*\right)} + \la\left(-1 + \left(\tfrac{r}{4} - \tfrac{1}{4r}\right) (\log r) (1-x^*)\right).
  \end{align*}
  \fi
  Finally notice that $r/4 - 1/2 = 1/2-1/(4r) = \sqrt{3}/4$, $r/4-1/(4r) = \sqrt{3}/2$, and $x^* = 1/(1+2y) = 1/(1+2\rho(r^2))$.
\end{proof}

\section{An open problem} \label{sec:open}

At the moment, the Dueck zero-error capacity region $\Of$ of the binary adder channel is far from fully understood. \cref{thm:exact} gives the exact value of
\[
  \sup\dset{\frac12 R_1 + \frac12 R_2}{(R_1, R_2) \in \Of}.
\]
In general, the \emph{weighted} average zero-error capacity of the binary adder channel with complete feedback is defined by
\[
  R(\Of; c_1, c_2) := \sup\dset{c_1 R_1 + c_2 R_2}{(R_1, R_2) \in \Of},
\]
for every $(c_1, c_2) \in \Delta^1$. One can derive from \cref{thm:dueck-rv} that $R(\Of; c_1, c_2)$ is given by the optimum of an optimization problem with $n \in \mathbb{N}$ as a variable. Our numerical experiments suggest that $n=1$ readily gives the optimal value.

\begin{conjecture}
  For every $(c_1, c_2) \in \Delta^1$, the weighted average zero-error capacity $R(\Of; c_1, c_2)$ of the binary adder channel with complete feedback equals the optimum of the following optimization problem.
  \ifieee
  \begin{align*}
    \text{Maximize: } & c_1 H_2(a, \aaa) + c_2 H_2(b,\bb), \\
  \text{subject to: } & S_3\left(a\bb + \aaa b + 2c\right) \ge S_4(a, b, c) \\
  & \text{for every }c \in [-\min(a\bb, \aaa b), \min(ab, \aaa\bb)], \\
  & a \in [0,1], b \in [0,1].
  \end{align*}
  \else
  \begin{align*}
    \text{Maximize: } & c_1 H_2(a, \aaa) + c_2 H_2(b,\bb), \\
  \text{subject to: } & S_3\left(a\bb + \aaa b + 2c\right) \ge S_4(a, b, c) \text{ for every }c \in [-\min(a\bb, \aaa b), \min(ab, \aaa\bb)], \\
  & a \in [0,1], b \in [0,1].
  \end{align*}
  \fi
\end{conjecture}

Once the conjecture is established, it would be interesting to compute the optimal value of the above optimization problem. The explicit formula of $R(\Of; c_1, c_2)$ for $(c_1, c_2) \in \Delta^1$ would in turn yield a complete description of the Dueck zero-error capacity region $\Of$.

\section*{Acknowledgements}

This research was conducted while the first and second authors were participants and the third author was a mentor in the PRIMES-USA program of the MIT Mathematics Department. We thank Prof.\ Pavel Etingof, Dr.\ Slava Gerovitch and Dr.\ Tanya Khovanova for their role in advising the program. We thank Nikita Polyanskii and Rahul Thomas for inspirational discussion at the early stage of the project. We are very thankful to the referees for extensive comments on the manuscript.

\ifieee
\appendices
\else
\appendix
\fi

\section{Proof of \texorpdfstring{\cref{lem:phi-prop}}{Proposition 13}} \label{sec:f-prime-over-f}

Throughout this section, we fix $(a, b) \in [0,1]^2$, and we use the shorthand $\ph$ for $\phab$.

\begin{proof}[Proof of \cref{lem:phi-prop}\ref{lem:us-0}]
  Since \eqref{eqn:cab-prop-c-01} says that $\cab(0) = -\min(a\bb, \aaa b)$ and $\cab(1) = \min(ab, \aaa\bb)$, we obtain $\ph(0) = a\bb + \aaa b - 2\min(a\bb, \aaa b) = \abs{a-b} \ge 0$ and $\ph(1) = a\bb + \aaa b + 2\min(ab, \aaa \bb) - 1 = -\abs{a + b - 1} \le 0$. Furthermore, $\ph(0) = 0$ if and only if $a = b$, and $\ph(1) = 0$ if and only if $a + b = 1$.
\end{proof}

\begin{proof}[Proof of \cref{lem:phi-prop}\ref{lem:us-c2}]
  Recall that $c = \cab(x)$ is the unique solution to
  \begin{subequations} \label{eqn:phi-prop-cab}
    \ifieee
    \begin{gather}
      \begin{multlined}
        m(c, x) := (a\bb + c)(\aaa b+c)(1-x)^2 \\
        - (ab-c)(\aaa \bb - c)(2x)^2 = 0,
      \end{multlined} \\
      -\min(a\bb, \aaa b) \le c \le \min(ab, \aaa\bb).
    \end{gather}
    \else
    \begin{gather}
      m(c, x) := (a\bb + c)(\aaa b+c)(1-x)^2 - (ab-c)(\aaa \bb - c)(2x)^2 = 0, \\
      -\min(a\bb, \aaa b) \le c \le \min(ab, \aaa\bb).
    \end{gather}
    \fi
  \end{subequations}
  Because $\ph(x) = a\bb + \aaa b + 2\cab(x) - x$, it suffices to prove $\cab$ is continuously differentiable on $[0,1]$.
  
  When $ab\aaa\bb = 0$, it is easy to check that $c = 0$ is a solution to \eqref{eqn:phi-prop-cab}, hence $\cab(x) = 0$ and so
  \begin{equation} \label{eqn:ph-formula-special}
    \ph(x) = a\bb + \aaa b - x
  \end{equation}
  is clearly continuously differentiable on $[0,1]$.
  
  Hereafter, we consider the case where $(a, b) \in (0,1)^2$. One can compute
  \[
    \fpp{m}{c} = (a\bb + \aaa b + 2c)(1-x)^2 + (ab + \aaa\bb - 2c)(2x)^2.
  \]
  When $0 < x < 1$, \cref{lem:diff}\ref{item:diff-cab-diff-3} says that $-\min(a\bb, \aaa b) < \cab(x) < \min(ab,\aaa\bb)$, and so $\partial m/\partial c(\cab(x), x) > 0$. Moreover $\partial m/\partial c(\cab(0), 0) = \partial m/\partial c(-\min(a\bb, \aaa b), 0) = \abs{a-b}$, which is positive if $a \neq b$, and and $\partial m/\partial c(\cab(1), 1) = \partial m/\partial c(\min(ab, \aaa \bb), 1) = \abs{1-a-b}$, which is positive if $a + b \neq 1$. The implicit function theorem then implies that
  \begin{enumerate}[label=(\alph*)]
    \item if $a \neq b$ then $\cab$ is continuously differentiable on $[0,1)$; and
    \item if $a + b \neq 1$ then $\cab$ is continuously differentiable on $(0,1]$.
  \end{enumerate}

  It suffices to show that if $a = b$ then $\cab(x)$ is continuously differentiable about $x = 0$; and if $a + b = 1$ then $\cab(x)$ is continuously differentiable about $x = 1$. We break the rest of the argument into two cases.

  \ifieee
    \medskip
    \textit{Case 1:} $a = b$.
  \else
    \paragraph{Case 1: $a = b$.}
  \fi
  Observe that $c = c_{a,b}(x)$ is also the unique solution to
  \ifieee
  \begin{gather*}
    m_0(c, x) := (a\aaa + c)(1-x) - \sqrt{(a^2-c)(\aaa^2-c)}(2x) = 0, \\
    -a\aaa \le c \le \min(a^2, \aaa^2),
  \end{gather*}
  \else
  \[
    m_0(c, x) := (a\aaa + c)(1-x) - \sqrt{(a^2-c)(\aaa^2-c)}(2x) = 0, \quad -a\aaa \le c \le \min(a^2, \aaa^2),
  \]
  \fi
  whose partial derivative with respect to $c$ is
  \[
    \fpp{m_0}{c} = (1-x) + \frac{a^2 + \aaa^2 - 2c}{\sqrt{(a^2-c)(\aaa^2-c)}}x.
  \]
  When $0 < x < 1$, $\partial m_0/\partial c(\cab(x), x) > 0$ because $-a\aaa < \cab(x) < \min(a^2,\aaa^2)$ (via \cref{lem:diff}\ref{item:diff-cab-diff-3}). Moreover $\partial m_0/\partial c(\cab(0), 0) = \partial m_0/\partial c(-a\aaa, 0) = 1$. The implicit function theorem then implies that $\cab$ is continuously differentiable on $[0,1)$.

  \ifieee
    \medskip
    \textit{Case 2:} $a + b = 1$.
  \else
    \paragraph{Case 2: $a + b = 1$.}
  \fi
  Observe that $c = c_{a,b}(x)$ is also the unique solution to
  \ifieee
  \begin{gather*}
    m_1(c, x) := \sqrt{(a^2+c)(\aaa^2 + c)}(1-x) - (a\aaa - c)(2x), \\
    -\min(a^2, \aaa^2) \le c \le a\aaa,
  \end{gather*}
  \else
  \[
    m_1(c, x) := \sqrt{(a^2+c)(\aaa^2 + c)}(1-x) - (a\aaa - c)(2x), \quad -\min(a^2, \aaa^2) \le c \le a\aaa,
  \]
  \fi
  whose partial derivative with respect to $c$ is
  \[
    \fpp{m_1}{c} = \frac{a^2 + \aaa^2 + 2c}{2\sqrt{(a^2+c)(\aaa^2+c)}}(1-x) + 2x.
  \]
  When $0 < x < 1$, $\partial m_1/\partial c(\cab(x), x) > 0$ because $-a\aaa < \cab(x) < \min(a^2,\aaa^2)$ (via \cref{lem:diff}\ref{item:diff-cab-diff-3}). Moreover $\partial m_1/\partial c(\cab(1), 1) = \partial m_1/\partial c(a\aaa, 0) = 2$. The implicit function theorem then implies that $\cab$ is continuously differentiable on $(0,1]$.
\end{proof}

\begin{proof}[Proof of \cref{lem:phi-prop}\ref{lem:us-a}]
  In view of \eqref{eqn:ph-formula-special}, $\ph(x) = -x$ when $a = b \in \set{0,1}$. Hereafter we consider the case where $a = b \in (0,1)$. 

  When $x \neq 1/3$, $m(c, x)$ can be rewritten as a quadratic function in $c$:
  \ifieee
  \begin{multline*}
    m(c, x) = \left((1-x)^2 - (2x)^2\right)c^2 + ((ab + \aaa\bb)(2x)^2 \\
    + (a\bb + \aaa b)(1-x)^2)c + ab\aaa\bb\left((1-x)^2 - (2x)^2\right),
  \end{multline*}
  \else
  \[
    m(c, x) = \left((1-x)^2 - (2x)^2\right)c^2 + ((ab + \aaa\bb)(2x)^2 + (a\bb + \aaa b)(1-x)^2)c + ab\aaa\bb\left((1-x)^2 - (2x)^2\right),
  \]
  \fi
  whose solutions can be routinely computed as follows:
  \begin{equation} \label{eqn:plus-minus}
    c = \tfrac12\left(- (a\bb + \aaa b) + u \pm \sqrt{v}\right),
  \end{equation}
  where
  \[
    u = \frac{-4x^2}{(1+x)(1-3x)}, \quad v = u^2 - 2(a\bb + \aaa b)u + (a-b)^2.
  \]
  When $0 \le x < 1/3$, since $u \le 0$ and $c \ge -\min(a\bb, \aaa b) \ge -(a\bb + \aaa b)/2$, we choose the plus sign in \eqref{eqn:plus-minus}. Therefore for $0 \le x < 1/3$,
  \begin{equation} \label{eqn:phi-formula-zero-to-third}
    \ph(x) = u + \sqrt{v} - x.
  \end{equation}
  We compute the Maclaurin series of $u$ and $v$ as follows:
  \ifieee
  \begin{align*}
    u & = -4x^2 - 8x^3 + o(x^3), \\
    v & = u^2 - 4a\aaa u = 16a\aaa x^2 + 32a\aaa x^3 + o(x^3).
  \end{align*}
  \else
  \[
    u = -4x^2 - 8x^3 + o(x^3), \quad v = u^2 - 4a\aaa u = 16a\aaa x^2 + 32a\aaa x^3 + o(x^3).
  \]
  \fi
  As $a \in (0,1)$, we obtain the Maclaurin series of $\sqrt{v}$ and $\ph(x)$:
  \ifieee
    \begin{align*}
      \sqrt{v} &= 4\sqrt{a\aaa}x+4\sqrt{a\aaa}x^2+o(x^2), \\
      \ph(x) &= (4\sqrt{a\aaa}-1)x + (4\sqrt{a\aaa}-4)x^2 + o(x^2). 
    \end{align*}
  \else
  \begin{equation*}
    \sqrt{v} = 4\sqrt{a\aaa}x+4\sqrt{a\aaa}x^2+o(x^2), \quad \ph(x) = (4\sqrt{a\aaa}-1)x + (4\sqrt{a\aaa}-4)x^2 + o(x^2). \qedhere
  \end{equation*}
  \fi
\end{proof}

\begin{proof}[Proof of \cref{lem:phi-prop}\ref{lem:us-b}]
  Suppose $(a,b) \not\in \oooo$. We introduce
  \[
    s := a\bb + \aaa b, \quad t := \abs{a-b}.
  \]
  One can check that for $(a,b) \in [0,1]^2$,
  \[
    2(a\bb + \aaa b) = 1+\abs{a-b}^2-\abs{a+b-1}^2, \quad \abs{a+b-1} \le 1 - \abs{a-b},
  \]
  which imply that $2s \le 1+t^2$, and $2s \ge 1 + t^2 - (1-t)^2 = 2t$. Thus we can estimate $t$ and $s$ as follows:
  \[
    0 \le t \le 1, \quad t \le s \le \tfrac12(1+t^2).
  \]
  It is also easy to see that $(a,b) \not\in \oooo$ is equivalent to
  \[
    (s,t) \neq (0,0).
  \]

  We consider the solutions of the equation
  \[
    \ph(x) = x\ph'(x), \quad x \in (0,1],
  \]
  and we split the rest of the proof into three cases.

  \ifieee
    \medskip
    \textit{Case 1:} $x = 1/3$.
  \else
    \paragraph{Case 1: $x = 1/3$.}
  \fi
  It is easy to check that $c = 0$ is a solution to \eqref{eqn:phi-prop-cab}. Thus $\cab(1/3) = 0$, and so
  \[
    \ph(1/3) = a\bb + \aaa b - 1/3 = s - 1/3.
  \]
  One can compute
  \[
    \fpp{m}{x} = -2(a\bb + c)(\aaa b + c)(1-x) - 4(ab-c)(\aaa\bb-c)(2x).
  \]
  The implicit function theorem then implies that
  \[
    \cab'(1/3) = -\fpp{m}{x}(0,1/3)\bigg/\fpp{m}{c}(0, 1/3) = 9ab\aaa\bb.
  \]
  One can check that $4ab\aaa\bb = (a\bb + \aaa b)^2 - \abs{a-b}^2$. Thus $\cab'(1/3) = 9(s^2-t^2)/4$, and
  \[
    \ph'(1/3) = 2c'(1/3) - 1 = 9(s^2-t^2)/4 - 1.
  \]
  The following claim immediately shows that $\ph(x) > x\ph'(x)$ at $x = 1/3$.
  \setcounter{claim}{0}
  \begin{claim} \label{claim:third}
    For every $s, t$ such that $0 \le t \le 1$, $t \le s \le (1+t^2)/2$ and $(s,t) \neq (0,0)$,
    \[
      q := s - 3(s^2 - t^2)/4 > 0.
    \]
  \end{claim}
  \begin{claimproof}[Proof of \cref{claim:third}]
    When $t = 0$, $q = s(1-3s/4)$, which clearly is positive for $0 < s \le 1/2$. Hereafter we deal with the case where $t \in (0,1]$. Since $q$ is a quadratic polynomial of $s$ with a negative leading coefficient, it suffices to check $q > 0$ for $s \in \set{t, (1+t^2)/2}$. When $s = t$, $q = t$ which is positive; when $s = (1+t^2)/2$, $q = (1+3t^2)(5-t^2)/16$ which is also positive.
  \end{claimproof}
  
  \ifieee
    \medskip
    \textit{Case 2:} $x \in (0,1/3)$.
  \else
    \paragraph{Case 2: $x \in (0,1/3)$.}
  \fi
  Recall from \eqref{eqn:phi-formula-zero-to-third} that for $x \in (0,1/3)$
  \[
    \ph(x) = u + \sqrt{v} - x,
  \]
  where
  \begin{equation} \label{eqn:u-v}
    u = \frac{-4x^2}{(1+x)(1-3x)}, \quad v = u^2 - 2su + t^2.
  \end{equation}
  We can compute
  \begin{equation} \label{eqn:u-v-prime}
    u' = \frac{-8x(1-x)}{(1+x)^2(1-3x)^2}, \quad v' = 2uu' - 2su'.
  \end{equation}

  We can then rewrite the equation $\ph(x) = x\ph'(x)$ as
  \[
    u + \sqrt{v}-x = x(u' + (\sqrt{v})' -1).
  \]
  After multiplying both sides by $2\sqrt{v}$, using the fact that $2\sqrt{v}(\sqrt{v})' = v'$, we can rearrange the above equation to
  \[
    2\left(u-xu'\right)\sqrt{v} = -2v + xv'.
  \]
  After completely replacing $u, u', v, v'$ according to \eqref{eqn:u-v} and \eqref{eqn:u-v-prime} except for $\sqrt{v}$, we obtain after simplification that
  \ifieee
  \begin{multline*}
    \frac{8x^2(1+3x^2)}{(1+x)^2(1-3x)^2}\sqrt{v} \\
    = \frac{32x^4(1+3x^2)}{(1+x)^3(1-3x)^3}+\frac{16sx^3(1+3x)}{(1+x)^2(1-3x)^2}-2t^2,
  \end{multline*}
  \else
  \[
    \frac{8x^2(1+3x^2)}{(1+x)^2(1-3x)^2}\sqrt{v} = \frac{32x^4(1+3x^2)}{(1+x)^3(1-3x)^3}+\frac{16sx^3(1+3x)}{(1+x)^2(1-3x)^2}-2t^2,
  \]
  \fi
  which after multiplying by $(1+x)^2(1-3x)^2/(8x^2)$ gives
  \ifieee
  \begin{align*}
    \left(1+3x^2\right)\sqrt{v} & \stackrel{\phantom{\eqref{eqn:u-v}}}{=} \frac{4x^2\left(1+3x^2\right)}{(1+x)(1-3x)} + 2sx(1+3x) \\
    & \hspace{1.5em} - \frac{t^2(1+x)^2(1-3x)^2}{4x^2} \\
    & \stackrel{\eqref{eqn:u-v}}{=} -(1+3x^2)u + 2sx(1+3x) - \frac{4t^2x^2}{u^2}.
  \end{align*}
  \else
  \begin{align*}
    \left(1+3x^2\right)\sqrt{v} & \stackrel{\phantom{\eqref{eqn:u-v}}}{=} \frac{4x^2\left(1+3x^2\right)}{(1+x)(1-3x)} + 2sx(1+3x) - \frac{t^2(1+x)^2(1-3x)^2}{4x^2} \\
    & \stackrel{\eqref{eqn:u-v}}{=} -(1+3x^2)u + 2sx(1+3x) - \frac{4t^2x^2}{u^2}.
  \end{align*}
  \fi
  The following claim immediately shows that $\ph(x) \neq x\ph'(x)$ for $x \in (0,1/3)$.
  \begin{claim} \label{claim:le-third}
    For every $x \in (0,1/3)$ and $s, t$ such that $0 \le t \le 1$, $t \le s \le (1+t^2)/2$ and $(s,t)\neq(0,0)$,
    \[
      \left(1+3x^2\right)\sqrt{v} > -(1+3x^2)u + 2sx(1+3x).
    \]
  \end{claim}
  \begin{claimproof}[Proof of \cref{claim:le-third}]
    Since the left hand side of the inequality is non-negative, squaring both sides strengthens the inequality, which after substituting $v$ according to \eqref{eqn:u-v} gives that
    \[
      \left(1+3x^2\right)(u^2-2su + t^2) > \left(-(1+3x^2)u + 2sx(1+3x)\right)^2.
    \]
    The difference between the two sides of the last inequality, after substituting $u$ according to \eqref{eqn:u-v}, equals
    \[
      q := 4sx^2\left(2+6x^2-s(1+3x)^2\right) + \left(1+3x^2\right)t^2.
    \]
    When $t = 0$, $q = 2sx^2(3(1-x)^2+(1-2s)(1+3x)^2)$ which is clearly positive for $0 < s \le 1/2$ and $x \in (0,1/3)$. Hereafter we deal with the case where $t \in (0,1]$. Since $q$ is a quadratic polynomial of $s$ with a negative leading coefficient, it suffices to check $q > 0$ for $s \in \set{t,(1+t^2)/2}$. When $s = t$,
    \[
      q = t\left(\left(8x^2+ 24x^4\right)(1-t) + \left(1+7x^2-24x^3-12x^4\right)t\right).
    \]
    One can check that both $8x^2 + 24x^4$ and $1+7x^2-24x^3-12x^4$ are positive for $x \in (0,1/3)$. Thus $q > 0$ for $x \in (0,1/3)$. When $s = (1+t^2)/2$,
    \ifieee
    \begin{multline*}
      q = (1+t^2)x^2\left(4+12x^2-(1+t^2)(1+3x)^2\right) \\
      + \left(1+3x^2\right)t^2,
    \end{multline*}
    \else
    \[
      q = (1+t^2)x^2\left(4+12x^2-(1+t^2)(1+3x)^2\right) + \left(1+3x^2\right)t^2,
    \]
    \fi
    which can be seen as a quadratic polynomial of $t^2$ with a negative leading coefficient. It suffices to check $q > 0$ for $t^2 = 0$ and $t^2 = 1$. Note that
    \[
      q = \begin{cases}
        3x^2(1-x)^2 & \text{when }t = 0; \\
        1+6x^2-6x^3+3x^4 & \text{when }t = 1.
      \end{cases}
    \]
    In either case, one can check that $q > 0$ for $x \in (0,1/3)$.
  \end{claimproof}
  
  \ifieee
    \medskip
    \textit{Case 3:} $x \in (1/3, 1]$.
  \else
    \paragraph{Case 3: $x \in (1/3, 1]$.}
  \fi
  Recall from \eqref{eqn:plus-minus} that for $x \neq 1/3$,
  \[
    \cab(x) = \tfrac12(-(a\bb + \aaa b) + u \pm \sqrt{v}) = \tfrac12(ab + \aaa\bb - 1 + u \pm \sqrt{v}).
  \]
  When $1/3 < x \le 1$, since $u > 1$ and $\cab(x) \le \min(ab, \aaa\bb) \le (ab+\aaa\bb)/2$, we choose the minus sign in the above formula. Therefore for $x \in (1/3, 1]$
  \[
    \ph(x) = u - \sqrt{v} - x.
  \]
  
  We can then rewrite the equation $\ph(x) = x \ph'(x)$ as
  \[
    u - \sqrt{v}-x = x(u' - (\sqrt{v})' - 1).
  \]
  After multiplying both sides by $2\sqrt{v}$, using the fact that $2\sqrt{v}(\sqrt{v})' = v'$, we can rearrange the above equation to
  \[
    2\left(u-xu'\right)\sqrt{v} = 2v - xv',
  \]
  which, by comparing with the corresponding computation in Case~2, is equivalent to
  \[
    \left(1+3x^2\right)\sqrt{v} = (1+3x^2)u - 2sx(1+3x) + \frac{4t^2x^2}{u^2}.
  \]
  
  The following claim immediately implies that $\ph(x) = x\ph'(x)$ could possibly have a solution in $(1/3, 1]$ only when $t = 1$ or $(s,t) = (1/2,0)$, which corresponds to $(a,b) \in \set{(1/2, 1/2),(0,1),(1,0)}$. However one can compute that $\ph_{1/2,1/2}(x) = x(1-x)/(1+x)$ and $\ph_{0,1}(x) = \ph_{1,0}(x) = 1-x$, and one can check directly that in neither case $\ph = x\ph'$ has a solution in $(0,1]$.

  \begin{claim} \label{claim:ge-third}
    For every $x \in (1/3,1]$, $s, t$ such that $0 \le t < 1$, $t \le s \le (1+t^2)/2$ and $(s,t) \not\in \set{(0,0), (1/2, 0)}$,
    \[
      \left(1+3x^2\right)^2v > \left((1+3x^2)u - 2sx(1+3x) + \frac{4t^2x^2}{u^2}\right)^2.
    \]
  \end{claim}
  \begin{claimproof}[Proof of \cref{claim:ge-third}]
    After substituting $v$ according to \eqref{eqn:u-v}, the difference between the two sides of the last inequality equals
    \ifieee
    \begin{multline*}
      q := \left(1+3x^2\right)^2(u^2-2su+t^2) \\ - \left((1+3x^2)u - 2sx(1+3x) + \frac{4t^2x^2}{u^2}\right)^2.
    \end{multline*}
    \else
    \[
      q := \left(1+3x^2\right)^2(u^2-2su+t^2) - \left((1+3x^2)u - 2sx(1+3x) + \frac{4t^2x^2}{u^2}\right)^2.
    \]
    \fi
    Since $q$ is a quadratic polynomial of $s$ with a negative leading coefficient, it suffices to check $q \ge 0$ when $t = 0$ and $s \in \set{0,1/2}$, and $q > 0$ when $t \in (0,1)$ and $s \in \set{t, (1+t^2)/2}$.

    When $s = t$, we have
    \ifieee
    \begin{align*}
      q & = \left((1+3x^2)(u-t)\right)^2 \\
      & \phantom{{}={}} - \left((1+3x^2)u - 2tx(1+3x) + \frac{4t^2x^2}{u^2}\right)^2 \\
      & = \left(t(3x-1)(x+1) - \frac{4t^2x^2}{u^2}\right) \\
      & \phantom{{}={}} \cdot \left(2(1+3x^2)u - t(9x^2+2x)+\frac{4t^2x^2}{u^2}\right),
    \end{align*}
    \else
    \begin{align*}
      q & = \left((1+3x^2)(u-t)\right)^2 - \left((1+3x^2)u - 2tx(1+3x) + \frac{4t^2x^2}{u^2}\right)^2 \\
      & = \left(t(3x-1)(x+1) - \frac{4t^2x^2}{u^2}\right)\left(2(1+3x^2)u - t(9x^2+2x)+\frac{4t^2x^2}{u^2}\right),
    \end{align*}
    \fi
    which after substituting $u$ according to \eqref{eqn:u-v} can be further factorized into
    \ifieee
    \begin{multline*}
      q = t\biggl(t(1-x)^2+4(1-t)x^2\biggr)^2\biggl(\left(8x^2+24x^4\right)(1-t)\\
      +(-1+4x+10x^2-12x^3+15x^4)t\biggr)\bigg/(16x^4).
    \end{multline*}
    \else
    \[
      q = t\biggl(t(1-x)^2+4(1-t)x^2\biggr)^2\biggl(\left(8x^2+24x^4\right)(1-t)+(-1+4x+10x^2-12x^3+15x^4)t\biggr)\bigg/(16x^4).
    \]
    \fi
    One can check that both $8x^2+24x^4$ and $-1+4x+10x^2-12x^3+15x^4$ are positive for $x \in (1/3, 1]$. Thus $q = 0$ when $t = 0$, and $q > 0$ when $t \in (0,1)$.

    When $s = (1+t^2)/2$, we have
    \ifieee
    \begin{multline*}
      q = \left(1+3x^2\right)^2(u^2-(1+t^2)u+t^2) \\ - \left((1+3x^2)u - (1+t^2)x(1+3x) + t^2\frac{4x^2}{u^2}\right)^2,
    \end{multline*}
    \else
    \[
      q = \left(1+3x^2\right)^2(u^2-(1+t^2)u+t^2) - \left((1+3x^2)u - (1+t^2)x(1+3x) + t^2\frac{4x^2}{u^2}\right)^2,
    \]
    \fi
    which can be seen as a quadratic polynomial of $t^2$. We compute $q$ for $t = 0$ and $t = 1$. Note that, after substituting $u$ according to \eqref{eqn:u-v},
    \ifieee
    \[
      q = \begin{dcases}
        48(1-x)^2x^6 & \text{when }t = 0; \\
        \begin{multlined}
          (1-x)^4(-1+4x+10x^2\\-12x^3+15x^4)
        \end{multlined} & \text{when }t = 1. \\
      \end{dcases}
    \]
    \else
    \[
      q = \begin{cases}
        48(1-x)^2x^6 & \text{when }t = 0; \\
        (1-x)^4(-1+4x+10x^2-12x^3+15x^4) & \text{when }t = 1. \\
      \end{cases}
    \]
    \fi
    One can check that $-1+4x+10x^2-12x^3+15x^4 > 0$ for $x \in (1/3,1]$. Thus in each case $q \ge 0$, which implies that $q \ge 0$ for $t = 0$ and $q > 0$ for $t \in (0,1)$.
  \end{claimproof}

  This finishes the proof of the three cases in \cref{lem:phi-prop}\ref{lem:us-b}.
\end{proof}

\ifieee
\IEEEtriggeratref{10}
\bibliographystyle{IEEEtran}
\else
\bibliographystyle{plain}
\fi
\bibliography{zero_error}

\ifieee
\begin{IEEEbiographynophoto}{Samuel H.~Florin} is an undergraduate student at Massachusetts Institute of Technology currently planning on studying mathematics and computer science. He plans to graduate in 2025. His interests currently include information theory, quantum computing, and artificial intelligence.
\end{IEEEbiographynophoto}
\begin{IEEEbiographynophoto}{Matthew H.~Ho} is currently an undergraduate student at Massachusetts Institute of Technology who is studying mathematics. He plans to graduate in 2025. His research instests lie in applied mathematics and theoretical computer science.
\end{IEEEbiographynophoto}
\begin{IEEEbiographynophoto}{Zilin Jiang} is an Assistant Professor at Arizona State University. He was a postdoctoral fellow at Technion--Israel Institute of Technology between 2016 and 2018, and he was then an applied mathematics instructor at Massachusetts Institute of Technology up until the end of 2020. He graduated from Peking University in 2011 with a B.Sc.\ degree in Mathematics, and he received his Ph.D.\ in Algorithm, Combinatorics and Optimization from Carnegie Mellon University in 2016.
\end{IEEEbiographynophoto}
\vfill
\fi

\end{document}